\documentclass[11pt]{article}
\usepackage{amsmath,amsthm,amssymb,amsfonts}
\usepackage{latexsym}
\usepackage{graphicx,psfrag,import}
\usepackage{fullpage}
\usepackage{framed}
\usepackage{verbatim}
\usepackage{color}
\usepackage{epsfig}
\usepackage{hyperref}
\usepackage{geometry}
\usepackage{mathtools}
\usepackage{float}
\usepackage{enumerate}
\usepackage{subcaption}
\usepackage{multicol}
\usepackage{a4wide}
\usepackage{booktabs}
\usepackage{enumitem}
\usepackage{lineno}
\usepackage{parcolumns}
\usepackage{thmtools}
\usepackage{authblk}
\usepackage{xr}
\usepackage{epstopdf}

\theoremstyle{plain}
\newtheorem{theorem}{Theorem}[section]

\theoremstyle{definition}
\newtheorem{definition}[theorem]{Definition}
\newtheorem{example}[theorem]{Example}

\newtheorem{observation}[theorem]{Observation}
\newtheorem{problem}[theorem]{Problem}

\newtheorem{claim}[theorem]{Claim}
\newtheorem{lemma}[theorem]{Lemma}
\newcommand{\Int}{\int\limits}
\newcommand\R{\mathbb{R}}

\newcommand\RS{\mathcal{RS}}
\newcommand\DS{\mathcal{DS}}

\DeclareMathOperator{\e}{\mathrm{e}}
\DeclareMathOperator{\E}{\mathbb{E}}

\newenvironment{bullets}
{\begin{list}
{\noindent\makebox[0mm][r]{$\bullet$}}
{\leftmargin=5.5ex \usecounter{enumi}
\topsep=1.5mm \itemsep=-.75ex}
}
{\end{list}}

\date{}

\begin{document}

\title{\textbf{Designing the Optimal Bit: Balancing Energetic Cost, Speed and Reliability}}
\author[1,2]{Abhishek Deshpande\thanks{Corresponding author: deshabhi123@gmail.com}}
\author[3]{Manoj Gopalkrishnan}
\author[4]{Thomas E. Ouldridge}
\author[1]{Nick S. Jones}
\affil[1]{Department of Mathematics, Imperial College London, London SW7 2AZ, United Kingdom}
\affil[2]{School of Technology and Computer Science, Tata Institute of Fundamental Research, Mumbai 400005, India}
\affil[3]{Department of Electrical Engineering, Indian Institute of Technology Bombay, Mumbai 400076, India}
\affil[4]{Department of Bioengineering, Imperial College London, London SW7 2AZ, United Kingdom}
\maketitle

\begin{abstract}
We consider the challenge of operating a reliable bit that can be rapidly erased. We find that both erasing and reliability times are non-monotonic in the underlying friction, leading to a trade-off between erasing speed and bit reliability. Fast erasure is possible at the expense of low reliability at moderate friction, and high reliability comes at the expense of slow erasure in the underdamped and overdamped limits. Within a given class of bit parameters and control strategies, we define ``optimal" designs of bits that meet the desired reliability and erasing time requirements with the lowest operational work cost. We find that optimal designs always saturate the bound on the erasing time requirement, but can exceed the required reliability time if critically damped. The non-trivial geometry of the reliability and erasing time-scales allows us to exclude large regions of parameter space as suboptimal. We find that optimal designs are either critically damped or close to critical damping under the erasing procedure.  
\end{abstract}


\section{Introduction}

Certain information processing operations such as erasing a bit, or copying the state of one bit into another previously randomised bit have fundamental lower bounds on work input~\cite{szilard1964decrease,landauer1961irreversibility,bennett1982thermodynamics,bennett1988notes,bennett2003notes}. These lower bounds such as the famous $k_BT \ln2$ minimal cost for erasing arise due to equilibrium thermodynamics: there is a need to compensate for any entropy reduction in the information-carrying system with an entropy increase elsewhere. Practical devices, however, do not approach these bounds~\cite{Frank2002,Pop2010} and insights gained from thinking about the lower bound have not yet translated into more energy-efficient technology. A partial explanation is that  man-made devices and biological cells need to operate on fast time-scales and hence cannot involve the quasistatic manipulations necessary to reach lower bounds~ \cite{zulkowski2014optimal,Ouldridge2015}. An alternative suggestion from von Neumann is that the need to store information for long periods of time (\textbf{reliability}) leads to high-cost architectures~\cite{von1966theory}. We explore the interplay between reliability, speed and the energetic cost of bit operation. Equilibrium thermodynamic bounds such as the Landauer limit cannot account for these inherently kinetic phenomena.

This general question of how to design fast, cheap and reliable bits has obvious technological relevance to the optimal design of low power computational devices~\cite{sarpeshkar2010ultra,sarpeshkar1998analog,rapoport2012glucose} Additionally, since the discovery of the structure of DNA and the central dogma of molecular biology, it has become well accepted that information processing is at the heart of many natural phenomena. Many authors have explored information processing in biological systems, both to understand natural examples and design synthetic analogs~ \cite{bennett1982thermodynamics,Tu2008,Lan2012,govern2014,Ouldridge2015,Ouldridge2017,Bo2014,Barato2014}. The question of the interplay between reliability, speed and cost are also relevant here, although under-explored.

In this paper, we explore the challenge of building fast, cheap and reliable bits, and provide a framework for it's analysis in terms of reliability and erasure time-scales. We also take the first steps towards exploring the physics of the optimal design problem by considering a simple model: a particle in a 1-D potential, which is a quartic double-well potential in the device's ``resting'' state. We require that the bit be reliable, so that a particle equilibrated in either well stays in that well for a specified long time on average. Simultaneously, we require the implementation of an ``erase'' or ``reset'' operation using an external control, so that erasure is completed within a specified short amount of time. Our principal question is to find values for the design parameters which consist of the height of the double well, the friction coefficient, and the control parameters to guarantee these requirements without expending more energy than required. Our main contribution is an exploration of this design space, which demonstrates the previously under-appreciated role of friction. In particular, we identify a ``Goldilocks zone'' where the friction coefficient takes moderate values. This is somewhat counter-intuitive because historically friction has been viewed as a nuisance to computing, to be sent as low as possible~\cite{anacker1980josephson,buttiker1983thermal,klein1982thermal,likharev1982classical}.

In Section~\ref{sec:description}, we describe the model which will provide intuition for our work. We formalize the time-scale over which the bit stores information through the notion of \textbf{reliability time}. In Section~\ref{sec:control_strategy}, we describe one simple family of control protocols for resetting a bit. We calculate the work done in erasing a bit for this form of control. We will use this particular control protocol to illustrate our subsequent ideas. In Section~\ref{subsubsec:erasure}, we introduce the notion of \textbf{erasing time}. In Section~\ref{sec:times}, we consolidate from the literature the analytical forms and approximations for our two time-scales of interest, and confirm them with numerical simulations. We find that both the reliability and erasing time-scales are non-monotonic, roughly U-shaped functions of the friction coefficient. It follows that high reliability is obtained by setting the friction to a low or high value, whereas a low erasing time is favoured by an intermediate value of friction, implying a conflict between the two time scales for a given class of protocols. In Section~\ref{geometry}, we investigate how this conflict feeds into  the geometry of \textbf{optimal bits}: bits that fulfil the desired reliability and erasing time requirements with the minimum energy cost. We find and partially characterize a ``Goldilocks zone'' in design space where optimal bits reside. In Section~\ref{conclusion}, we discuss the robustness of our results when more freedom is allowed in the choice of design parameters, and the control protocol.

\section{The Double-Well Bit}\label{sec:description}
We will represent a device to store one bit of information by a particle in a symmetric bistable potential $U_{A,B}\left(x\right) = A\left(\frac{x^2}{B^2} -1\right)^2$, where $A$ is the height of the well and $\pm B$ are the coordinates of the minima of the right and left wells. We will refer to the device as a whole by ``a bit". The device reports ``0'' when the particle is in the left well, i.e., $x<0$ and reports ``1'' otherwise (Figure~\ref{subfig:doublewell}). 

\begin{figure}
\centering
\begin{minipage}{0.48\textwidth}
  \centering
 \includegraphics[scale=0.3]{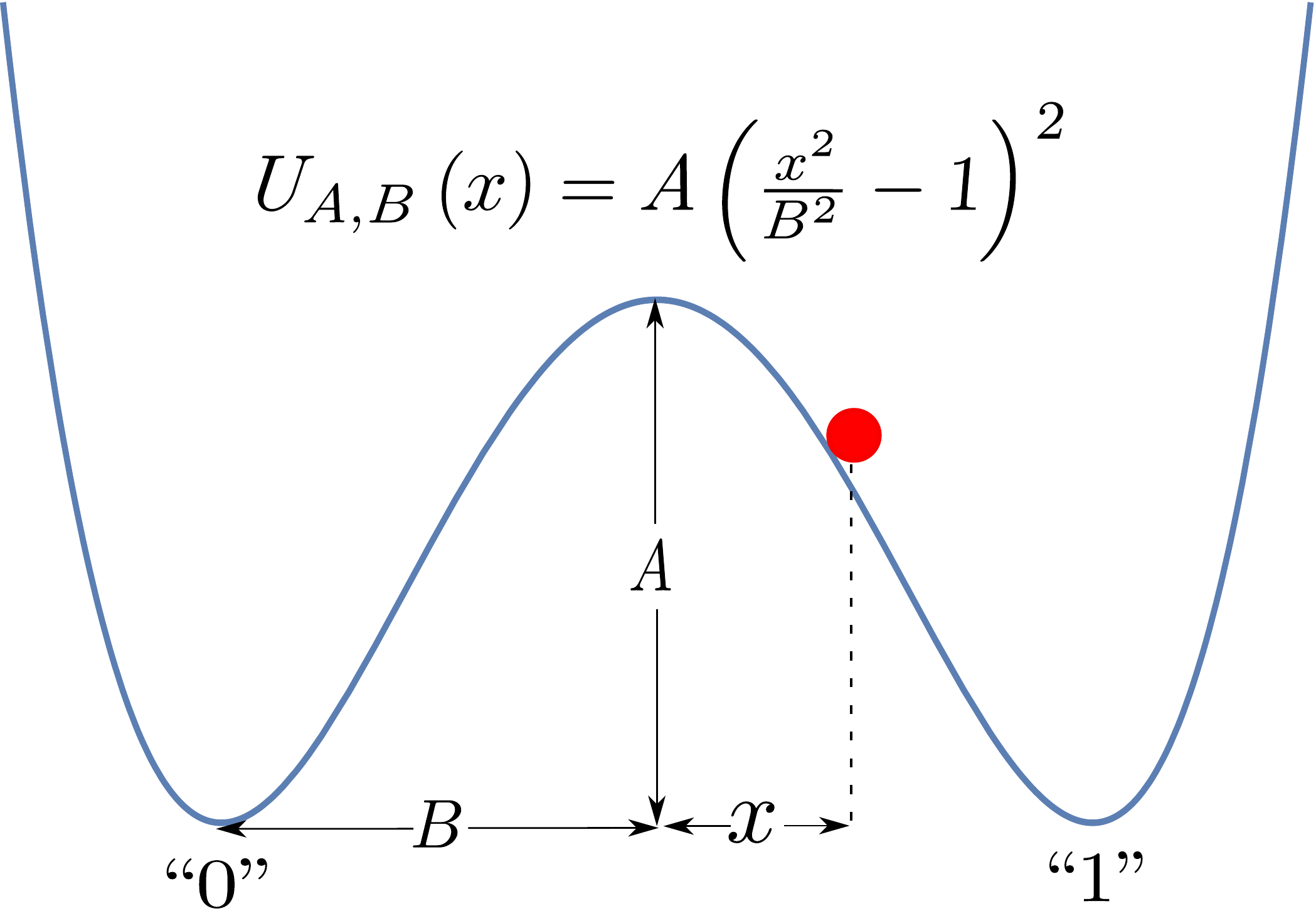}
  \subcaption{}\label{subfig:doublewell}
\end{minipage}
\begin{minipage}{0.48\textwidth}
  \centering
  \includegraphics[scale=0.3]{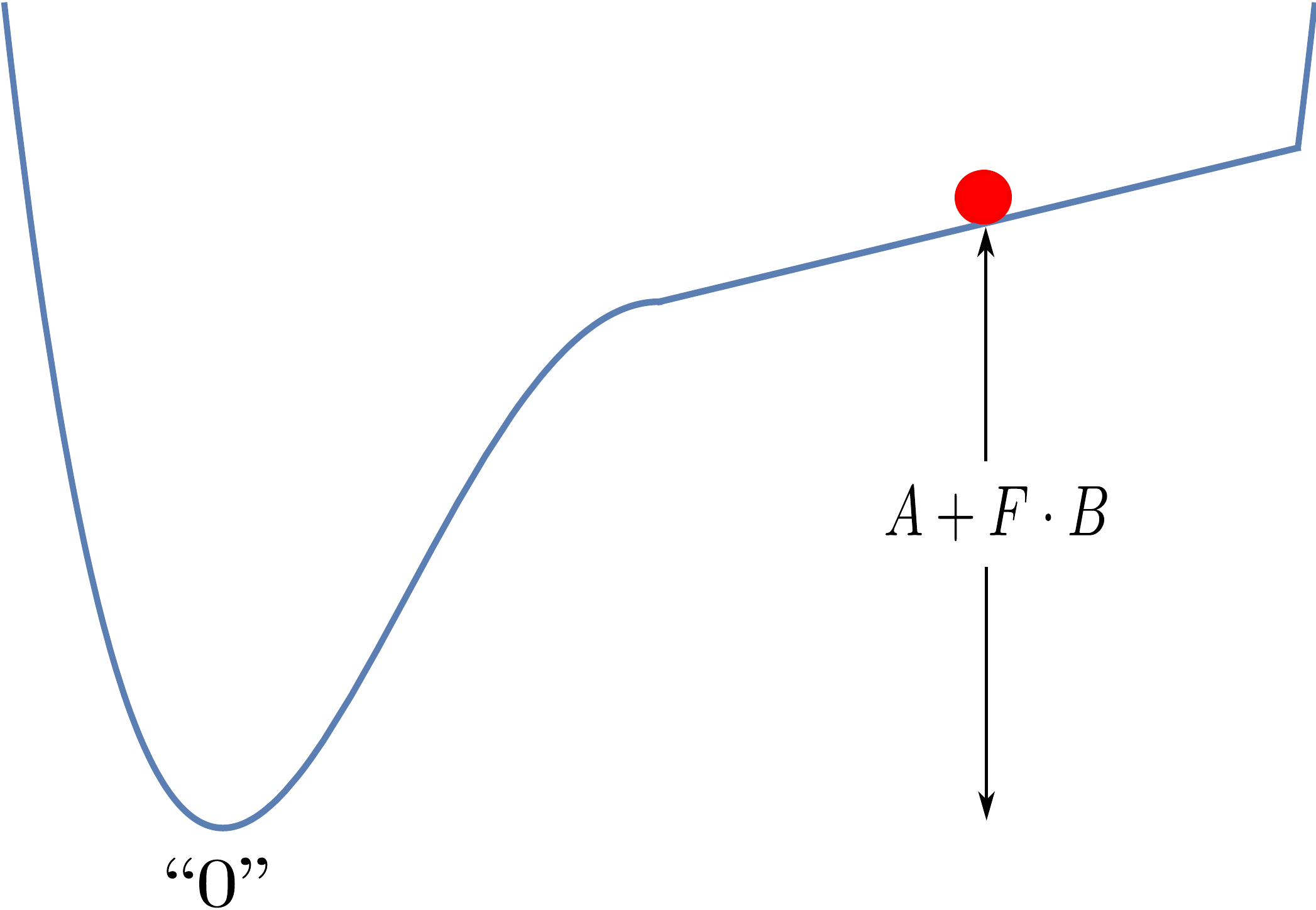}
  \subcaption{}\label{subfig:control}
 \end{minipage}
 \caption{A bit as represented by a particle in a 1-D potential. Figure~\ref{subfig:doublewell}: the bit in its resting state, with a barrier of height ``A'' separating particle locations that correspond to bit values of 0 or 1. Figure~\ref{subfig:control}: a control potential as in Example~\ref{ex:control_potential} is applied to erase the stored data.}\label{fig:erasing} 
\end{figure}

The dynamics of the particle is described by the Langevin equation.
\begin{eqnarray}\label{eq:langevin}
\begin{aligned}
m\,dx = {}& p\,dt \\
\,dp = {}& -\gamma p\,dt - \partial_xU_{A,B}\left(x\right)\,dt + \sqrt{2m\gamma k_B T}\,dW
\end{aligned}
\end{eqnarray}
where $m$ is the mass of the particle, $x$ is position, $p$ is momentum, $\gamma$ is the friction coefficient of the medium, $U_{A,B}(x)$ is the potential, $k_B$ is Boltzmann's constant, and $T$ is the temperature of the heat bath. The term $\sqrt{2m\gamma k_B T}\,dW$ represents the effect of noise from the surroundings. The Langevin equation is a stochastic differential equation, to be mathematically interpreted as a Stratonovich integral. For our case both the Ito and Stratonovich interpretations coincide~\cite[pp.~109]{kloeden1992higher} since the noise coefficient $\sqrt{2m\gamma k_B T}$ does not depend upon $p$.

From~\cite[pp.~182]{pavliotis2014stochastic}, the generator for the Langevin equation~\ref{eq:langevin} is
\begin{align}\label{eq:gen}
\mathcal{L} = \frac{p}{m}\partial_x - \left(\partial_xU_{A,B}(x)\right)\partial_p + \gamma\left(-p\partial_p + k_BT\partial^2_p\right)
\end{align}
The \textbf{Hamiltonian} of the system is $H\left(x,p\right) = U_{A,B}\left(x\right) + \frac{p^2}{2m}$. The \textbf{Gibbs distribution}
\begin{align}\label{eq:gibbs}
\pi(x,p) = \frac{\e^{-H(x,p)/k_BT}}{\int_{-\infty}^\infty \int_{-\infty}^\infty \e^{-H(x',p')/k_BT}dx'\,dp'}
\end{align}
is approached as the system relaxes to equilibrium. Convergence to $\pi(x,p)$ happens exponentially fast at a rate given by the first non-zero eigenvalue of the generator $\mathcal{L}$~\cite{mattingly2002geometric}. 

\subsection{Reliability}\label{subsec:reliability}

A device to store information should be able to store it with high fidelity for a specified long period of time. We introduce the \textbf{reliability time} to represent the time-scale over which our device can store data. Specifically, we define the reliability time $\tau_r$ as the expected first passage time for the particle to cross the barrier of the resting-state potential of the bit, given the Gibbs distribution $\pi(x,p)$ (Equation~\ref{eq:gibbs}) as the initial distribution. That is,
\begin{align}\label{eq:reliability_time} 
\tau_r := \E \left[ \inf \{ t\geq 0 \mid x(t) = 0 \} \right]
\end{align}
where the expectation is over trajectories $\left(x\left(t\right),p\left(t\right)\right)$ distributed as specified by Equation~\ref{eq:langevin} from the initial condition $\left(x\left(0\right),p\left(0\right)\right)\sim_\text{law}\pi\left(x,p\right)$. Note that $\tau_r$ is also the first passage time to cross the barrier for a bit prepared with a Gibbs distribution, but confined to either the left-hand well $\pi_0(x,p)$ or  right-hand  well $\pi_1(x,p)$.
\begin{equation}
    \pi_0\left(x,p\right) =
    \begin{cases*}
      2\pi\left(x,p\right) & if x $<$ 0 \\
      0        & otherwise
    \end{cases*}, 
\hspace{5mm}
    \pi_1\left(x,p\right) =
    \begin{cases*}
      2\pi(x,p) &if x $>$ 0 \\
      0        &otherwise
    \end{cases*}.
\end{equation}
Intuitively, once a typical particle has had enough time to reach the top of the barrier, the data stored is no longer reliable.

\subsection{Setting information}
A device intended to store information must provide functionality to load, or set this information into the device. Setting information is a two-bit operation. A common use case is when a reference bit and the bit to be set are initially at some arbitrary values. We require that after the SET operation the reference bit is unchanged whereas the bit to be set now holds a copy of the reference bit. This is the operation that Szilard~\cite{szilard1964decrease} refers to as ``copying'' (in contrast, Landauer~\cite{landauer1961irreversibility,gopalkrishnan2013hot} chooses to reserve the word ``copying'' for the operation where the bit to be set is initially already known to be in the state ``0''). 

Note that in the operation of setting information, or copying in the sense of Szilard, initially the two bits are uncorrelated and unknown whereas after the operation they are still unknown but correlated. Thus implementing this operation requires decreasing the entropy of the system. Since it is easier to study a one-bit system rather than a two-bit system, we will investigate a one-bit proxy for the task of decreasing the entropy of the system, which is the task of \textbf{erasing} a bit.

Erasing involves taking a device whose initial state is maximally unknown into a known reference state, usually ``0.'' Somewhat counter intuitively, given the name, erasing increases the information we know about the system. What is erased is not information but randomness. It helps to keep in mind the example of erasing a blackboard where some random state with chalk marks is reset to the ``all clear'' state. 

\subsubsection{Erasing}\label{sec:control_strategy}
The example that follows describes a simple family of control potentials to implement the erasing operation for our device, which will form the basis of our analysis. One control potential from this family is illustrated in Figure~\ref{subfig:control}. We chose such a simple class of controls to make a full understanding feasible, setting a framework for analysing more complex protocols. We also note that arbitrary variation of a physical potential in reality is highly non-trivial; experimental studies in which complex time-dependent potentials have been applied in fact use highly dissipate mechanisms to generate "effective" potentials~\cite{berut2012experimental,precision_feedback}.

\begin{example}\label{ex:control_potential}
Our control potentials are described by a single parameter $F\in\mathbb{R}_{>0}$ as follows.
\begin{align}
V_F\left(x\right):=
     \begin{cases}
       A + F\cdot x  - U_{A,B}\left(x\right) &\text{if } x\geq 0\,\, {\rm and}\,\, A - U_{A,B}\left(x\right)+F\cdot x \geq 0,\\
       0 &\text{otherwise.} \\
     \end{cases}
\end{align}
The Langevin equation in the presence of control is
\begin{align}\label{eq:control}
\begin{split}
m\,dx = {}& p\,dt  \\
dp  = {} & - \gamma p\,dt -\partial_x U_{A,B}\left(x\right)\,dt -\partial_x V_F\left(x\right)\,dt + \sqrt{2m\gamma k_B T}\,dW
\end{split}
\end{align}
\end{example}

Note that the control potential, as defined, is not differentiable at the boundary of the region in which it is non-zero. In practice, we assume that $\partial_x V_F$ changes rapidly but continuously in a small vicinity around these points.

In this work, we will consider variation of $A$, $F$ and $\gamma$ at fixed $m$, $B$, and $T$. In this case, $m$ specifies the natural mass scale, $B$ the natural length scale and $k_BT$ the natural energy scale; the natural time scale is then $\sqrt {mB^2/k_BT}$. Henceforth, all numerical quantities will be reported using reduced units defined with respect to these natural scales, although $m$, $B$ and $k_BT$ will be retained within formulae.

\subsubsection{Operational view of Erasing}\label{subsubsec:erasure}

The speed of bit operations is of practical importance: a useful bit must be reliable on much larger time-scales than those required to set or switch it. The control is switched on at time $0$ and switched off at an appropriately-chosen time $\tau$. The time $\tau$ is chosen beforehand, and does not depend on details of individual trajectories -- a trajectory-dependent control would require measurement and feedback that itself would need accounting for~\cite{sagawa_ueda_feedback,general_jarzynski,ponmurugan,abreau,horowitz,lahiri}. We could declare erasing as completed and switch off the control as soon as a majority of trajectories are expected to be  in the left well. However, many of these ``erased'' bits would have high energies compared to typical bits drawn from the equilibrium distribution in the left well, $\pi_0\left(x,p\right)$. Thus they could rapidly return to the right well after a very short stay in the left well. So we insist on a more stringent condition. We require that the time $\tau$ should be large enough so that the majority of bits are in the target well, with an expected next passage time close to the reliability time.

One way to guarantee that the next passage time is high is by insisting on mixing, in the sense that the initial distribution $\pi(x,p)$ comes close to a distribution of particles thermalised in the left-hand well, $\pi_0(x,p)$. If this happens, we can guarantee that the expected next passage time will be equal to, or close to, the expected first passage time. However, we found this criterion too stringent for the following reason. At the end of the erasing protocol, it is not necessary that the distribution is close to $\pi_0(x,p)$ -- only that the particles tend to relax to this distribution much faster than they cross back into the right-hand well, and thus have barrier passage times representative of particles initialised with $\pi_0(x,p)$. Nonetheless, we show in Section~\ref{subsec:accuracy} of the Supplementary Information that using such a criterion preserves the qualitative features reported below (in particular, the scaling of erasure time with friction in the high and low friction limits).

Instead, we define an \textbf{erasure region} in well ``0'' as all points $(x,p)$ with total energy $H(x,p)\leq A - 3k_BT$ where $A$ is the barrier height. We look for the average first passage time to reach the erasure region for particles initiated in well ``1'' and take this quantity to be representative of the erasing time-scale. The choice of $3k_BT$ criterion is somewhat arbitrary, but has been used before by Vega et al.~\cite{vega2002mean} to study atom-surface diffusion. As we show in Section~{\ref{subsec:4kT_criterion} of the Supplementary Information, using $4k_BT$ makes no qualitative difference to our conclusions. This metric has the merit that it provides a clear computable criterion for erasing. Below, we demonstrate that particles within the $3k_BT$ erasure region do indeed have  expected next passage times close to the reliability time, as required. 

For a range of well parameters, we used the Langevin $A$ algorithm from~\cite{davidchack2009langevin} [refer to Section~\ref{timestep_validation} of the Supplementary Information for integrator set-up and validation], to estimate $\tau(x,p)$, the average  barrier crossing time for particles initialised 
at position $x$ with momentum $p$ in the left well, for a grid of points $(x,p)$. The average reliability time for a given well can be approximated in terms of $\tau(x,p)$  as follows:
\begin{align}
\tau_r\approx\frac{\displaystyle\sum_{x,p} \tau(x,p)\e^{-{H(x,p)/k_BT}}}{\displaystyle\sum_{x,p}\e^{-{H(x,p)/k_BT}}}.
\end{align} 
The deviation $\delta(x,p) := |1 - \frac{\tau(x,p)}{\tau_r}|$ for every point $\left(x,p\right)$ in the grid is plotted in Figure~\ref{fig:heatmaps}, for a range of friction parameters at well height $A=7$. It is clear that, for all values of friction, the points with total energy $H(x,p)\leq A - 3k_BT$ have reliability times close to $\tau_r$. The same is true of other well heights $A$. This is because such particles typically undergo thermal mixing before they can escape the well. Once mixed, their next escape over the barrier will be on a time-scale of the order of $\tau_r$.

\begin{figure}[t]
\begin{minipage}{0.5\textwidth}
  \includegraphics[scale=0.15]{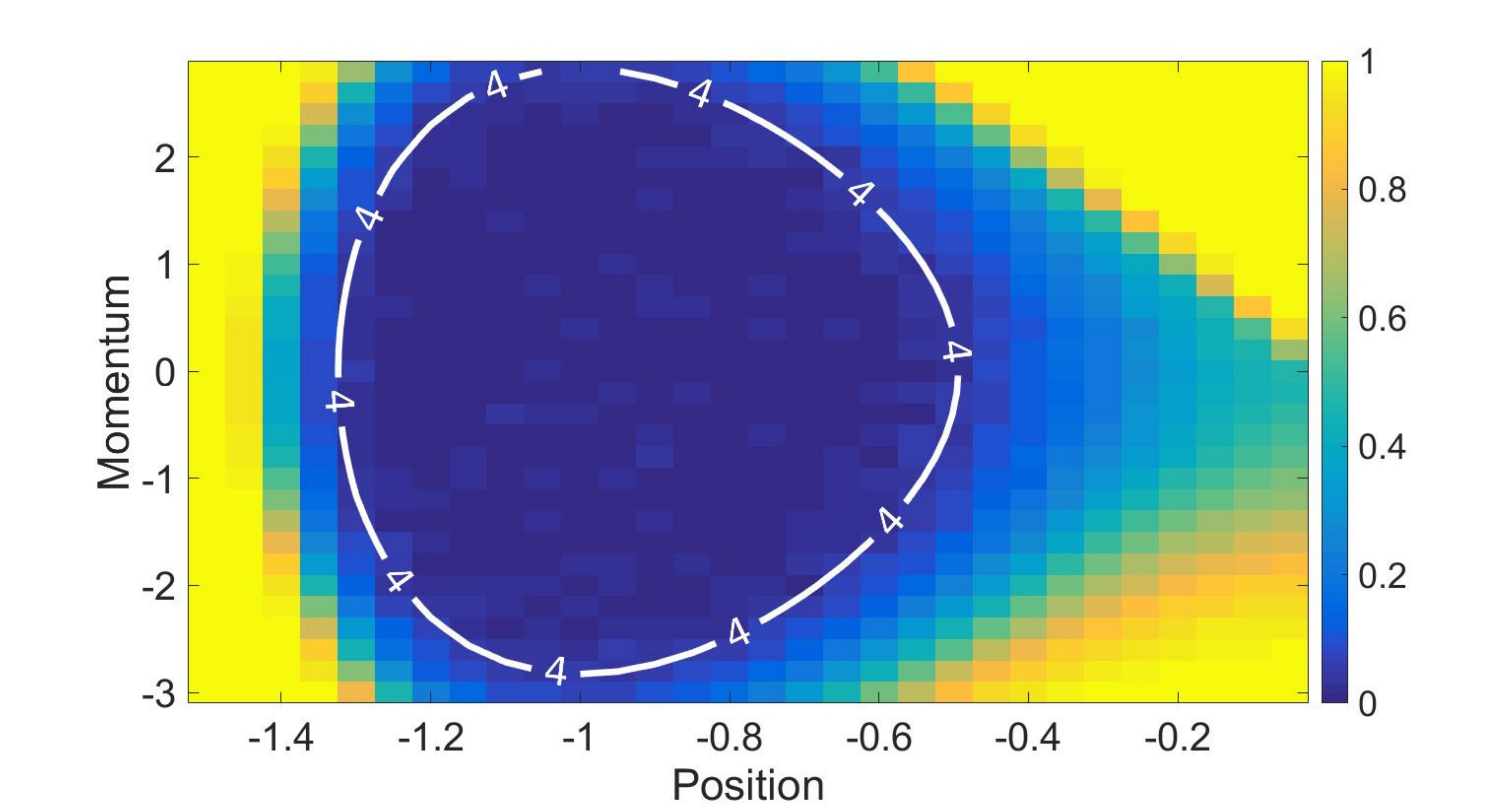}
  \subcaption{$\delta\left(x,p\right)$ when $\gamma=0.1$}
  \label{subfig:low_friction}
\end{minipage}%
\begin{minipage}{0.5\textwidth}
  \includegraphics[scale=0.15]{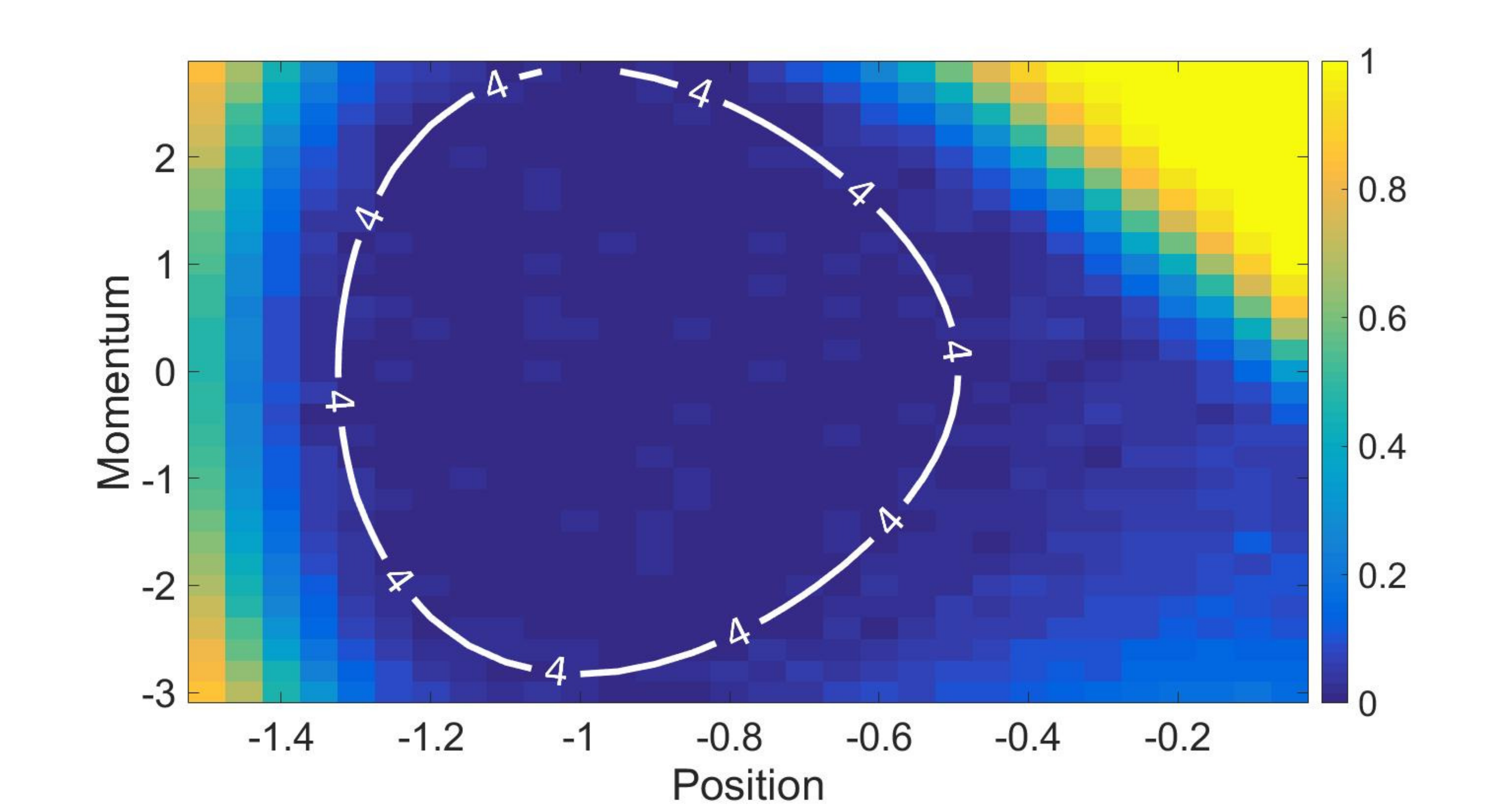}
  \subcaption{$\delta\left(x,p\right)$ when $\gamma=1$}
  \label{subfig:mid_friction_1}
\end{minipage}
\begin{minipage}{0.5\textwidth}
  \includegraphics[scale=0.15]{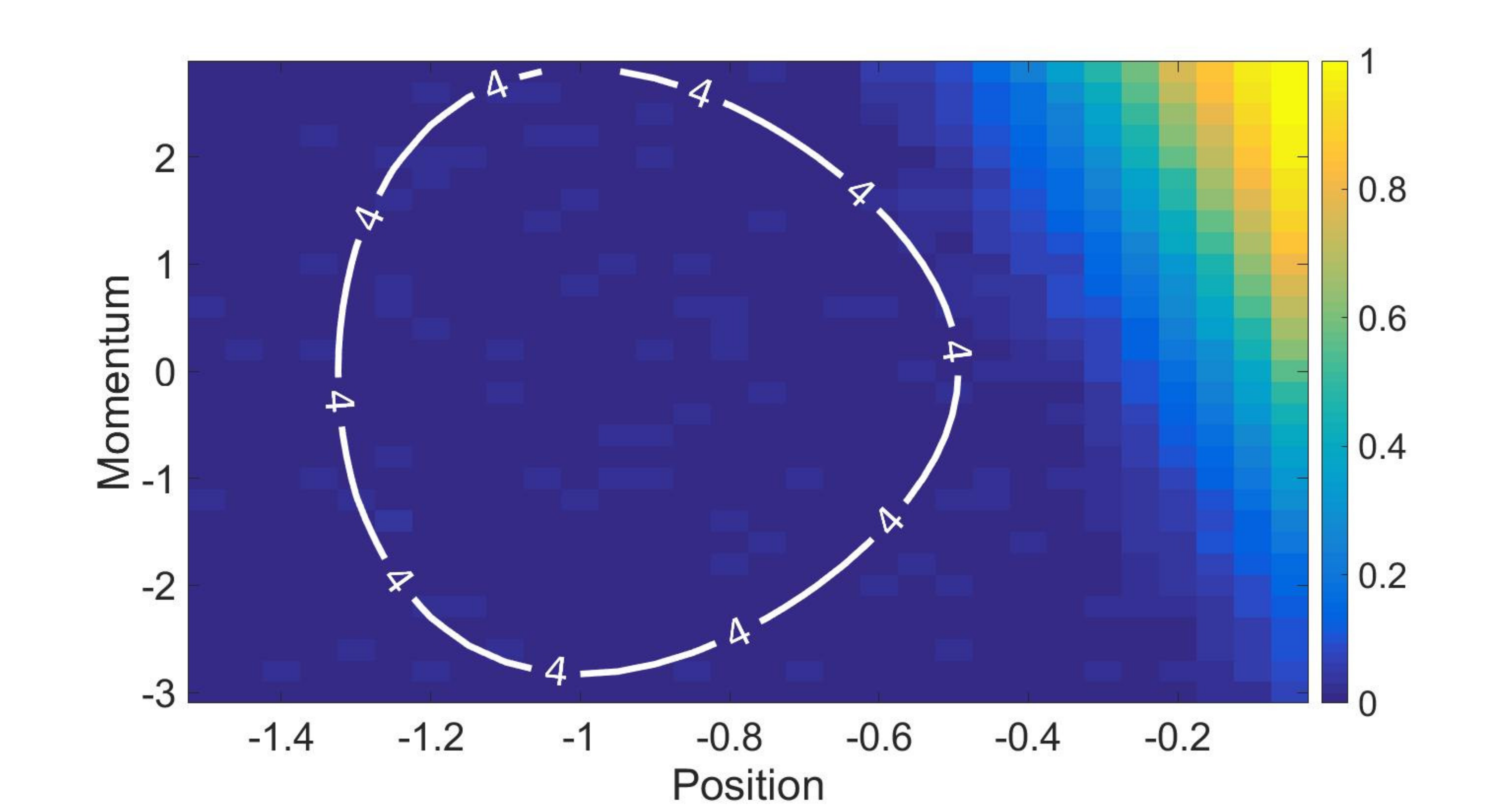}
  \subcaption{$\delta\left(x,p\right)$ when $\gamma=10$}
  \label{subfig:mid_friction_2}
\end{minipage}%
\begin{minipage}{0.5\textwidth}
  \includegraphics[scale=0.15]{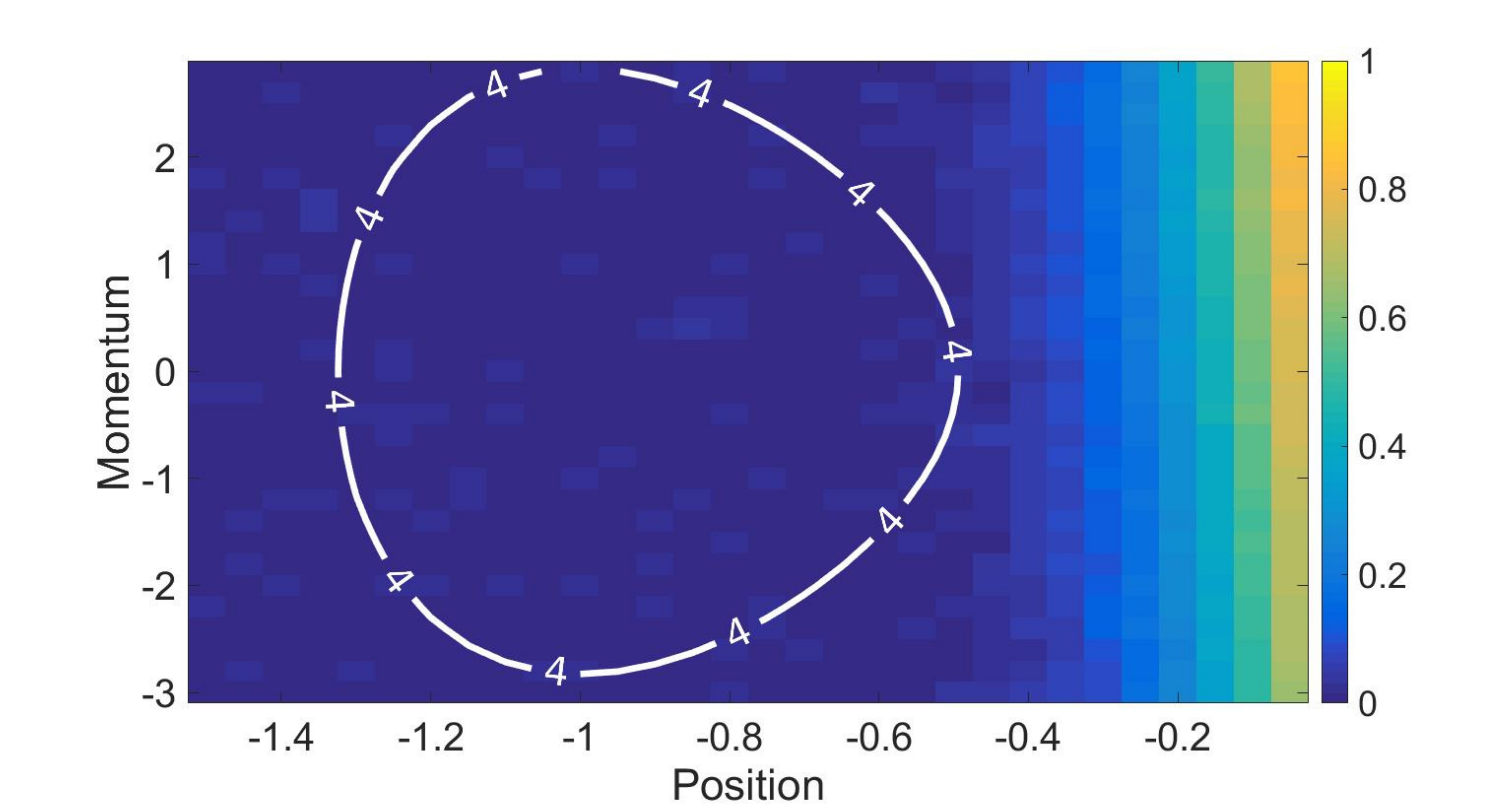}
  \subcaption{$\delta\left(x,p\right)$ when $\gamma=100$}
  \label{subfig:high_friction}
\end{minipage}
\caption{For particles initiated with $H(x,p)\leq A - 3k_BT$, well escape times are close to $\tau_r$. Heat maps show the fractional deviation in expected escape time $\delta(x,p)$ from the well-thermalised average $\tau_r$, as a function of initial position $x$ and momentum $p$. The labelled contours correspond to a well height $A=7$ with energy $H(x,p) = A - 3k_BT = 4k_BT$. These heat maps are representative of the situation for other barrier heights $A\geq 5k_BT$.
}
\label{fig:heatmaps}
\end{figure}

Despite the robustness of this result to the value of the friction, the heatmaps in Figure~\ref{fig:heatmaps} are friction-dependent. When $\gamma$ is low, the particle diffuses very slowly in energy space, and it is the challenge of diffusing within this energy space that prohibits escape from the well. As a result the heatmap corresponding to $\gamma=0.1$ (Figure~\ref{subfig:low_friction}) follows the shape of constant energy contours. As friction starts increasing (e.g. in Figure ~\ref{subfig:mid_friction_1} and ~\ref{subfig:mid_friction_2}), diffusion in momentum-space becomes more rapid, but diffusion in position-space slows down. Once $\gamma$ becomes very high (e.g. $\gamma=100$ in Figure~\ref{subfig:high_friction}), the behaviour of the heatmap is essentially determined by the initial position of the particle; those close to the barrier and with $U_{A,B}(x)$ sufficiently close to $A$ can escape easily, but the momentum is irrelevant. Using the total energy $H(x,p)$ as a criterion ensures that we account for all the regimes of friction.

Since we are interested in the typical time scale of transferring particles to a different well from the existent well, we will sample initial points only from the right well. We define the $\textbf{erasing time}$ $\tau_e$ as the expected time to hit the erasure region, given that the particle started in the right-hand well:
\begin{eqnarray}
\tau_e = \E\left[\inf\{t\geq 0 \mid x(t)<0 \text{ and } H(x(t),p(t))\leq A - 3k_BT\}\right]
\end{eqnarray}
where $\left(x\left(t\right),p\left(t\right)\right)$ is the solution to Equation~\ref{eq:control} with the initial condition $\left(x\left(0\right),p\left(0\right)\right)\sim_{\text{law}}\pi_1\left(x,p\right)$. 
Given this definition, $\tau_e$ indicates a typical time scale over which the control must be applied to successfully erase a large fraction of the bits. In practice, the control would be applied for a period $\tau > \tau_e$ to achieve high accuracy. We will use  $\tau_e$ as an indicative time scale of control operation for the purposes of our analysis. It is useful to decompose the erasing time $\tau_e$ as the sum of two times: the \textbf{transport time} and \textbf{mixing time}. 
\begin{itemize}
\item \textbf{Transport time ($\tau_t$)}: The time taken by the particle to reach well ``0'' given that it is initially distributed according to $\pi_1(x,p)$.
\begin{eqnarray}\label{eq:transport_time}
\tau_{t} = \E\left[\inf\{t\geq 0 \mid x\left(t\right)\leq 0\}\right]
\end{eqnarray}
where $x(t)$ is the solution to Equation~\ref{eq:control} with the initial condition $(x(0),p(0))\sim_{\text{law}}\pi_1(x,p)$.
\item \textbf{Mixing time ($\tau_m$)}: The time taken by the particle to mix sufficiently inside the well. This is the time starting from when the particle first reaches well ``0'' to when it first hits the erasure region. 
\begin{eqnarray}\label{eq:mixing_time}
\tau_m = \tau_e - \tau_t 
\end{eqnarray}
\end{itemize} 

\subsubsection{Cost of erasing}\label{subsubsec:cost}

In this section, we calculate the work done in erasing a bit. From Sekimoto's expression~\cite{sekimoto1997kinetic, sekimoto1998langevin}, for a protocol applied for  a time $\tau$ and with a region of effect $I=\{x\geq 0 \mid A - U_{A,B}\left(x\right)+F\cdot x \geq 0\}$,
\begin{eqnarray}\label{eq:sekimoto_work}
\langle W\rangle: = \Int_{0}^{\tau}\Int_{x\in I}\frac{\partial V_F\left(x,t\right)}{\partial t}p(x,t)\,dx\,dt,
\end{eqnarray}
where $p(x,t)dxdt$ is the probability that the particle is between position $(x,x+dx)$ in the time interval $(t,t+dt)$. There are two potential sources of work that appear in our calculation.

\begin{enumerate}
\item When we begin the erasure protocol by switching on the control to lift the particle.
\item At the end of the protocol when we switch off the control.
\end{enumerate}

We note that in our family of controls, there is negligible energy recovered when the control is switched off (refer to Section~\ref{subsec:no_recovery_energy} of the Supplementary Information), since the probability of the particle being in the region in which the control is applied is small. More generally, the question of whether energy might be recovered from small systems and stored efficiently is a complex one, despite the optimism shown in previous discussions of erasing. Indeed, current technology does not attempt to recover any energy from bits. 

We now calculate the work done for our protocol, Example \ref{ex:control_potential}. The particle's initial potential energy is approximately ${k_BT}/{2}$ on average, due to the equipartition theorem, and after the control is switched  the average potential energy is $A+ F \cdot B$ for a particle in the right well, since the particle is localised around $x=B$, and still ${k_BT}/{2}$ for a particle in the left well. So, ignoring energy recovery at the end of the operation, the net work done for the erasure protocol is $W={\left(A + F\cdot B -k_BT/2\right)}/{2}$. As justified analytically and numerically in Section~\ref{appendix-full work} of the Supplementary Information, this approximation is accurate for the values of $A$ and $F$ that we consider, and we will use this as the form of work for the rest of the manuscript.

\begin{observation}\label{observ:work_increasing_height}
Work is an increasing function of well height $A$ at fixed $F$ and $\gamma$. This follows immediately from the expression of work $W={\left(A + F\cdot B -k_BT/2\right)}/{2}$.
\end{observation}

\section{Friction-based trade-offs for reliability and erasing}\label{sec:times}

We explore the behaviour of the reliability and erasing time-scales as functions of the friction coefficient. We find that both these time-scales are non-monotonic, roughly U-shaped functions of the friction coefficient. A high reliability time requirement is favoured by a very low or very high friction; whereas a low erasing time requirement is helped by the choice of a moderate value of friction. Since a bit designer would seek reliable bits (needing high or low friction) that can be erased fast (needing intermediate friction), this yields a friction-based trade-off between reliability and speed of erasure. 

\subsection{Reliability Time}
Our definition of reliability time (Equation~\ref{eq:reliability_time}) is very similar to the classic problem of escape rates from one-dimensional wells (Fig~\ref{fig:kramer_figure}), as applied in transition state theory to understand chemical reactions. In a famous paper~\cite{kramers1940brownian}, Kramer found analytic expressions for  the escape  rate $k$ from a well by calculating the flux of particles between a source on one side of the barrier ($x_A$) and a sink at the other side ($x_B$). Kramer's expressions apply separately to the regimes of low friction, moderate to high friction and very high friction. Later the groups of Melnikov and Meshkov~\cite{mel1986theory} and Pollack, Grabert and H{\"a}nggi~\cite{pollak1989theory} gave formulae that interpolate accurately over all values of friction (see review in~\cite{hanggi1990reaction}). We will apply the result of Melnikov and Meshkov to estimate analytical forms of the escape rate for our bistable system
\begin{align}\label{eq:meshkov_melnikov}
\begin{split}
k ={}& \frac{\omega_0}{2\pi}\left[\sqrt{1 + \frac{\gamma^2}{4\omega_b^2}} - \frac{\gamma}{2\omega_b}\right]g\e^{-{A}/k_BT}, \mbox{where} \\
\ln g ={}& 
\frac{1}{2\pi}
\int_0^\frac{\pi}{2}
\ln
\left[
1 - \operatorname{exp}
\left(
\frac{-\gamma\, I(A)}{4k_BT\cos^2 x}
\right)
\right]dx.
\end{split}
\end{align}

Here, $\omega_b$ is the angular frequency at barrier height, $\omega_0$ is the angular frequency at the bottom of the well and
$I(A)$ is the action for barrier height $A$. Refer to Section~\ref{subsec:well_parameters} in the Supplementary Information for a detailed definition of these parameters and calculations for our system.

\begin{figure}[h!]
\begin{center}
\[\includegraphics[scale=0.4]{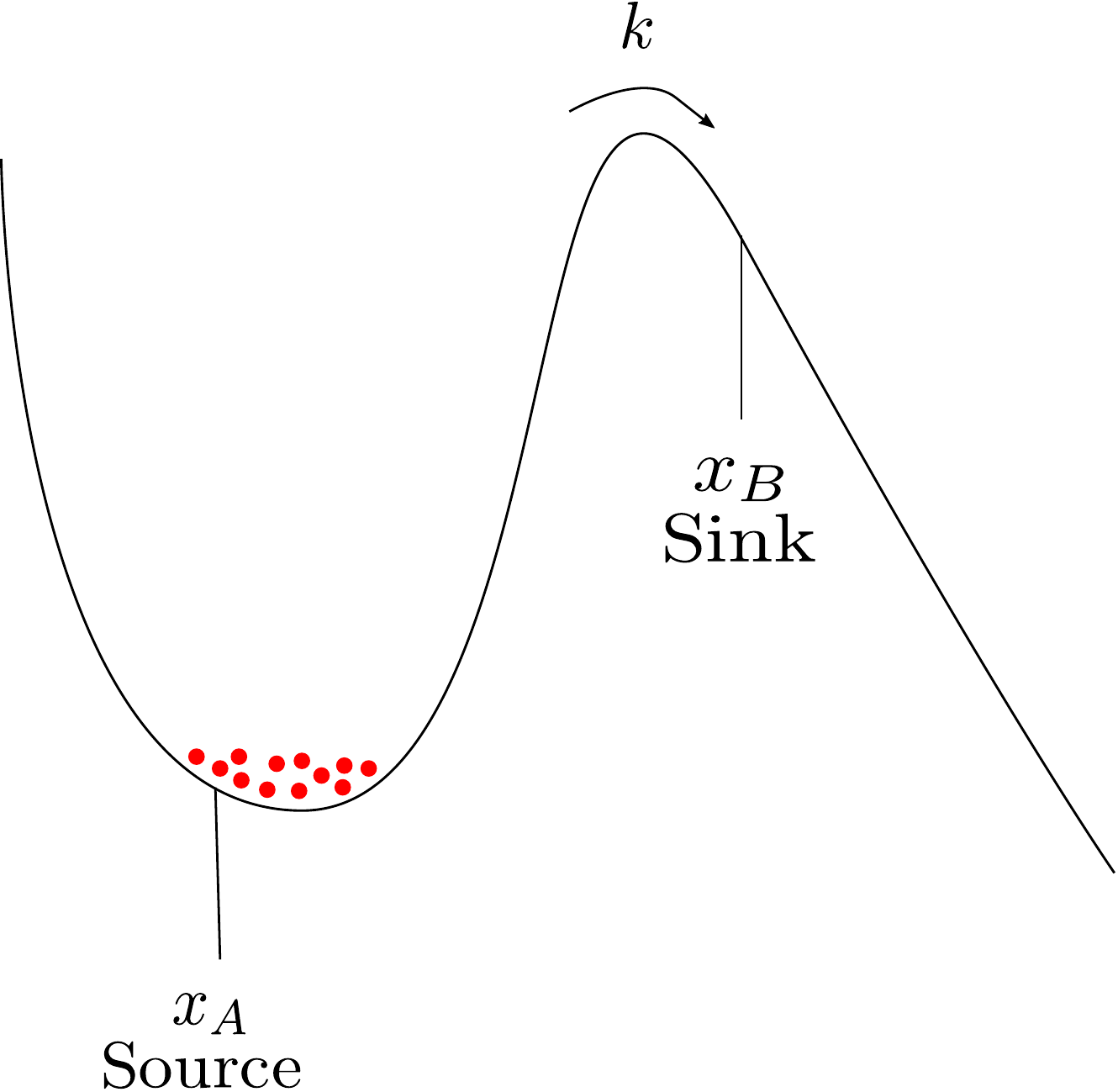}\]
\end{center}
\caption{The escape of particles from a one-dimensional well. Kramer~\cite{kramers1940brownian} considered a source of particles at the bottom of the well, and estimated the rate of escape to a sink on the far side of a barrier.}
\label{fig:kramer_figure}
\end{figure}

We plot the analytical prediction of $1/k$ given by Eq.~\ref{eq:meshkov_melnikov} in Fig.~\ref{fig:reliability_form} for two values of well height $A$, as a function of friction $\gamma$. This prediction is compared to average first passage time for particles to reach the top of the barrier from an initial Boltzmann distribution within a single well. The two quantities differ at large $\gamma$ because Kramer's definition does not treat a particle that crosses the barrier but then immediately crosses back as having ``escaped", whereas  our definition of reliability in terms of a first passage time treats such particles as no longer being reliable. In the underdamped regime, immediate recrossings are rare and hence $\tau_r$ and $1/k$ coincide; in the overdamped regime, particles that reach the barrier top have a $50\%$ chance of returning and so $\tau_r = 1/2k$. As can be seen from Fig.~\ref{fig:reliability_form}, $\tau_r$ smoothly interpolates between $\frac{1}{k}$ and $\frac{1}{2k}$, with the small numerical factor providing only a minor correction to the underlying physics of the analytical expression in Eq.~\ref{eq:meshkov_melnikov}.

\begin{figure}[h!]
\centering
\begin{minipage}{0.5\textwidth}
  \centering
  \includegraphics[scale=0.17]{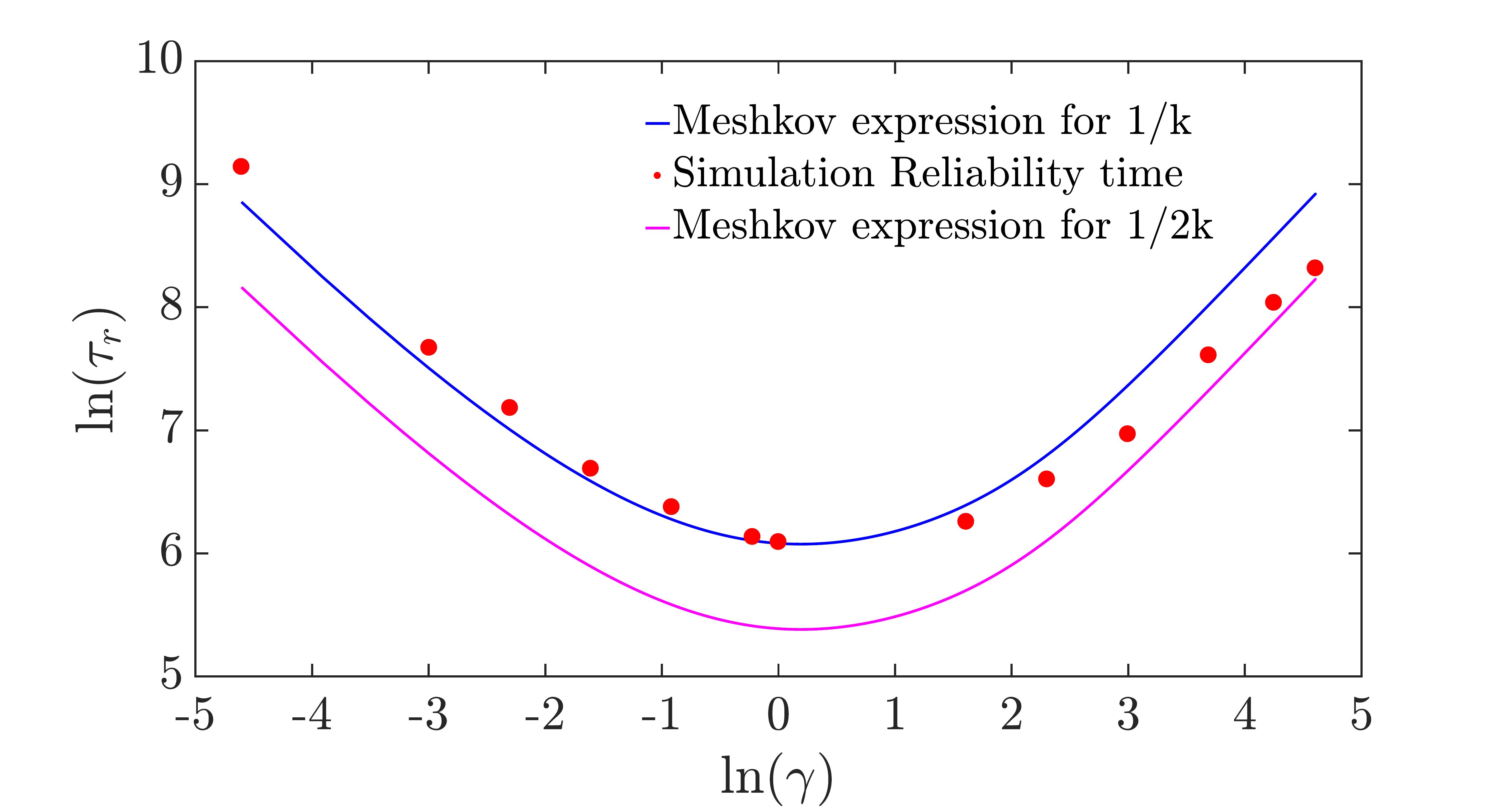}
  \subcaption{$A=6$}
\end{minipage}%
\begin{minipage}{0.5\textwidth}
  \centering
  \includegraphics[scale=0.17]{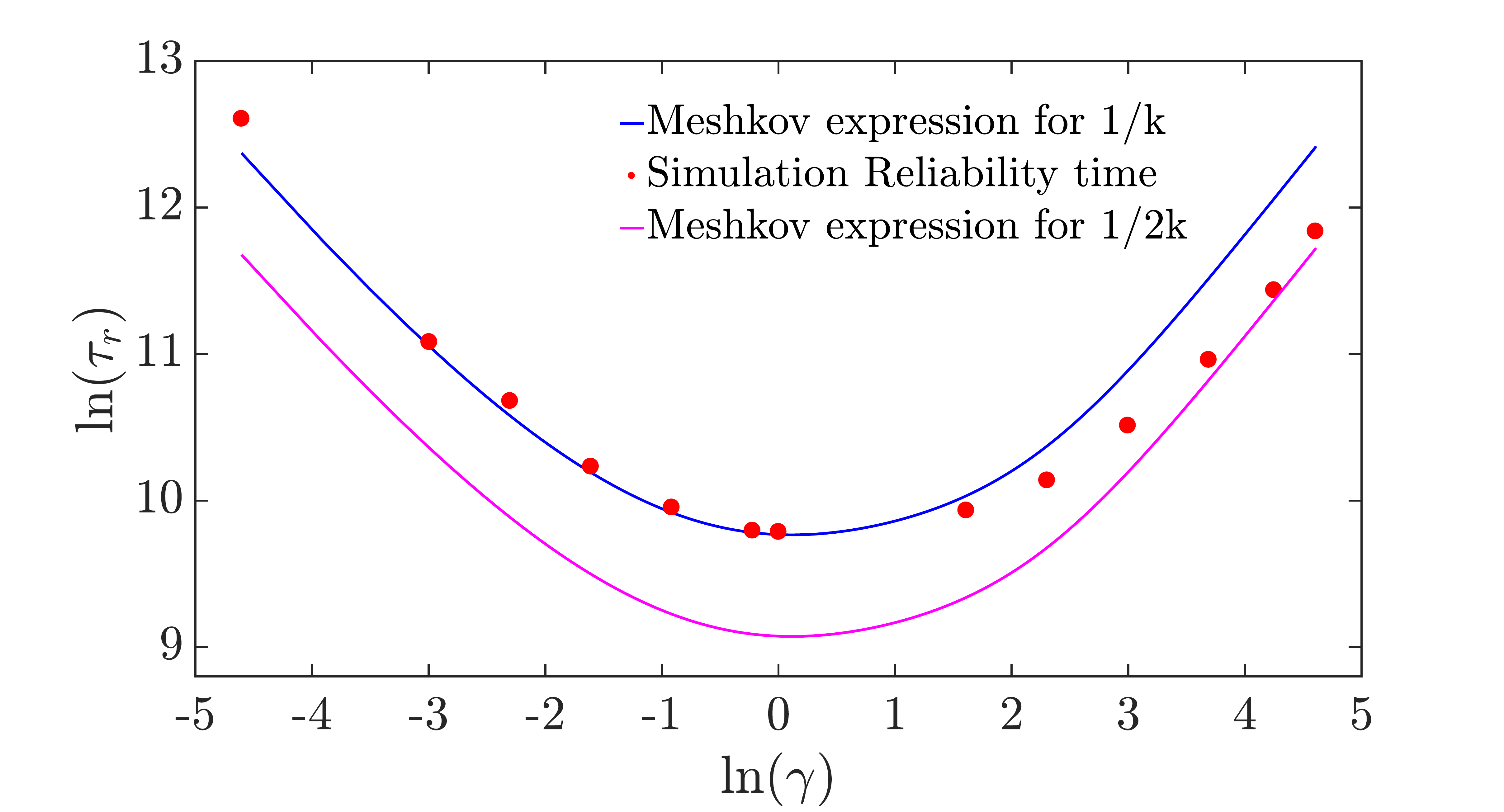}
  \subcaption{$A=10$}
\end{minipage}
\caption{The reliability time $\tau_r\propto\frac{1}{\gamma}$ in the low friction regime, $\tau_r\propto\gamma$ in the high friction regime and is minimum at moderate friction. Simulation results are compared to the inverse of escape rate from a single well $(1/k)$ and $(1/{2k})$, as predicted by Eq.~\ref{eq:meshkov_melnikov}. Here, and elsewhere in the manuscript, error bars are omitted when comparable to data points.} 
\label{fig:reliability_form}
\end{figure}

The Melnikov-Meshkov expression predicts an almost-exponential scaling of $1/k$ with barrier height $A$, which is reproduced by $\tau_r$ and expected from the Arrhenius rate law~\cite{arrhenius1889dissociationswarme}. Note that both $1/k$ and $\tau_r$ are non-monotonic in friction $\gamma$, with long reliability times in the underdamped and overdamped limits. This behaviour results from the need for particles to diffuse in both position and energy in order to reach the top of the barrier from an initial state thermalized within a single well. At high friction, particles rapidly sample different kinetic energies due to strong coupling with the environment, but move slowly in position space and hence take a long time to cross the barrier. At low friction, particles can move rapidly but their energy remains effectively constant over short time periods. They only cross the barrier when they have eventually gained enough total energy. Intermediate friction, when neither process is excessively slow, gives the shortest $\tau_r$. This behaviour is typical of equilibrating systems in which an initial out-of-equilibrium  condition (particles are guaranteed to be on one side of the well and not the other) relaxes towards an equilibrium state (particles on both sides of the barrier), and is thus insensitive to the details of our bit design. 

A more detailed analysis of the dependence of the reliability time on various parameters, and indeed the functional form of the well, is possible. However, these details are not necessary for the conclusions  we draw in the rest of this manuscript, and hence we omit them here.

\subsection{Erasing time}

As noted earlier, the erasing time is composed of two parts: the transport time defined in Eq.~\ref{eq:transport_time} and the mixing time defined in Eq.~\ref{eq:mixing_time}. We now present analytical estimates of these times and compare them with numerical solutions.

\subsubsection{Transport time}
We can obtain an analytical estimate of the transport time in low and high friction limits by assuming that a particle starting at $x=B$ moves deterministically under the influence of the potential slope and drag force. 

\begin{enumerate}
\item \textbf{Low friction regime}: The particle travels with a constant acceleration of $\frac{F}{m}$ and the time taken to travel a distance $B$ is $\tau_{t}\approx\sqrt{{2mB}/{F}}$.

\begin{figure}[h!]
\begin{minipage}{0.5\textwidth}
  \centering
  \includegraphics[scale=0.17]{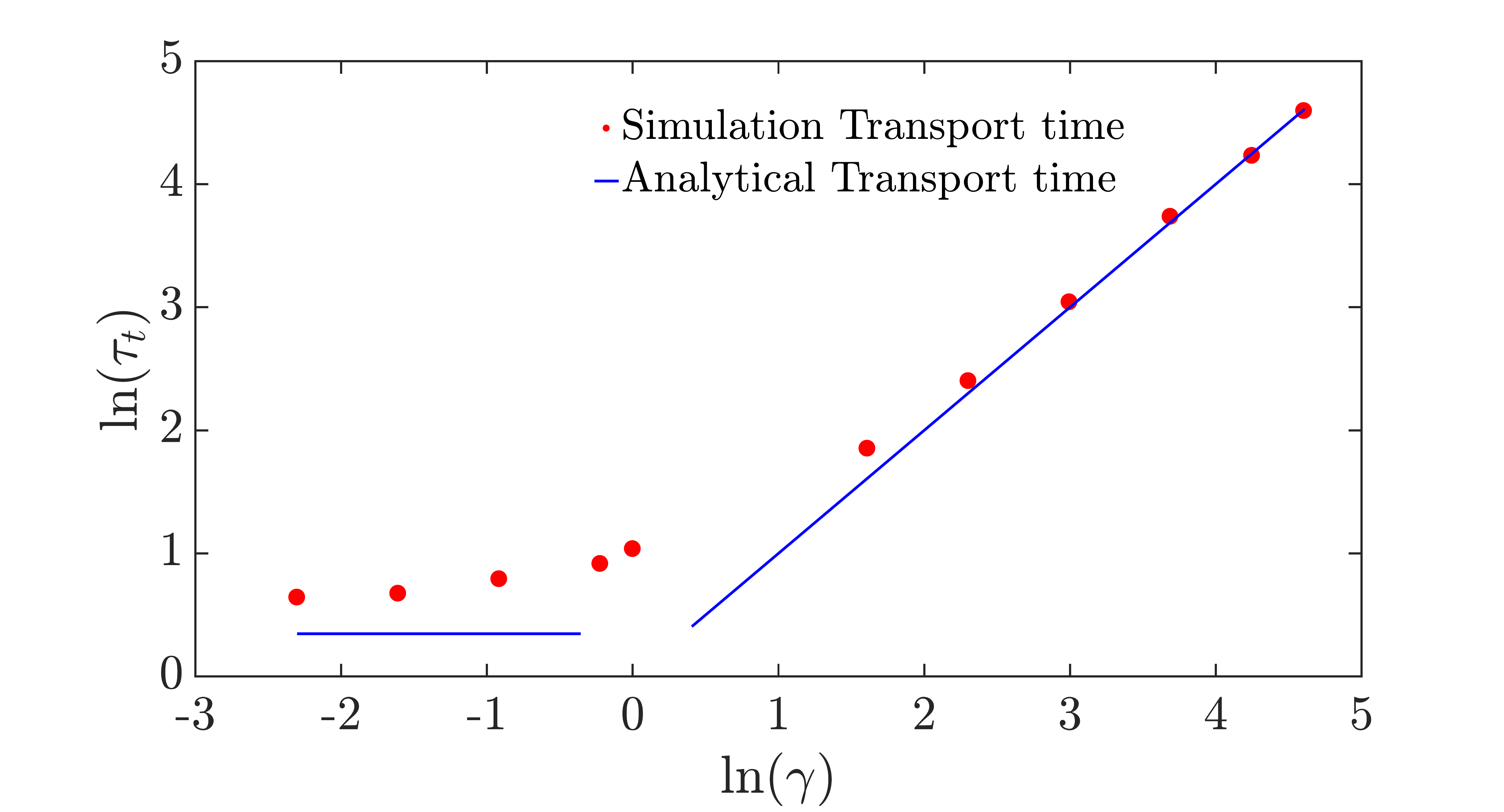}
  \subcaption{$A=10,F=1$}
  \label{fig:transport_low_F}
\end{minipage}
\begin{minipage}{0.5\textwidth}
  \centering
  \includegraphics[scale=0.17]{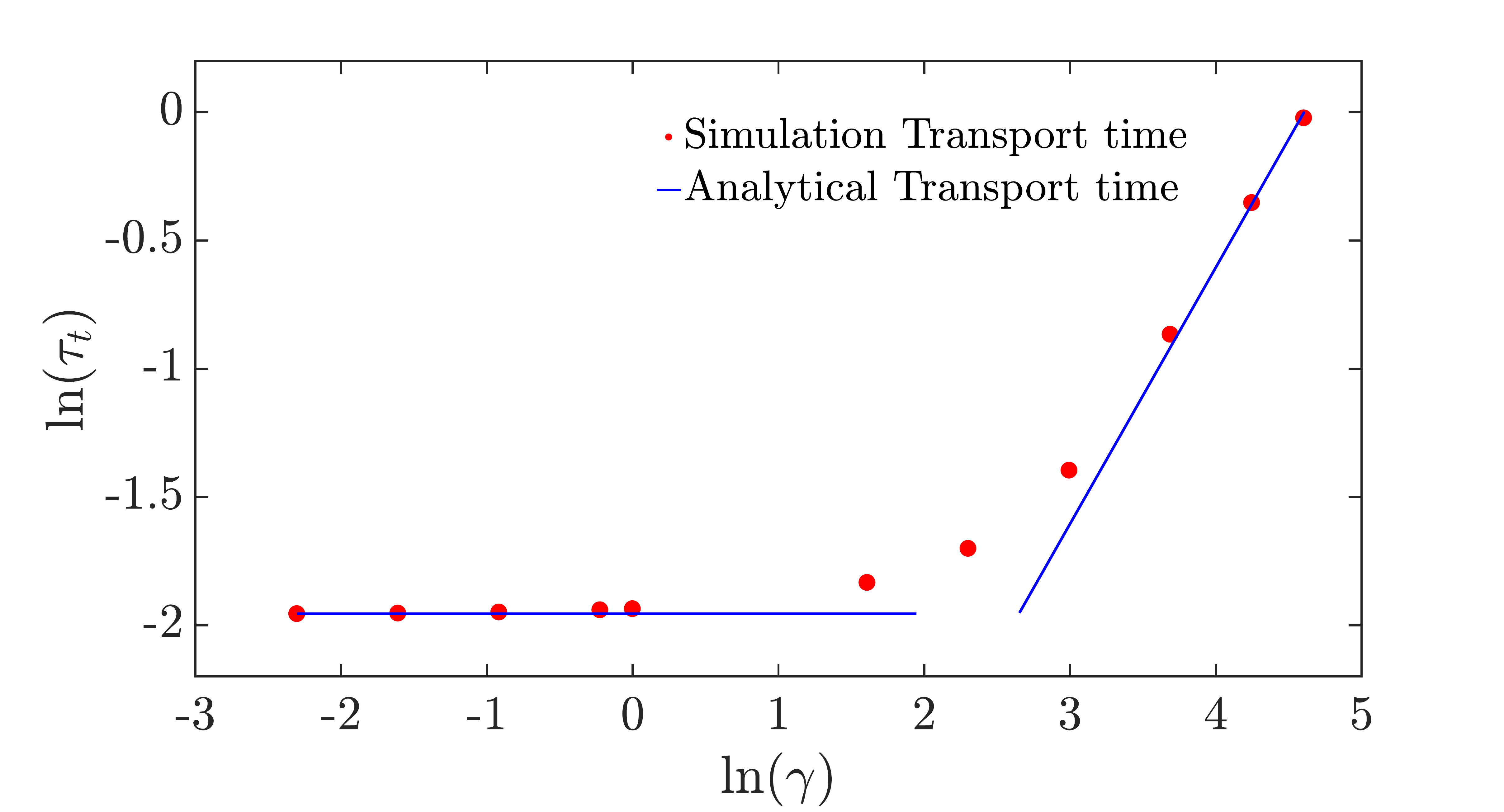}
  \subcaption{$A=10,F=100$}
  \label{fig:transport_high_F}
\end{minipage}
\caption{The transport time obtained from simulations approximates the analytical estimates of  $\tau_{t}\approx\sqrt{{2mB}/{F}}$ in the low friction regime, $\tau_{t}\approx{mB\gamma}/{F}$ ($\propto\gamma$) in the high friction regime.}
\label{fig:transport}
\end{figure}

\item \textbf{High friction regime}: In this regime, we assume that the net force on the particle (arising from the sum of drag and potential) is zero. The particle travels with a velocity of ${F}/{m\gamma}$, and the time taken to travel a distance $B$ is $\tau_{t}\approx{mB\gamma}/{F}$.
\end{enumerate}

We thus expect the transport time to be constant in the underdamped regime and increase linearly with friction in the overdamped regime. Figure~\ref{fig:transport} illustrates that this scaling is observed in Langevin simulations, and that numerical values are in reasonable agreement with these crude estimates. The largest quantitative deviations occur at low force and low friction (e.g. $F=1$ in Figure~\ref{fig:transport_low_F}), when the diffusion of the particle on the slope contributes significantly to $\tau_t$. This results in a simulation transport time larger than the analytical estimate.

\subsubsection{Mixing time}
Similar to the transport time, analytical estimates of the mixing time can be obtained in the limits of high and low friction.

\begin{enumerate}
\item \textbf{Low friction regime}: For purposes of approximate calculation we treat the well ``0'' as a harmonic oscillator. Deterministically, the energy of a harmonic oscillator decays exponentially in the underdamped regime. Therefore we have $E\left(t\right) = E_0\e^{-\gamma t}$, where $E_0$ is the initial energy of the particle when it first reaches $x=0$ and $E\left(t\right)$ is the energy of the particle at time $t$. In the underdamped regime, a particle starting at $B$ arrives at position $x=0$ with energy $E_0 \approx A + F\cdot B$. Thus solving for $E\left(\tau_{\rm mix}\right) = A-3k_BT$,
\begin{eqnarray}\label{eq:low_friction_mixing}
\tau_{mix}\approx\frac{1}{\gamma}\log{\frac{A + F\cdot B}{A-3k_BT}}
\end{eqnarray}

\begin{figure}[h!]
\begin{minipage}{0.5\textwidth}
  \centering
  \includegraphics[scale=0.17]{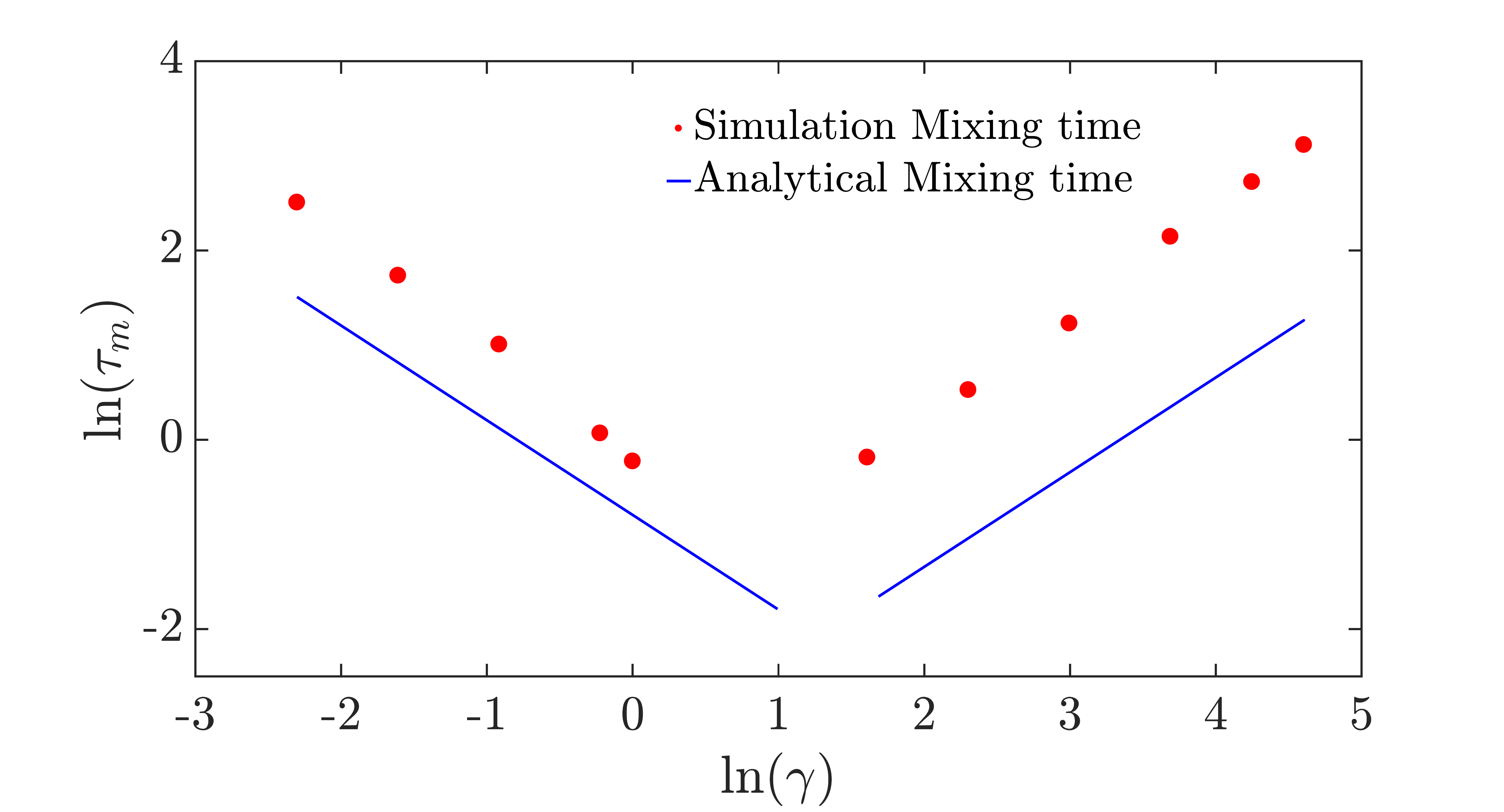}
  \subcaption{$A=10,F=1$}
  \label{fig:mixing_low_F}
\end{minipage}
\begin{minipage}{0.5\textwidth}
  \centering
  \includegraphics[scale=0.17]{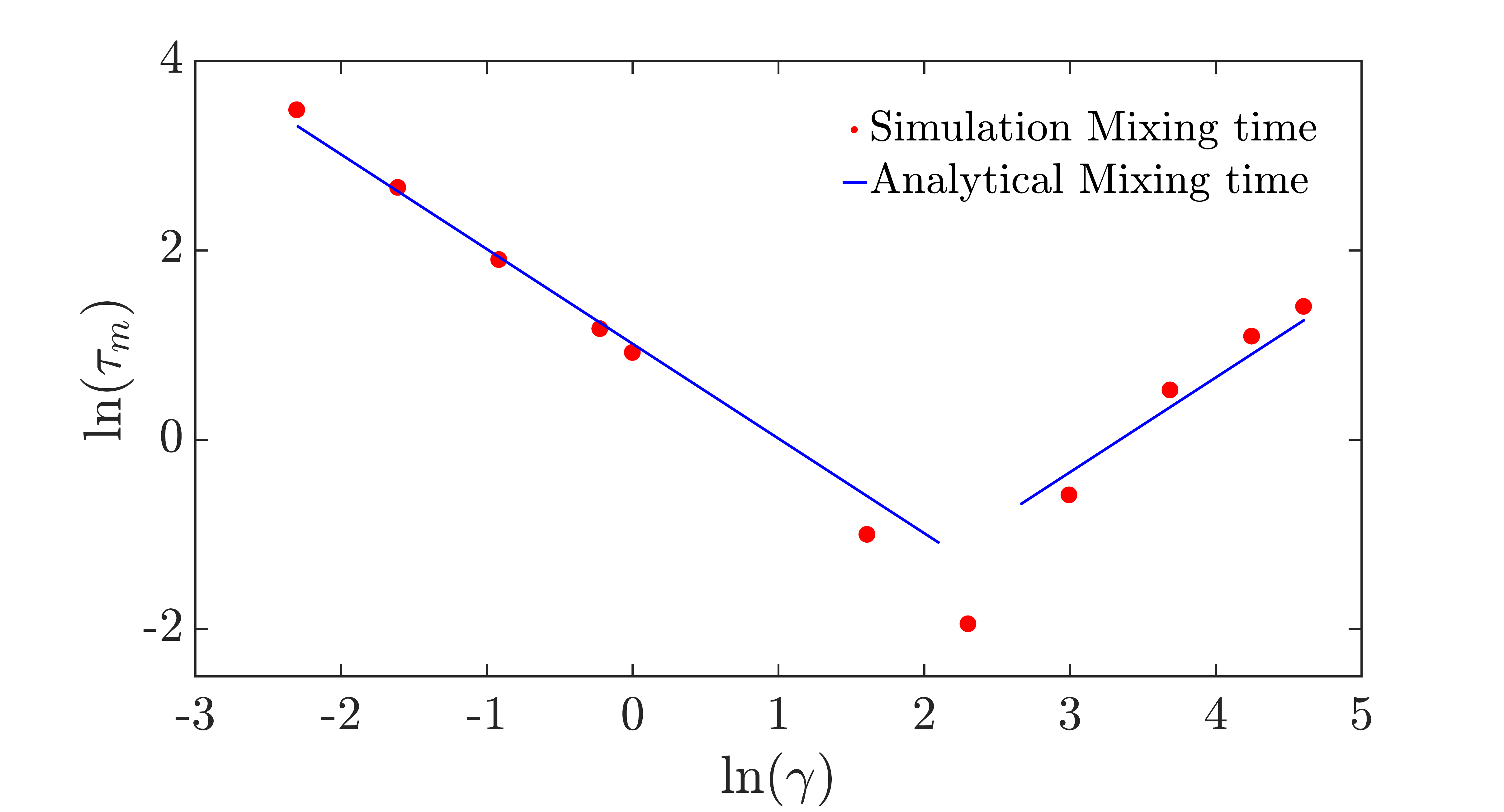}
  \subcaption{$A=10,F=100$}
  \label{fig:mixing_high_F}
\end{minipage}
\caption{Evidence from simulation that the mixing time $\tau_{mix}\approx\frac{1}{\gamma}\log{\frac{A + F\cdot B}{A-3k_BT}}\left(\propto\frac{1}{\gamma}\right)$ in the low friction regime, $\tau_{mix}\approx\frac{mB^2\gamma}{2\sqrt{2}A}\left(\propto\gamma\right)$ in the high friction regime and is minimised at moderate friction.} 
\label{fig:mixing}
\end{figure}

\item \textbf{High friction regime}: A sensible estimate of the behaviour can be obtained by explicitly modelling the diffusion of the particle near the barrier top. In the overdamped limit, the criterion of reaching a total energy of $E\left(\tau_{\rm mix}\right) = A-3k_BT$ is equivalent to reaching a point $d$ which has potential energy of $A-3k_BT$, since momenta are sampled arbitrarily rapidly in this limit. To proceed, we consider the typical time required to reach an absorbing barrier at $d$ starting from $x=0$, assuming a sufficiently large $F$ that we can treat $x=0$ as a reflecting barrier. Starting from the overdamped stochastic differential equation
\begin{eqnarray}\label{eq:overdamped_langevin}
m\gamma\, dx = -\partial_{x}U_{A,B}\left(x\right)\,dt + \sqrt{2m\gamma k_BT}\,dW
\end{eqnarray}
with  generator $\mathcal{L}= \frac{k_BT}{m\gamma}\e^{\frac{U_{A,B}\left(x\right)}{k_BT}}\partial_x\e^{\frac{-U_{A,B}\left(x\right)}{k_BT}}\partial_x$, we apply the standard methods outlined in Pavliotis~\cite[(7.1), pp.~239]{pavliotis2014stochastic}, which leads to the following system of equations for the average mixing time $\tau_{\rm mix}(x)$ as a function of the initial position $x$ 
\begin{eqnarray}\label{eq:zwanzig}
\begin{aligned}
\frac{k_BT}{m\gamma}\e^{\frac{U_{A,B}\left(x\right)}{k_BT}}\partial_x\e^{\frac{-U_{A,B}\left(x\right)}{k_BT}}\partial_x\tau_{mix}\left(x\right) ={}& -1, d < x \leq 0. \\
\tau_{mix}\left(x\right) = {}& 0, x = d. 
\end{aligned}
\end{eqnarray}
We can solve Equation~\ref{eq:zwanzig} using appropriate limits to get
\begin{eqnarray}\label{eq:subsequent}
\tau_{mix}\left(x\right) = \frac{m\gamma}{k_BT}\int_{0}^{\sqrt{\frac{3k_BT}{2A}}B}\int_{0}^{q}\e^\frac{U\left(q\right)-U\left(r\right)}{k_BT}dqdr,
\end{eqnarray}
where we have approximated the potential near the barrier as an inverted harmonic oscillator to estimate $d = \sqrt{\frac{3k_BT}{2A}}B$. Repeating this approximation within the integral, we obtain
\begin{eqnarray}\label{eq:high_friction_mixing}
\tau_{mix}\left(x\right)\approx\frac{m\gamma}{k_BT}\int_{0}^{B\sqrt{\frac{3k_BT}{2A}}}\int_{0}^{q}\e^{\frac{2A\left(r^2-q^2\right)}{B^2k_BT}}dqdr\approx\frac{mB^2\gamma}{2\sqrt{2}A}.
\end{eqnarray}
\end{enumerate}

Equations~\ref{eq:low_friction_mixing} and~\ref{eq:high_friction_mixing} predict that the mixing time will scale as $1/\gamma$ in the low friction limit and as $\gamma$ in the high friction limit. In the first case, mixing within the well is limited by the rate at which the particle can reduce its total energy, whereas in the second it is determined by the speed with which the particle can diffuse in position space to a configuration with lower potential energy.  We plot simulation results for the mixing time, along with the analytic predictions, in Fig.~\ref{fig:mixing}, confirming this scaling and the resultant non-monotonicity. Quantitatively, simulation results deviate from the crude analytic predictions at low force (e.g. $F=1$ in Figure~\ref{fig:mixing_low_F}), when it is no longer reasonable to treat $x=0$ as either a reflecting barrier or a steep side of a harmonic well. Instead, excursions of the particle back onto the slope occupying the region $x>0$ lead to much larger mixing times. Nonetheless, the scaling and non-monotonicity in friction are preserved. Combining $\tau_{trans}$ and $\tau_{mix}$ gives $\tau_e$, plotted in Fig.\,~\ref{fig:non_monotonic_gamma_fixed_A_and _F}. Analytically, the erasing time is given as:

\begin{figure}[h!]
\begin{minipage}{0.5\textwidth}
  \centering
  \includegraphics[scale=0.17]{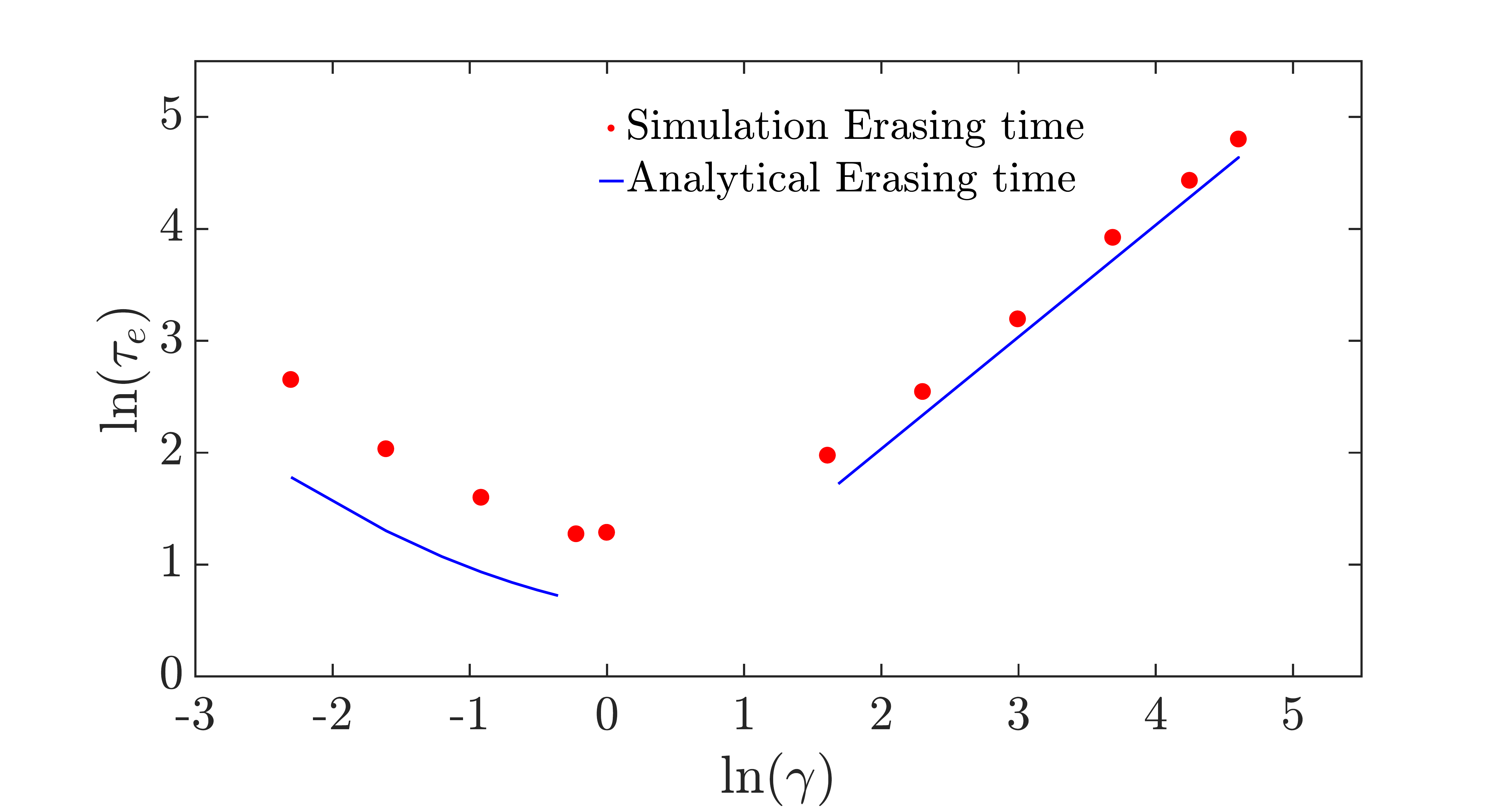}
  \subcaption{$A=10,F=1$}
  \label{fig:erasing_low_F}
\end{minipage}
\begin{minipage}{0.5\textwidth}
  \centering
  \includegraphics[scale=0.17]{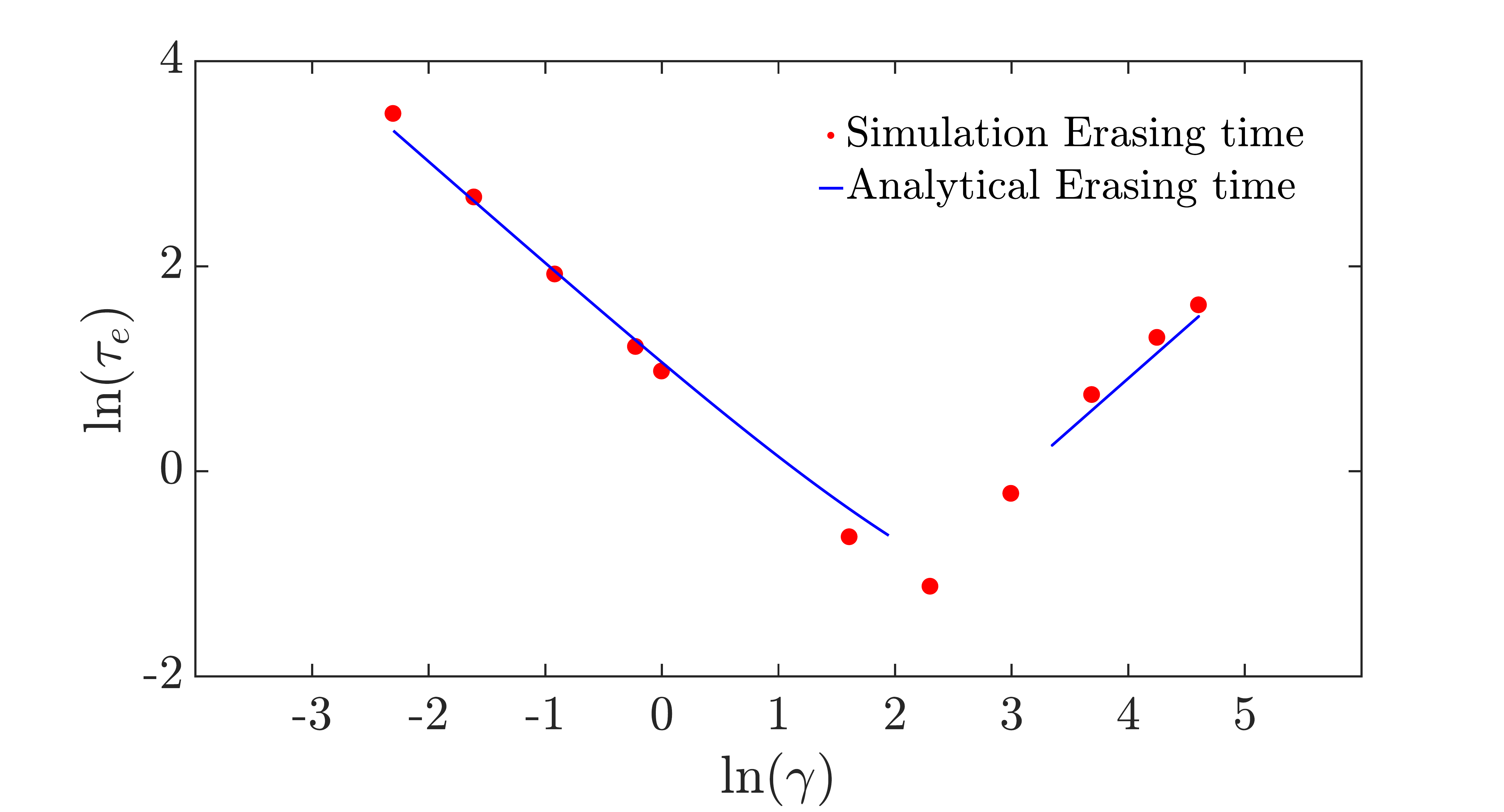}
  \subcaption{$A=10,F=100$}
  \label{fig:erasing_high_F}
\end{minipage}
\caption{Evidence from simulation that the erasing time $\tau_e\approx\sqrt{\frac{2mB}{F}}+\frac{1}{\gamma}\log{\frac{A + F\cdot B}{A-3k_BT}}$  in the low friction regime, scaling as $1/\gamma$, and $\tau_e\approx\frac{mB\gamma}{F}+\frac{mB^2\gamma}{2\sqrt{2}A}$ in the high friction regime, scaling as $\gamma$. The erasing time is minimised at moderate friction.}
\label{fig:non_monotonic_gamma_fixed_A_and _F} 
\end{figure}

\begin{enumerate}
\item\textbf{Low friction regime}:
\begin{eqnarray}\label{eq:erasing_low}
\tau_e\approx \sqrt{\frac{2mB}{F}} + \frac{1}{\gamma}\log{\frac{A + F\cdot B}{A-3k_BT}}.
\end{eqnarray}
\item\textbf{High friction regime}:
\begin{eqnarray}\label{eq:erasing_high} 
\tau_e \approx \frac{mB\gamma}{F} + \frac{mB^2\gamma}{2\sqrt{2}A}
\end{eqnarray}
\end{enumerate}

Like reliability, erasing time is large in the underdamped and overdamped limits, and minimized at intermediate values of friction. The physical cause is the same as before; our erasing protocol involves setting the system into a non-equilibrium state, and waiting for the system to relax towards an equilibrium in the perturbed potential. This process requires the system to diffuse in energy space and also explore configuration space, and is therefore favoured by intermediate friction. Specifically, if the friction is too low, the particle oscillates and slowly loses energy to be confined within the desired well. If the friction is too high, both the transport and mixing times increase as the particle's movement through space is so slow. The relative importance of these effects can be seen in Fig.~\ref{fig:mixing_transport_combined}. We note that the value of the damping $\gamma$ that minimises $\tau_e$ is quite sensitive to $F$ (Fig.~\ref{fig:mixing_transport_combined}). Fundamentally, a larger $F$ means the challenge of moving in position-space is made easier, and a greater loss of energy is needed to reach equilibrium. Therefore a higher friction coefficient is optimal. As with the reliability time, further analysis is possible but not necessary for the conclusions we wish to draw. Once again, the key point is the trade-off between high and low friction, which is not specific to our control. Indeed, is likely to be quite generic since any protocol will necessary push the system out of equilibrium, and will require particles to be typically confined within the target well before the control is removed.

\begin{figure}[h!]
\begin{minipage}{0.5\textwidth}
  \centering
  \includegraphics[scale=0.17]{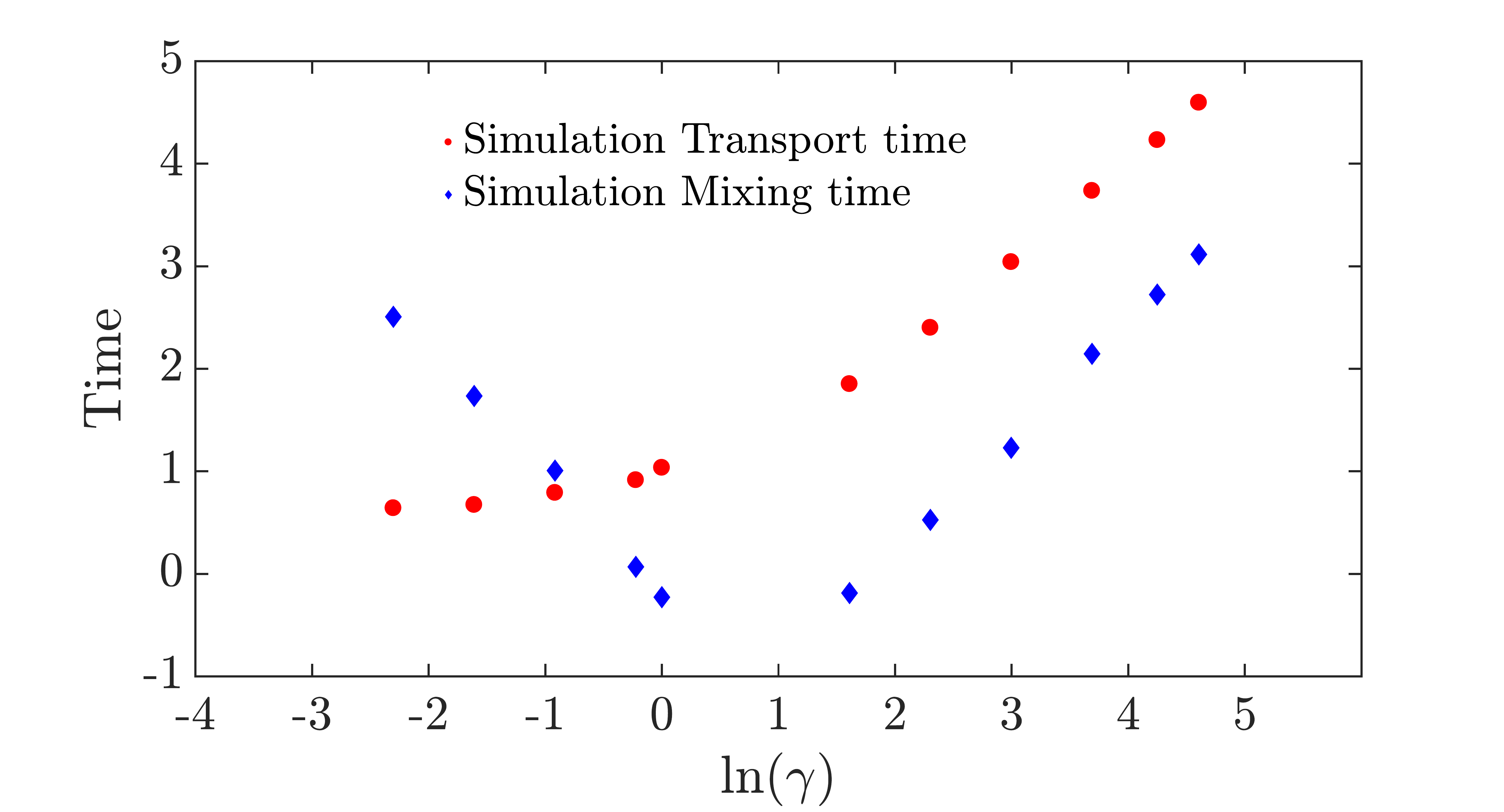}
  \subcaption{$A=10,F=1$}
  \label{fig:mixing_transport_low_F}
\end{minipage}
\begin{minipage}{0.5\textwidth}
  \centering
  \includegraphics[scale=0.17]{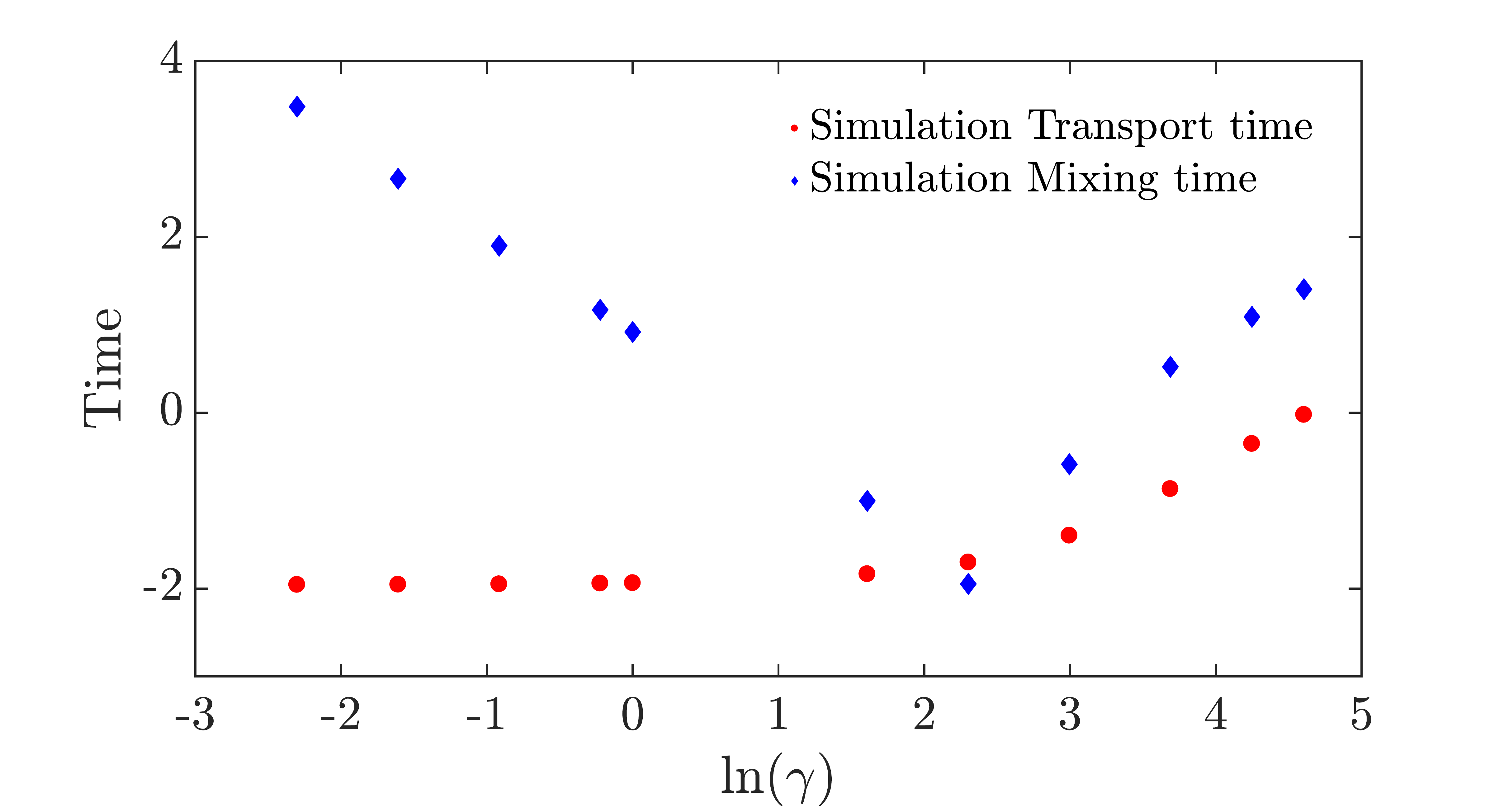}
  \subcaption{$A=10,F=100$}
  \label{fig:mixing_transport_high_F}
\end{minipage}
\caption{Comparison of transport and mixing times. Transport time dominates the mixing time for low force at high friction.}
\label{fig:mixing_transport_combined}
\end{figure}


Both erasing and reliability times exhibit a trade-off in friction, being  minimised by intermediate values. This fact sets up a second trade-off between designing bits with extreme values of friction to optimise reliability, or moderate values of friction to optimise erasing. The consequences of this secondary trade-off will be explored in Section~\ref{geometry}.

\subsubsection{Additional dependencies of the erasing time}
\label{subsubsec:more_erasing}
A larger value of $A$ implies a steeper descent into the target left-hand well, making mixing faster. We therefore expect that the mixing time and hence the erasing time monotonically decreases with $A$.

\begin{observation}\label{observ:erasing_decreasing_height}
The erasing time is a strictly decreasing function of well height $A$ at fixed $F$, $\gamma$. This can be seen from the analytic expressions of erasing time (Equations~\ref{eq:erasing_low} and ~\ref{eq:erasing_high}) backed up with numerical simulations (Figure~\ref{fig:erasing_height}).
\end{observation}

\begin{figure}[h!]
\begin{minipage}{0.5\textwidth}
  \centering
  \includegraphics[scale=0.17]{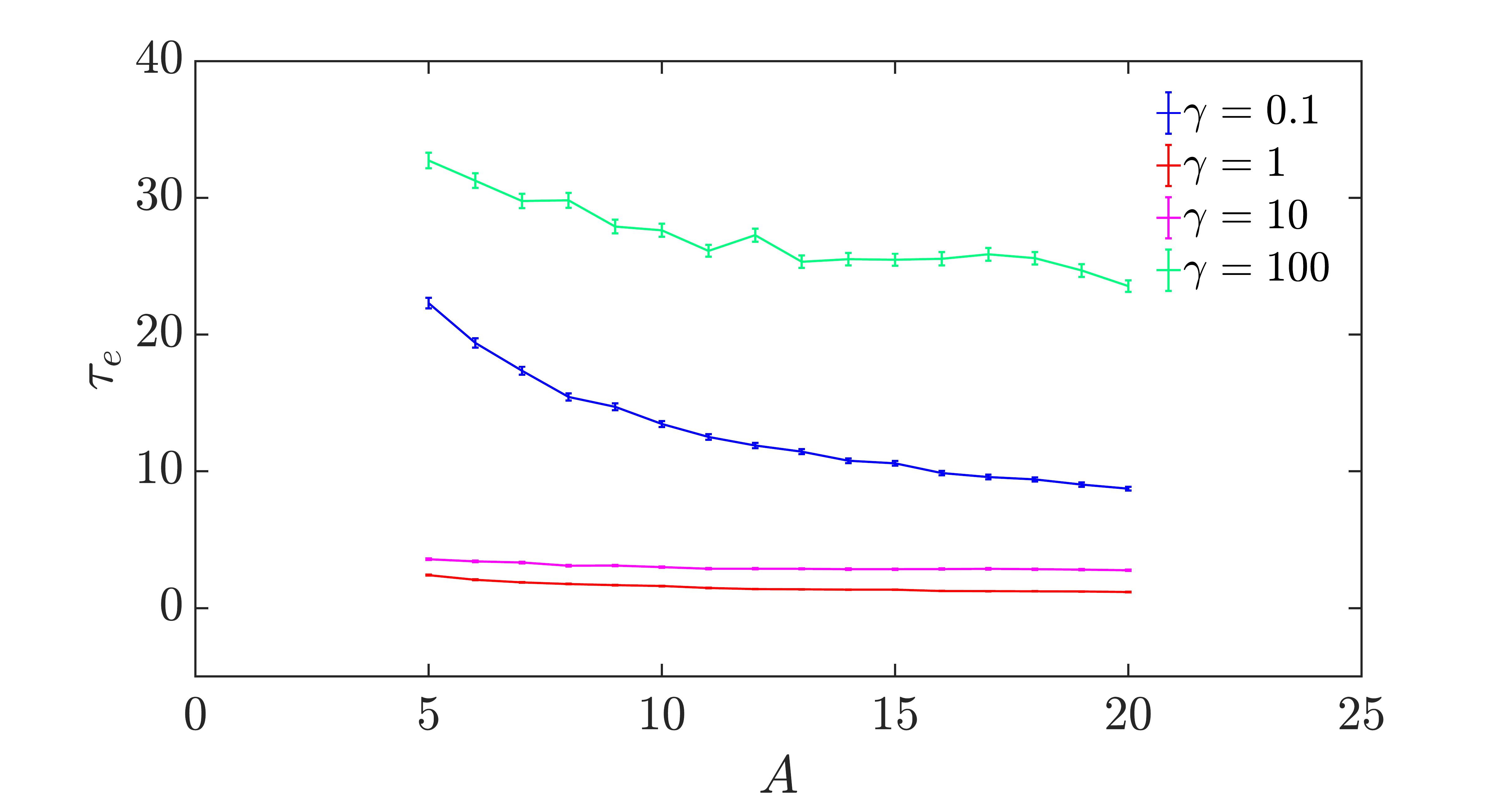}
  \subcaption{$F=5$}
\end{minipage}
\begin{minipage}{0.5\textwidth}
  \centering
  \includegraphics[scale=0.17]{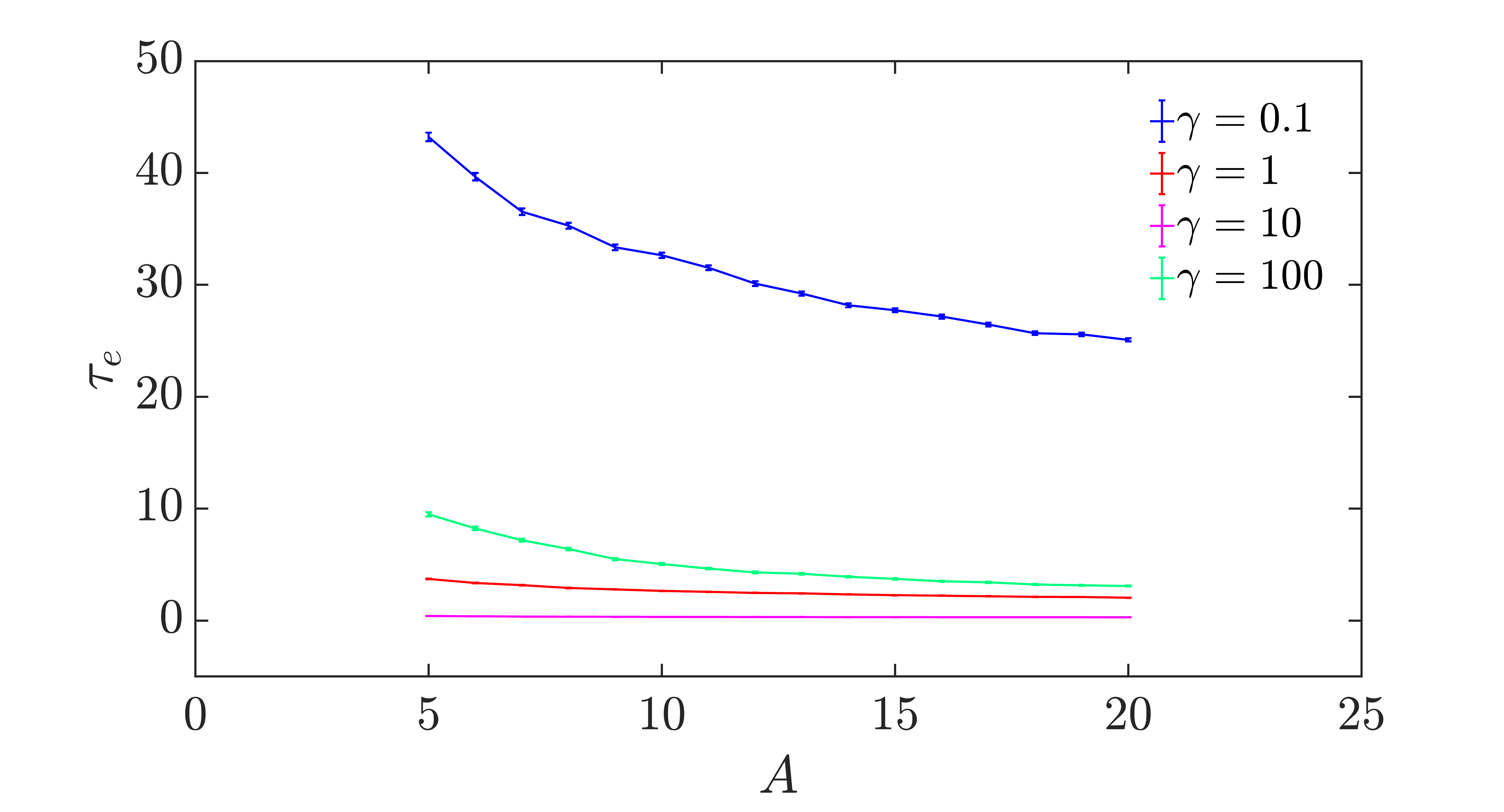}
  \subcaption{$F=100$}
\end{minipage}
\caption{Evidence that the erasing time is a strictly decreasing function of well height across a range of $F$ and $\gamma$. Other values of $F$ and $\gamma$ show similar behaviour.}
\label{fig:erasing_height}
\end{figure}

By contrast, erasing time shows a non-monotonic dependence on $F$ at fixed $A$, $\gamma$. Applying too little force leads to slow transport, and doesn't effectively trap the particle within the target well. But applying too much force supplies the particle with too much energy, which must subsequently be lost during the mixing period. The fact that erasing time monotonically decreases with $A$ at fixed $F$ and $\gamma$, and shows a non-monotonic dependence on $F$ at fixed $A$ and $\gamma$, leads to non-monotonic dependence of $\tau_e$ on $F$ at fixed $W=A+F$ and $\gamma$. We illustrate this non-monotonicity in Figure~\ref{fig:non_monotonic_F_fixed_work}, in which simple regression formulae have been fitted to the simulation data to enable interpolation at fixed $W$ and $\gamma$ (see Section~\ref{subsec:regression} of the Supplementary Information). As friction increases, the force required to provide the particle with excess energy increases, leading to minima at higher values of $F$.
 
\begin{figure}[h!]
\begin{minipage}{0.5\textwidth}
  \centering
  \includegraphics[scale=0.17]{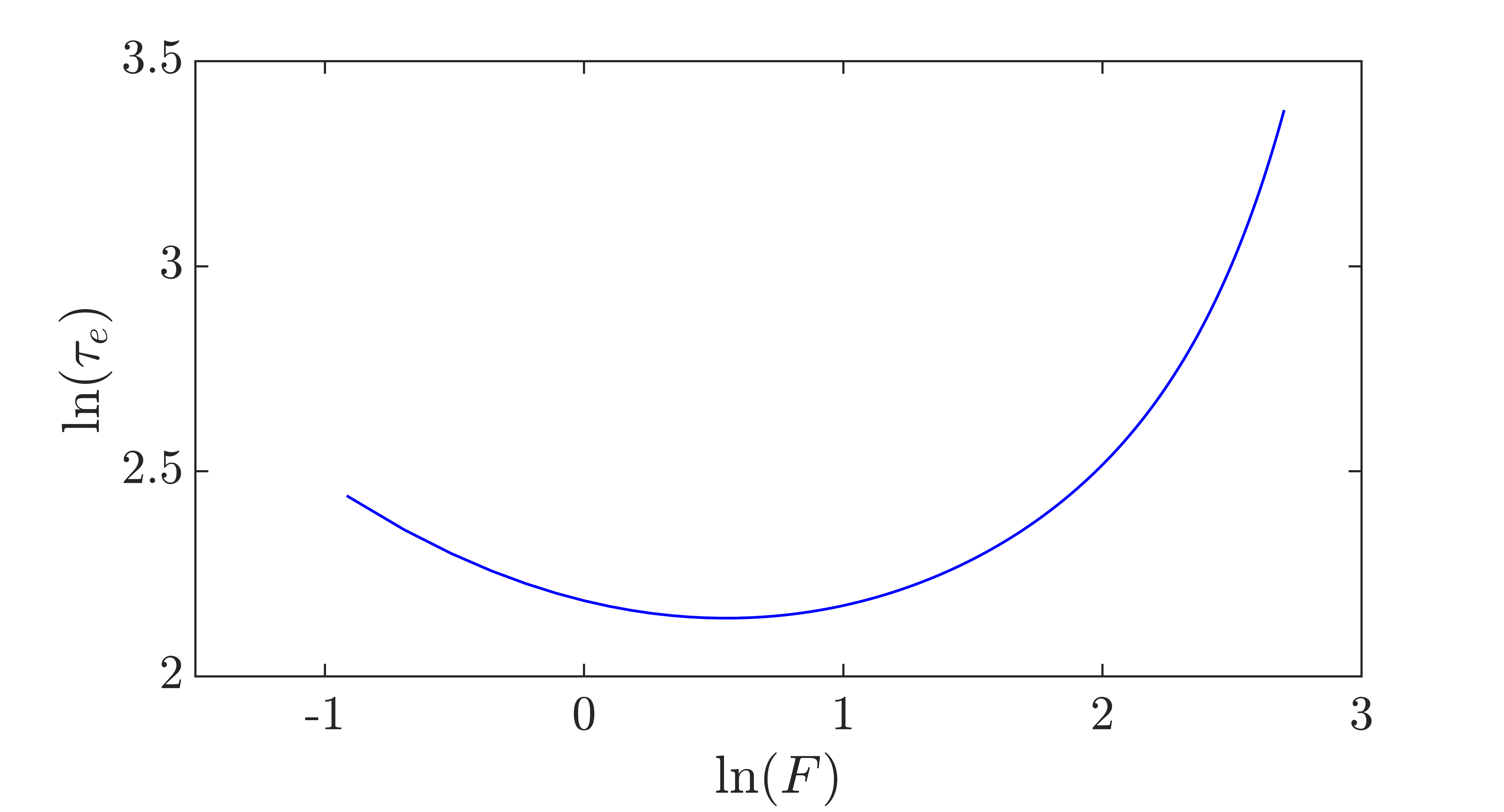}
  \subcaption{$\gamma=0.1$}
  \label{fig:erasing_work_low_friction_1}
\end{minipage}%
\begin{minipage}{0.5\textwidth}
  \centering
 \includegraphics[scale=0.17]{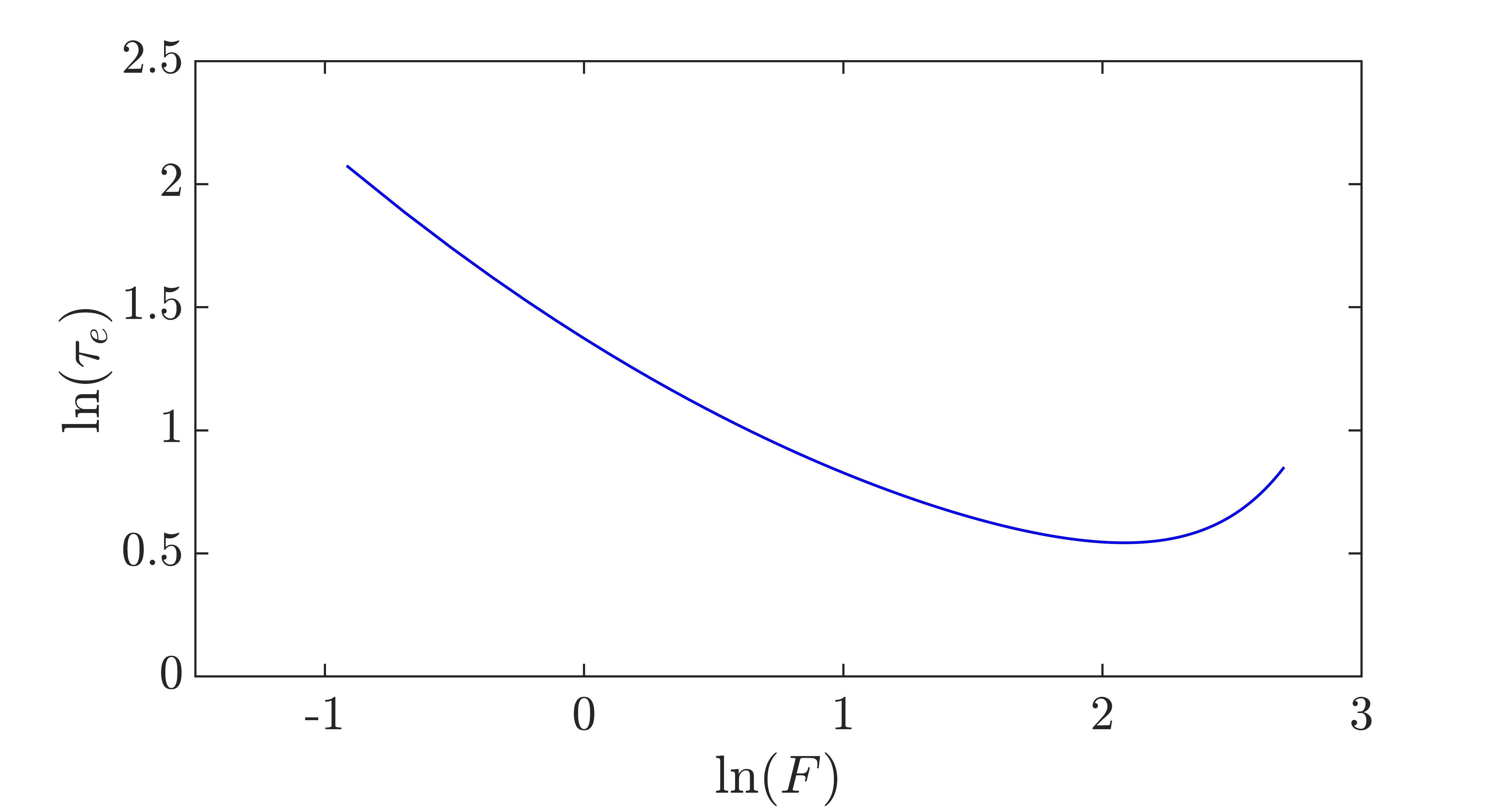}
  \subcaption{$\gamma=1$}
   \label{fig:erasing_work_low_friction_2}
\end{minipage}
\centering
\begin{minipage}{0.5\textwidth}
  \centering
  \includegraphics[scale=0.17]{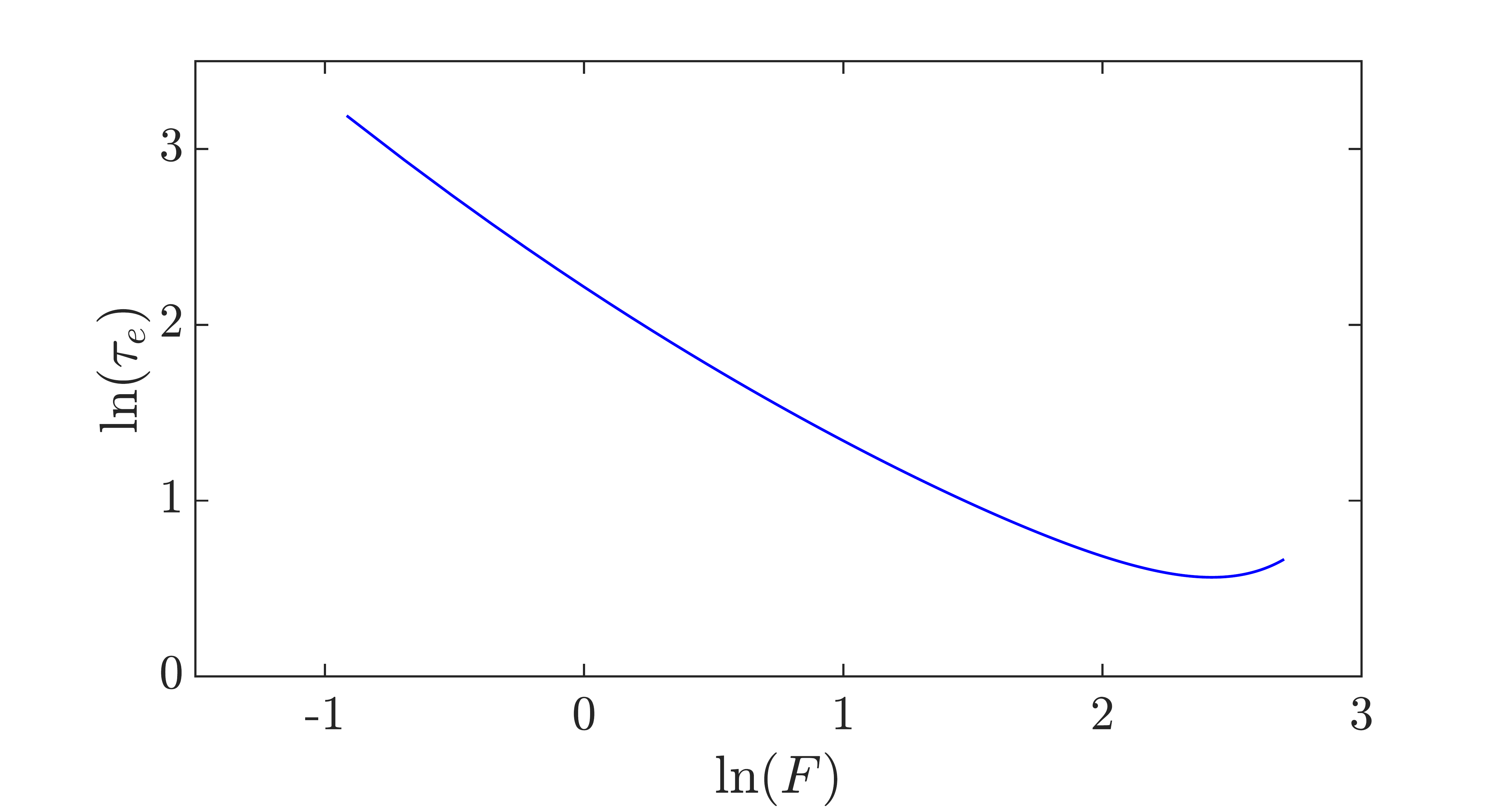}
  \subcaption{$\gamma=10$}
   \label{fig:erasing_work_high_friction_1}
\end{minipage}%
\begin{minipage}{0.5\textwidth}
  \centering
  \includegraphics[scale=0.17]{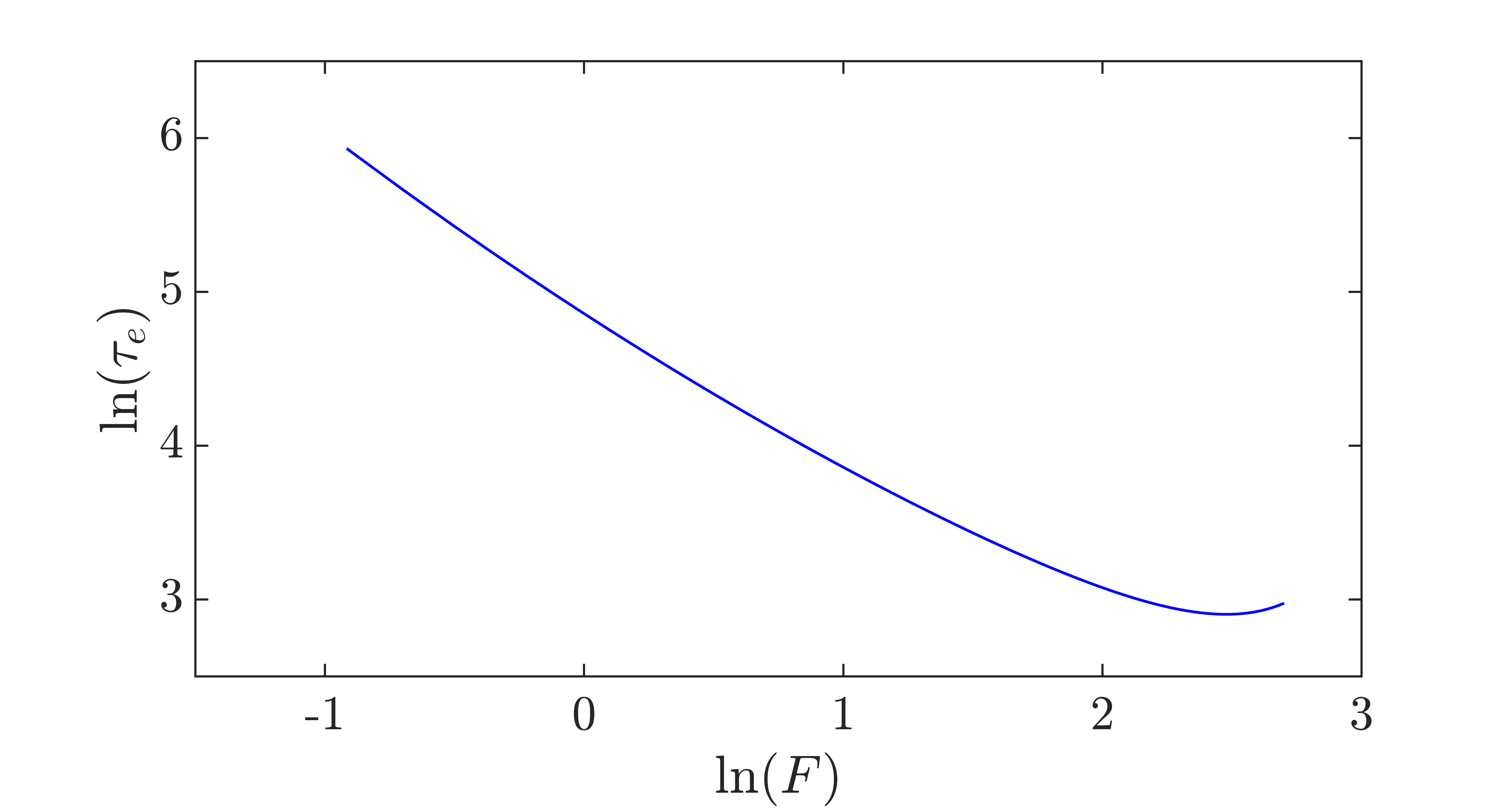}
  \subcaption{$\gamma=100$}
   \label{fig:erasing_work_high_friction_2}
\end{minipage}
\caption{For a fixed value of $W$ and $\gamma$, the erasing time is a non-monotonic function of $F$ and is minimum at moderate $F$. This is illustrated at work $W=20$ for various values of $\gamma$.}
\label{fig:non_monotonic_F_fixed_work}
\end{figure}

We make the following observation which will be used in the subsequent section.
\begin{observation}\label{observ:unique_local_minimum}
We have found no evidence of multiple local minima of erasing time in a level set of work for our control family. Refer to Section~\ref{single_minimum_work} of the Supplementary Information for characterstic plots showing the minima of erasing time in a level set of work. Physically this is unsurprising since the non-monotonicity in $\tau_e$ with $\gamma$ and $F$ mentioned above arise from fairly simple trade-offs, producing curves with single minima.  
\end{observation}

As with the reliability time, a more detailed analysis of the dependence of $\tau_e$ on other parameters, and even the shape of the control, is possible. However, these details are likely to be difficult to generalise, and are not necessary for the conclusions we draw in the subsequent sections.

\section{Design of Bits}\label{geometry}
We are now ready to study the question of how to design good bits. A design involves choosing parameters $A,F,\gamma$ for a bit to satisfy requirement specifications in terms of speed of erasing and reliability, without expending more work than required.  The most general formulation of our problem would require us to also allow the length scale $B$, the temperature $T$ and the mass $m$ to vary, as well as allowing arbitrary controls. Such a formulation would appear to make the problem even more challenging, so it seems prudent in a first analysis to restrict  our analysis to the variables $A,F$, and $\gamma$. Our restricted analysis is not without value since the underlying technology in any given construction typically does not allow arbitrary variation. Our numerical analysis with Example~\ref{ex:control_potential} will guide us in our assumptions and analysis, but our results will hold in greater generality. We will construct our proofs based on general assumptions, and subsequently explain how these assumptions are met by our control family.

We introduce the following terms.
\begin{enumerate}
\item The design of a bit is completely specified by the \textbf{design triple} $\left(A,F,\gamma\right)$.  
\textbf{Design Space} $\left(\DS\right)$ is the space of all design triples $\left(A,F,\gamma\right)$.
\item A \textbf{requirement specification} is a tuple $(t_r, t_e)\in\mathbb{R}^2_{>0}$ denoting the reliability and erasing time that we require of the bit. \textbf{Requirements Space} $\left(\RS\right)$ is the space of all requirement specifications.
\item \textbf{Erasing time} $\tau_e: \DS\to\mathbb{R}_{>0}$ takes a design triple $(A,F,\gamma)$ to the time required for erasing the corresponding bit under the control protocol specified by $F$. \textbf{Reliability time} $\tau_r : \DS \to \mathbb{R}_{>0}$ takes a design triple $(A, F, \gamma)$ to the reliability time of the corresponding bit. Note that $\tau_r$ is constant as a function of $F$ since it is a property of the dynamics in the absence of control.
\item \textbf{Work} $W:\DS\to\mathbb{R}_{>0}$ represents the expected work done by the control in erasing the corresponding bit. We will assume that $W$ is constant as a function of $\gamma$, as is the case in Example~\ref{ex:control_potential}. 
\item A design $(A,F,\gamma)$ is \textbf{feasible} for a requirement $(t_r, t_e)$ iff both $\tau_r(A,F,\gamma)\geq t_r$ and $\tau_e(A,F,\gamma)\leq t_e$. A $(t_r,t_e)$-feasible design $(A,F,\gamma)$ is $(t_r,t_e)$-\textbf{optimal} iff the work $W(A,F,\gamma)$ is minimum among all $(t_r,t_e)$-feasible designs. 
\item Inspired by the observation that non-trivial minima of erasing time at fixed work exist for our family of protocols (Section~\ref{subsubsec:more_erasing}), we define the notion of \textbf{trapped} bits. A design $(A,F,\gamma)$ is \textbf{trapped} iff for all designs $(A',F',\gamma')$ with $W(A,F,\gamma)=W(A',F',\gamma')$, the erasing time $\tau_e(A,F,\gamma) \leq \tau_e(A',F',\gamma')$. 
 A design $(A,F,\gamma)$ is \textbf{uniquely trapped} iff for all designs $(A',F',\gamma')$ with $W(A,F,\gamma)=W(A',F',\gamma')$, the erasing time $\tau_e(A,F,\gamma)\leq\tau_e(A',F',\gamma')$ with equality iff $(A,F,\gamma)=(A',F',\gamma')$. A design $(A,F,\gamma)$ is \textbf{locally trapped} iff there exists a neighbourhood of $(A,F,\gamma)$ consisting of bits $(A',F',\gamma')$ with $W(A,F,\gamma)=W(A',F',\gamma')$ such that the erasing time $\tau_e(A,F,\gamma)\leq\tau_e(A',F',\gamma')$. More informally, a trapped design has the lowest erasing time within a level set of work; a  trapped design is unique if it is the {\em only} design within that level set of work to have the minimal erasing time; and a locally trapped design has the minimal erasing time within a local neighbourhood of designs of equal work.
\item A requirement specification $(t_r,t_e)$ is \textbf{unsaturated} iff there exists a $(t_r,t_e)$-optimal design $(A,F,\gamma)$ such that either $\tau_r(A,F,\gamma) > t_r$ or $\tau_e(A,F,\gamma) < t_e$. A feasible requirement specification that is not unsaturated is called \textbf{saturated}.

\end{enumerate}

Throughout this section, we will assume that $\tau_e,\tau_r$, and $W$ are continuous functions. \\ 

We will state the main results related to the properties of the optimal design leaving the detailed proofs to the Supplementary Information. We first claim that an optimal design always saturates the bound on the erasing time constraint. Further, if the optimal bit is not locally trapped, then it also saturates the bound on the reliability time constraint.

\begin{claim}[Saturation of timescales]\label{lem:improvee}
Let us assume that it is possible to locally decrease work at fixed reliability time (This is generally possible since one can perturb the control parameters to reduce work; but reliability time does not depend on the control parameters). Fix requirement specifications $(t_r,t_e)\in\RS$. Suppose $(A,F,\gamma)$ is a $(t_r,t_e)$-optimal design. Then 
\begin{enumerate}
\item\label{erasing_tight} $\tau_e(A,F,\gamma) = t_e$. 
\item\label{reliability_tight} If the design $(A,F,\gamma)$ is not locally trapped, then $\tau_r(A,F,\gamma) = t_r$.
\end{enumerate}
\end{claim}

\begin{proof}
Refer Section~\ref{sec:geometry_optimal_bit} in the Supplementary Information.
\end{proof}

\begin{figure}[h!]
\begin{center}
\[
\includegraphics[scale=0.45]{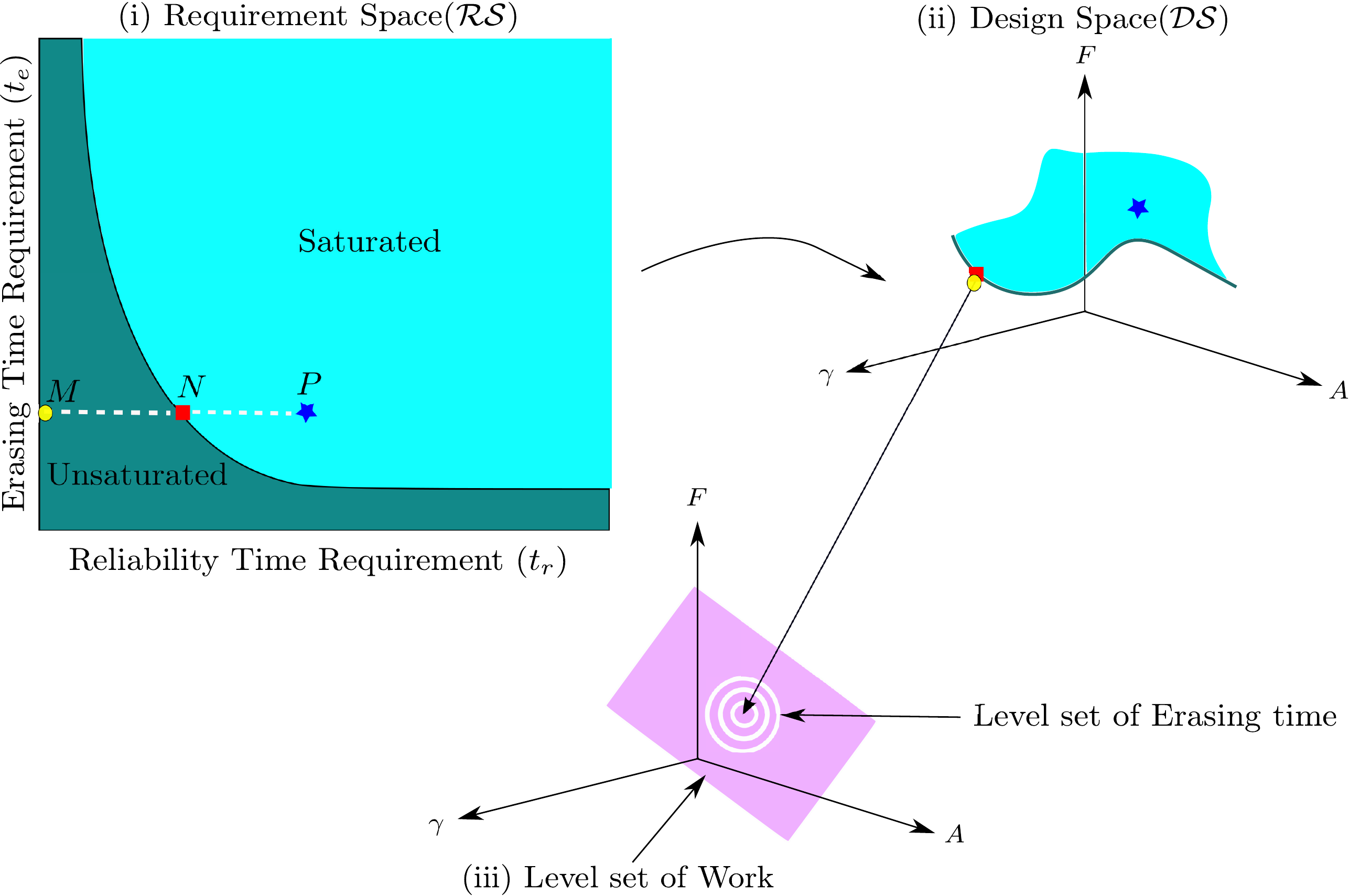}
\]
\end{center}
\caption{An illustration of the mapping of requirement specifications to optimal designs.  The design space is divided by a curve corresponding to the reliability and erasing times of trapped designs. Points $M,N$ in (i) Requirement Space$(\RS)$ having the same erasing time requirement get mapped to the same optimal bit in (ii) Design Space$(\DS)$:  a trapped design with $\tau_r$ and $\tau_e$ equal to the requirements at $N$. The requirement specifications represented by points like $P$ having the same $t_e$ but greater $t_r$ than $N$ are mapped to distinct points in design space. (iii) A representation of a level set of work $W$ within $\DS$, illustrating that the optimal designs to which unsaturated requirements are mapped minimize erasing time among all designs requiring the same work.}
\label{fig:optimal_bit_manifold}
\end{figure}

The next claim provides insight into the geometry of optimal designs. In particular, it states that under mild assumptions the requirement space is divided into two regions by a boundary given by the reliability and erasing times of trapped designs. Requirements with $t_r<t_r^\prime$ and $t_e=t_e^\prime$, where $(t_r^\prime, t_e^\prime)$ is a requirement on the dividing line,  are unsaturated, while other requirement specifications are saturated.

\begin{claim}[Saturated and Unsaturated Requirements]\label{claim2}
Assume that the erasing time of trapped designs is a strictly decreasing function of the work (refer to Observation~\ref{observ:decreasing_work} in the Supplementary Information for a justification), and that as before it is always possible to decrease work  at fixed reliability time. Let $(A^*,F^*,\gamma^*)$ be a trapped design such that $\tau_e(A^*,F^*,\gamma^*)=t_e$.
\begin{enumerate}
\item\label{same_optimal_bit} If $t_r\leq \tau_r(A^*,F^*,\gamma^*)$ then $(A^*,F^*,\gamma^*)$ is $(t_r,t_e)$-optimal.
\item\label{unsaturated} If $t_r < \tau_r(A^*,F^*,\gamma^*)$ then $(t_r,t_e)$ is unsaturated.
\item\label{saturated} Make the additional assumption that locally trapped designs are uniquely trapped (as noted for our family of protocols (Example~\ref{ex:control_potential}) in Observation~\ref{observ:unique_local_minimum}). 

If $t_r\geq \tau_r(A^*,F^*,\gamma^*)$, then $(t_r,t_e)$ is saturated.
\end{enumerate}
\end{claim}

\begin{proof}
Refer Section~\ref{sec:geometry_optimal_bit} in the Supplementary Information.
\end{proof}

The claims about saturation/unsaturation of times-scales can also be proved using KKT conditions [Refer Section~\ref{subsec:KKT_conditions} in the Supplementary Information], a standard tool from optimization theory. 

A more intuitive picture of the results can be understood from Figure~\ref{fig:optimal_bit_manifold}. In this figure, we illustrate how finding an optimal design subject to a specification  maps a point in the requirement space to a point in the design space. For a trapped design $(A^*,F^*,\gamma^*)$, requirements with $t_r < \tau_r(A^*,F^*,\gamma^*)$ and $t_e = \tau_e(A^*,F^*,\gamma^*)$ are unsaturated and get mapped to the same design $(A^*,F^*,\gamma^*)$ (claims~\ref{claim2}.~\ref{unsaturated} and~\ref{claim2}.~\ref{same_optimal_bit}). If the design $(A^*,F^*,\gamma^*)$ is uniquely trapped, then requirements with $t_r \geq\tau_r(A^*,F^*,\gamma^*)$ and $t_e=\tau_e(A^*,F^*,\gamma^*)$ are saturated (claim~\ref{claim2}.~\ref{saturated}).

\begin{figure}[h!]
\begin{minipage}{0.5\textwidth}
  \centering
  \includegraphics[scale=0.17]{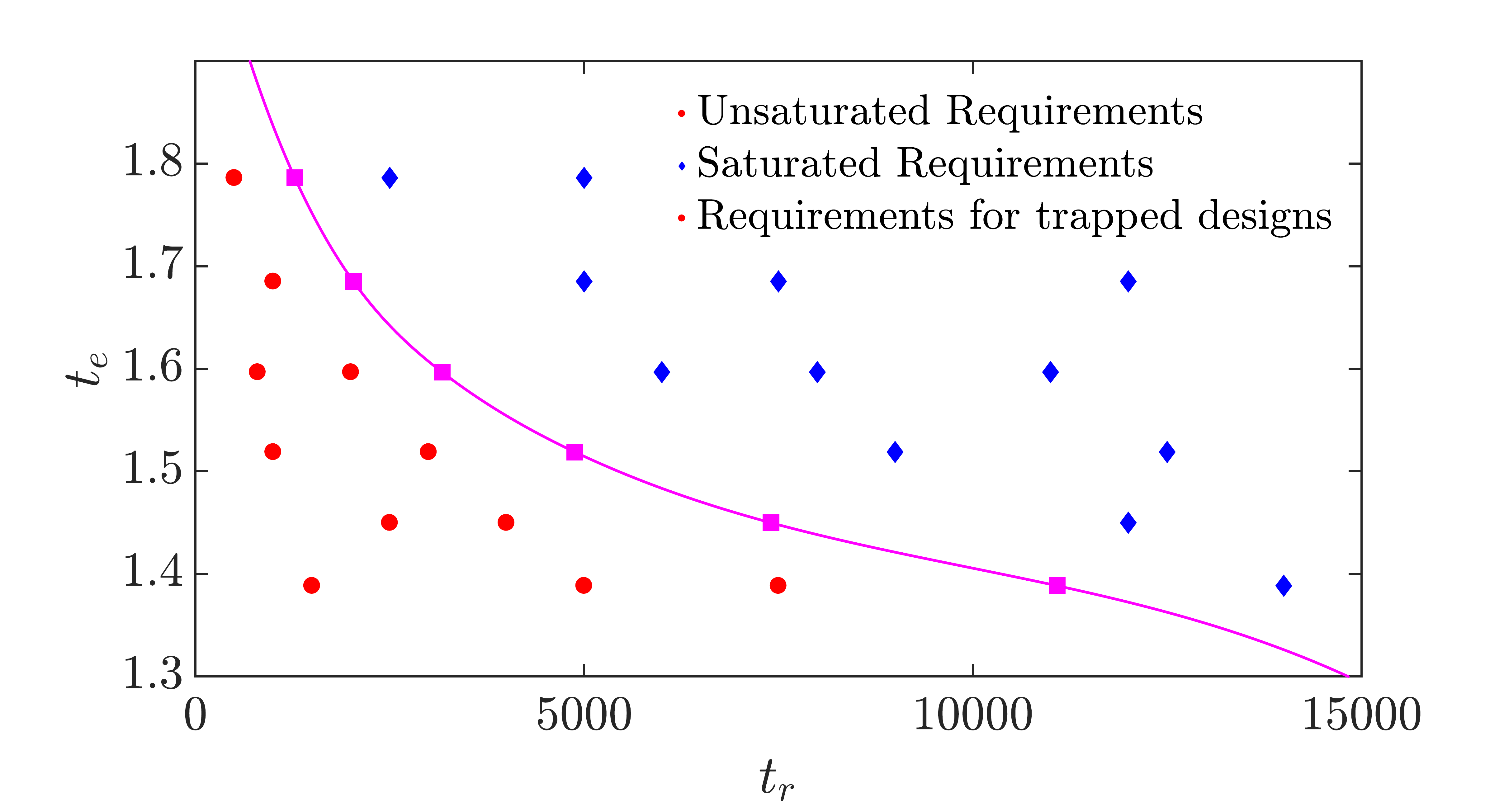}
   \subcaption{}
\end{minipage}
\begin{minipage}{0.5\textwidth}
  \centering
  \includegraphics[scale=0.17]{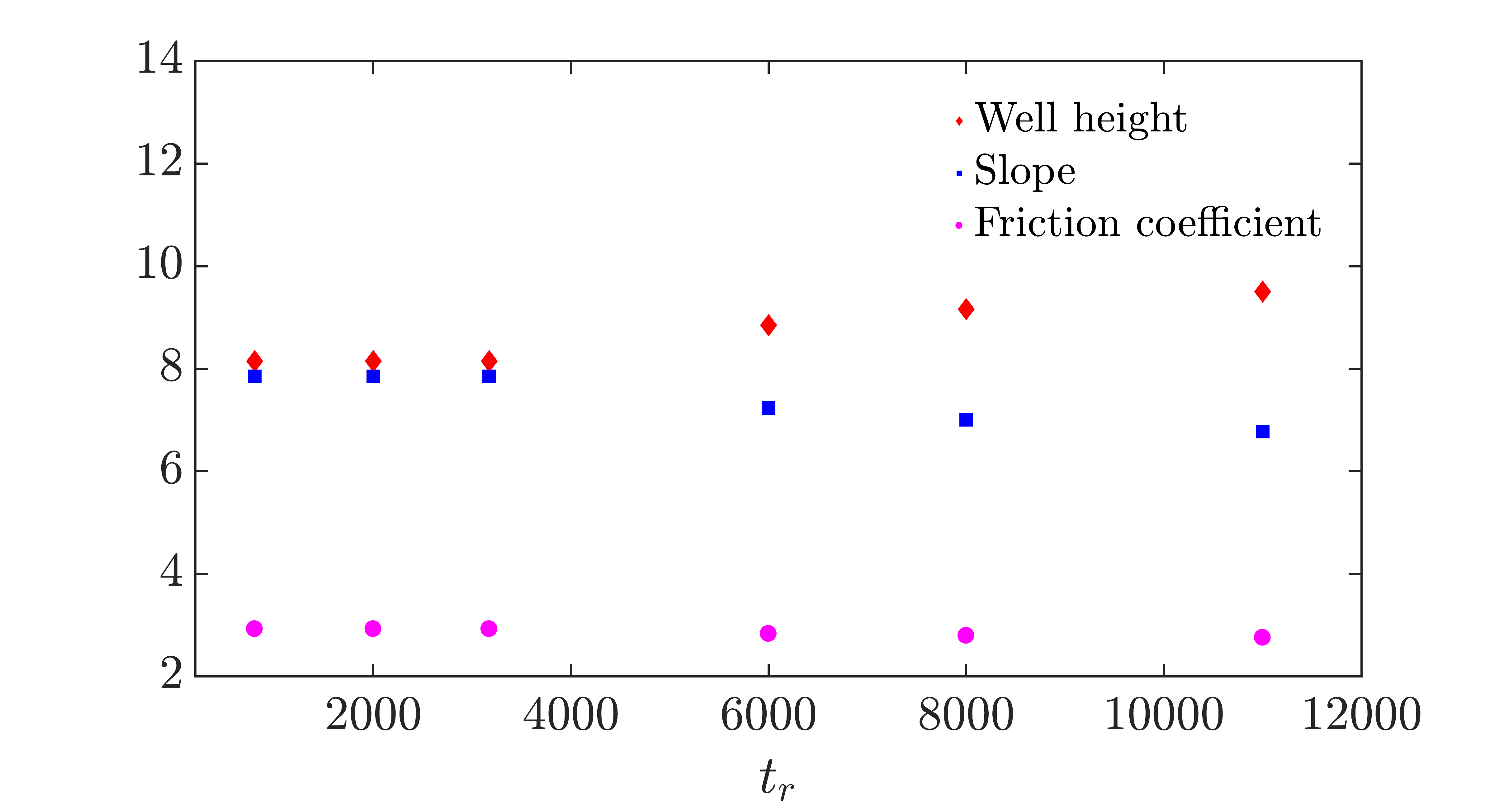}
  \subcaption{}
\end{minipage}
\caption{Illustration of the division of Requirement Space$(\RS)$ into saturated and unsaturated regions by requirements that correspond to trapped designs. (a) Squares show requirements $(t_e,t_r)$ that are saturated by trapped designs for the family of protocols we consider. Numerical optimization shows that requirements to the left of the locus defined by these points are unsaturated (circles), whereas requirements to the right are saturated (diamonds). (b) A plot of the optimal designs for points from (a) at $t_e =1.5967$. It is clear that for requirements $t_r \leq 3173$, lying to the left of the trapped-design locus in (a), optimal design parameters are identical whereas they are distinct for $t_r > 3173$.}
\label{fig:satunsat}
\end{figure}

Figure~\ref{fig:satunsat} illustrates these results for our example family of controls (Example~\ref{ex:control_potential}). As discussed in Section~\ref{subsec:regression} of the Supplementary Information, we have implemented simple regression to fit the functions $\tau_e(.)$ and $\tau_r(.)$ to our simulation results. We then identified trapped designs using numerical minimisation, plotting the requirement specifications saturated by these designs. For each trapped bit $(A^*,F^*,\gamma^*)$, we randomly selected requirements with $t_e = \tau_e(A^*,F^*,\gamma^*)$, but with $t_r$ either greater than equal to or less than $\tau_r(A^*,F^*,\gamma^*)$, and used numerical optimization techniques to search for the optimal designs. The results support our analysis; requirements with $t_r<\tau_r(A^*,F^*,\gamma^*)$ are unsaturated, and those with $t_r\geq\tau_r(A^*,F^*,\gamma^*)$ are saturated. Furthermore, as we show in Figure~\ref{fig:satunsat}\,(b), unsaturated requirements at fixed $t_e$ all map to the same trapped design.

\subsection{Optimal friction for simple controls}\label{subsec:forbidden_friction}

In Section~\ref{sec:times}, we demonstrated that both reliability and erasing times are non-monotonic in friction, with short erasing times favoured by moderate values of friction, and long reliability times favoured by extreme values. In what follows, we give a precise quantification of the resultant trade-off in finding the friction of an optimal bit. The analysis is significantly simplified for our family of controls, in which work is independent of the friction coefficient. \\

Let us introduce the following terms. Fix an $A$ and $F$. Then,
\begin{enumerate}
\item $\gamma^e_{\rm crit}$ is the friction coefficient that minimizes erasing time as a function of friction coefficient $\gamma$ at fixed $A$ and $F$, i.e., for all $\gamma'\in\mathbb{R}_{>0}$, we have:
\begin{align}
\tau_e(A,F,\gamma^e_{\rm crit}) \leq \tau_e(A,F,\gamma'). 
\end{align}
We call the design $(A,F,\gamma^e_{\rm crit})$ \textbf{critically damped}.
\item $\gamma^r_{\rm crit}$ is the friction coefficient that minimizes reliability time as a function of friction coefficient $\gamma$ at fixed $A$ and $F$, i.e., for all $\gamma'\in\mathbb{R}_{>0}$, we have:
\begin{align}
\tau_r(A,F,\gamma^r_{\rm crit})\leq\tau_r(A,F,\gamma'). 
\end{align}
\end{enumerate}

It is easy to note that trapped bits are also critically damped. In Figure~\ref{fig:optimal_friction} we show illustrative curves of the erasing and reliability times as a function of friction coefficient $\gamma$ at fixed $A,F$. These curves have single minima at $\gamma^e_{\rm crit}$ and $\gamma^r_{\rm crit}$, respectively. Also shown on this graphs are regions of friction space that can be eliminated from consideration for optimal bits. To eliminate extreme values of friction, we note that the design must have a minimal finite $A$ to be a well-defined two-state system in the resting state. For our bit, it is $A_{\rm min} \approx 3$. In the next claim we precisely describe which regions of friction can be eliminated.

\begin{figure}[h!]
\begin{center}
\[
\includegraphics[scale=0.55]{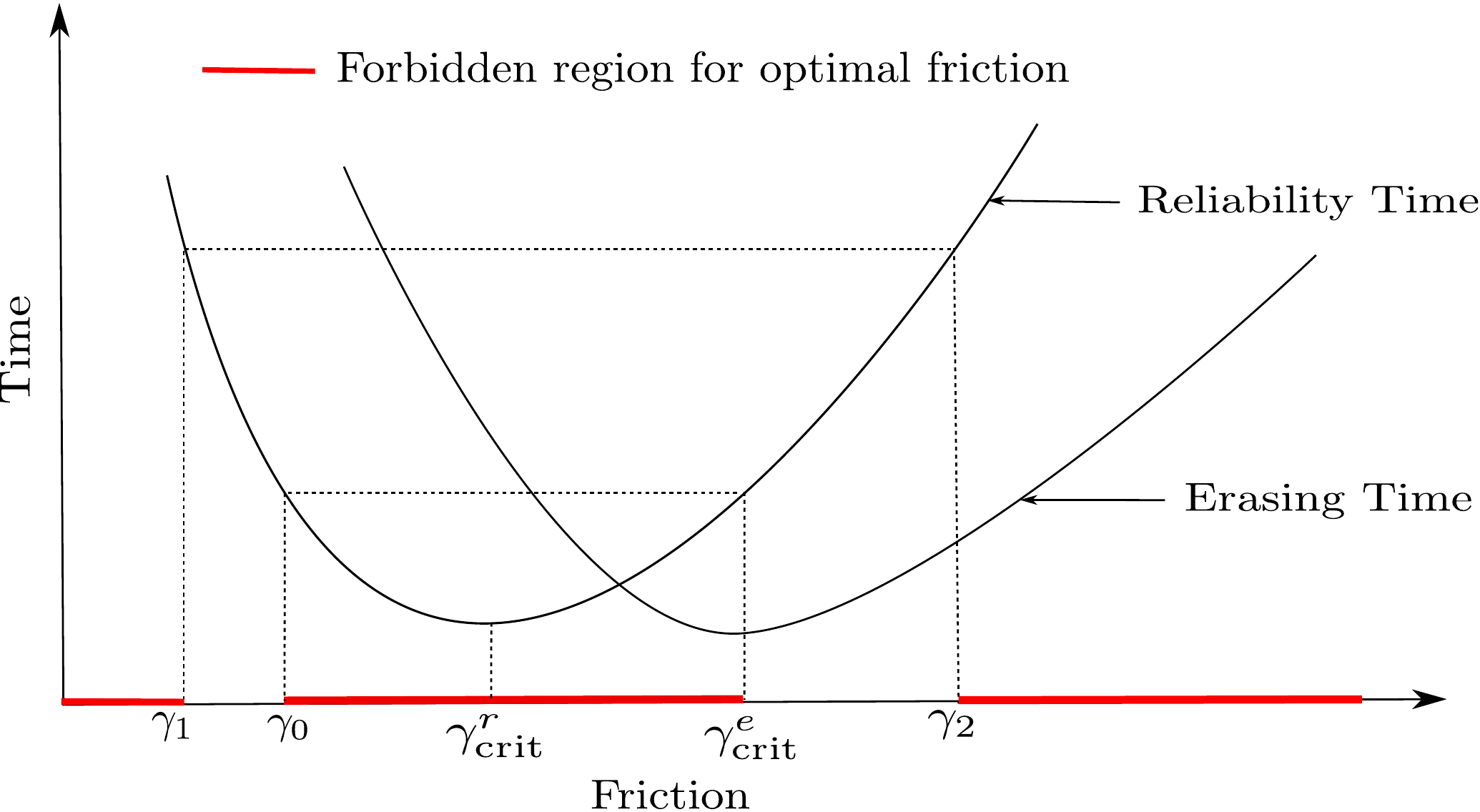}
\]
\end{center}
\caption{Regions of friction-space can be eliminated from the search for optimal bits for our class of controls. As a result, the optimal friction is either critical damping, or lies somewhere within two regions of moderate friction. Illustrative curves of $\tau_e$ and $\tau_r$ at fixed $A,F$ indicate these regions.}
\label{fig:optimal_friction}
\end{figure}

\begin{claim}\label{claim4}[Forbidden regions for optimal friction]
Assume that both $\tau_e$ and $\tau_r$ have a single-well-defined minimum and tend to infinity as $\gamma$ tends to zero or infinity. Let $(A,F,\gamma)$ be a $(t_r,t_e)$-optimal design. (Refer to Figure~\ref{fig:optimal_friction} for notational convenience)
\begin{enumerate}
\item Let $\gamma_0$ be such that $\tau_r(A,F,\gamma_0) = \tau_r(A,F,\gamma^e_{\rm crit})$. 
\begin{enumerate}
\item If $\gamma^e_{\rm crit} > \gamma^r_{\rm crit}$, then $\gamma\notin\left(\gamma_0,\gamma^e_{\rm crit}\right)$.
\item If $\gamma^e_{\rm crit} < \gamma^r_{\rm crit}$, then $\gamma\notin\left(\gamma^e_{\rm crit},\gamma_0\right)$.
\end{enumerate}
i.e. the friction of the optimal bit does not reside in the central red region in Figure~\ref{fig:optimal_friction}.
\item Let $A_{\rm min}$ be the minimum height for a bit to be meaningfully bistable and let $\gamma_1<\gamma_2$ be such that $\tau_r(A_{\rm min},F,\gamma_1)=\tau_r(A_{\rm min},F,\gamma_2)=t_r$. If $(A,F,\gamma)$ is not locally trapped, then $\gamma\notin(0,\gamma_1)\cup(\gamma_2,\infty)$.

i.e. the friction of the optimal bit does not arise from the extreme red regions in Figure~\ref{fig:optimal_friction}.

\end{enumerate}
\end{claim}
\begin{proof}
\begin{enumerate}
\item We prove it for the case when $\gamma^e_{\rm crit} > \gamma^r_{\rm crit}$, the other case proceeds in identical fashion. For contradiction, assume that $\gamma\in(\gamma_0,\gamma^e_{\rm crit})$. Then, due to the single minima in both $\tau_e$ and $\tau_r$, and the fact that $\tau_r$ tends to infinity as $\gamma$ tends to zero or infinity, there exists a design $(A,F,\gamma^\prime)$ with $\gamma^\prime > \gamma_0$ and $\tau_r(A,F,\gamma^\prime) = \tau_r(A,F,\gamma) \geq t_r$, but $\tau_e(A,F,\gamma^\prime) < \tau_e(A,F,\gamma) \leq t_e$. The design $(A,F,\gamma^\prime)$ is $(t_r,t_e)$-optimal since it is $(t_r,t_e)$-feasible and has $W(A,F,\gamma^\prime)=W(A,F,\gamma)$, contradicting Lemma~\ref{lem:improvee}.~\ref{erasing_tight} that the optimal bit saturates the bound on the erasing time constraint.
\item For contradiction, suppose that $\gamma <\gamma_1$ or $\gamma >\gamma_2$. Then since $A \geq A_{\rm min}$ and the reliability time increases with well height and more extreme values of $\gamma$, either $\tau_r(A,F,\gamma) \geq \tau_r(A_{\rm min},F,\gamma) > \tau_r(A_{\rm min},F,\gamma_{1}) = t_r$ or $\tau_r(A,F,\gamma) \geq \tau_r(A_{\rm min},F,\gamma) > \tau_r(A_{\rm min},F,\gamma_{2})=t_r$, contradicting claim~\ref{lem:improvee}.~\ref{reliability_tight} that an optimal design that is not locally trapped saturates the bound on the reliability time constraint.
\end{enumerate}
\end{proof}

For clarity, let us assume initially that $\gamma^e_{\rm crit} > \gamma^r_{\rm crit}$ (equivalent arguments hold for the alternative). We see that optimal designs reside either at $\gamma^e_{\rm crit}$, or lie within two regions at moderate friction, as illustrated in Figure~\ref{fig:optimal_friction}. Interestingly, one region is adjacent to $\gamma^e_{\rm crit}$, whereas the other is not. It is not easy to see how designs in one region ($\gamma_1 \leq \gamma \leq  \gamma_0$) as in Figure~\ref{fig:optimal_friction} can outperform those in the other region ($\gamma^e_{\rm crit} < \gamma \leq \gamma_2 )$. Indeed, when we performed numerical optimisation on the regression-based fits to our simulation data, we only observed optimal bits that are either critically damped or lie in the allowed region adjacent to critical damping. This is illustrated in Figure~\ref{fig:critical_friction}, where we plot the optimal friction as a function of erasing time requirement at fixed reliability time requirement, for two values of reliability time requirements. We also plot $\gamma^e_{\rm crit}$ and $\gamma^r_{\rm crit}$ for comparison. At low erasing time requirements, designs reside at $\gamma^e_{\rm crit}$. At slightly higher erasing time requirements, the designs become saturated and the optimal friction lies adjacent to $\gamma^e_{\rm crit}$ in the region $\gamma^e_{\rm crit} < \gamma \leq \gamma_2$. Eventually, $\gamma^e_{\rm crit}$ crosses $\gamma^r_{\rm crit}$. At the crossing point, we have $\gamma = \gamma^e_{\rm crit}=\gamma^r_{\rm crit}$. At higher values of erasing time requirements, $\gamma$ still occupies the region adjacent to $\gamma^e_{\rm crit}$, which is now $\gamma_1 \leq \gamma \leq\gamma^e_{\rm crit}< \gamma^r_{\rm crit}$.

\begin{figure}[h!]
\begin{minipage}{0.5\textwidth}
  \centering
  \includegraphics[scale=0.17]{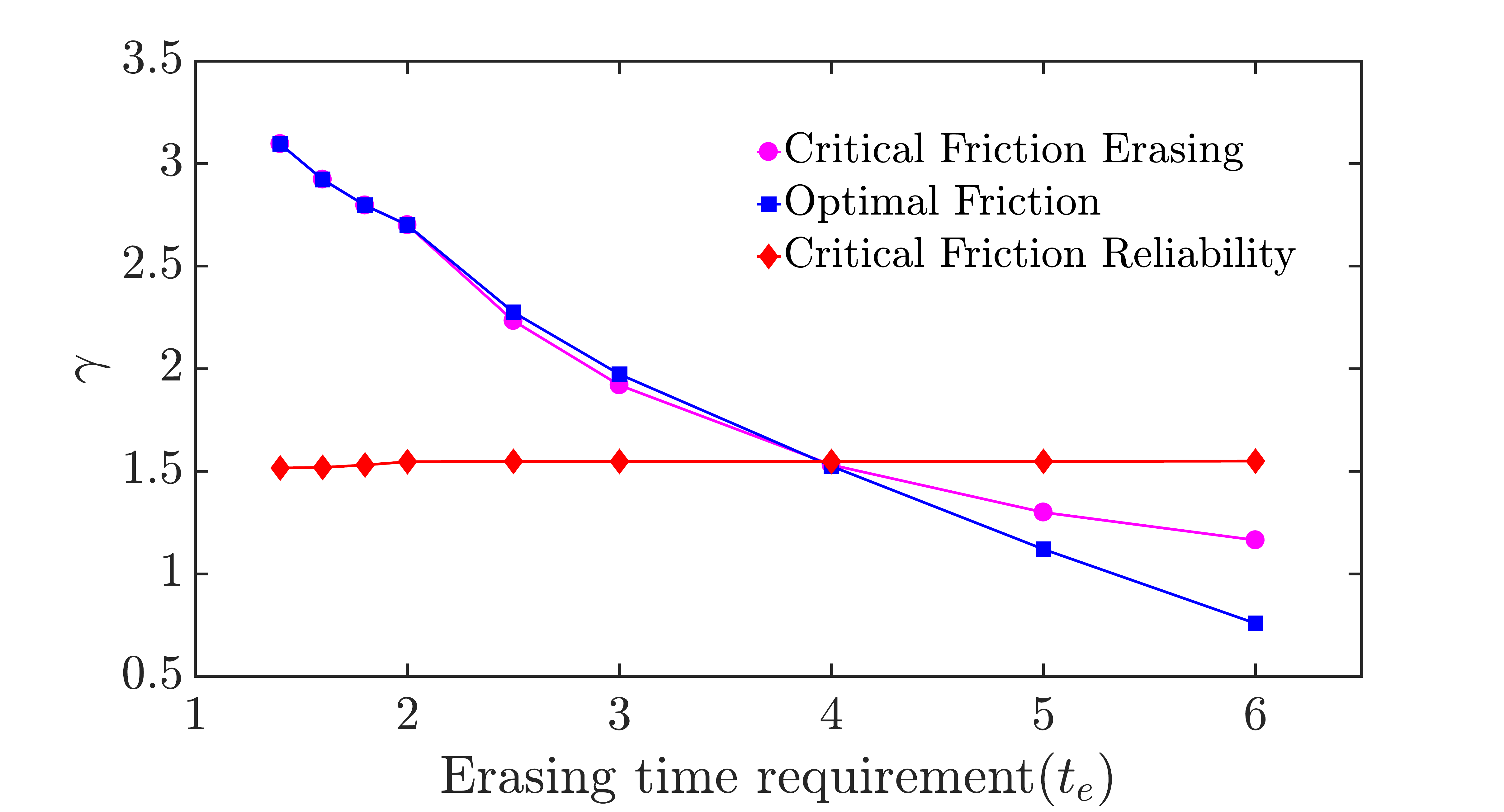}
  \subcaption{Reliability time requirement$(t_r)=500$}
\end{minipage}
\begin{minipage}{0.5\textwidth}
  \centering
  \includegraphics[scale=0.17]{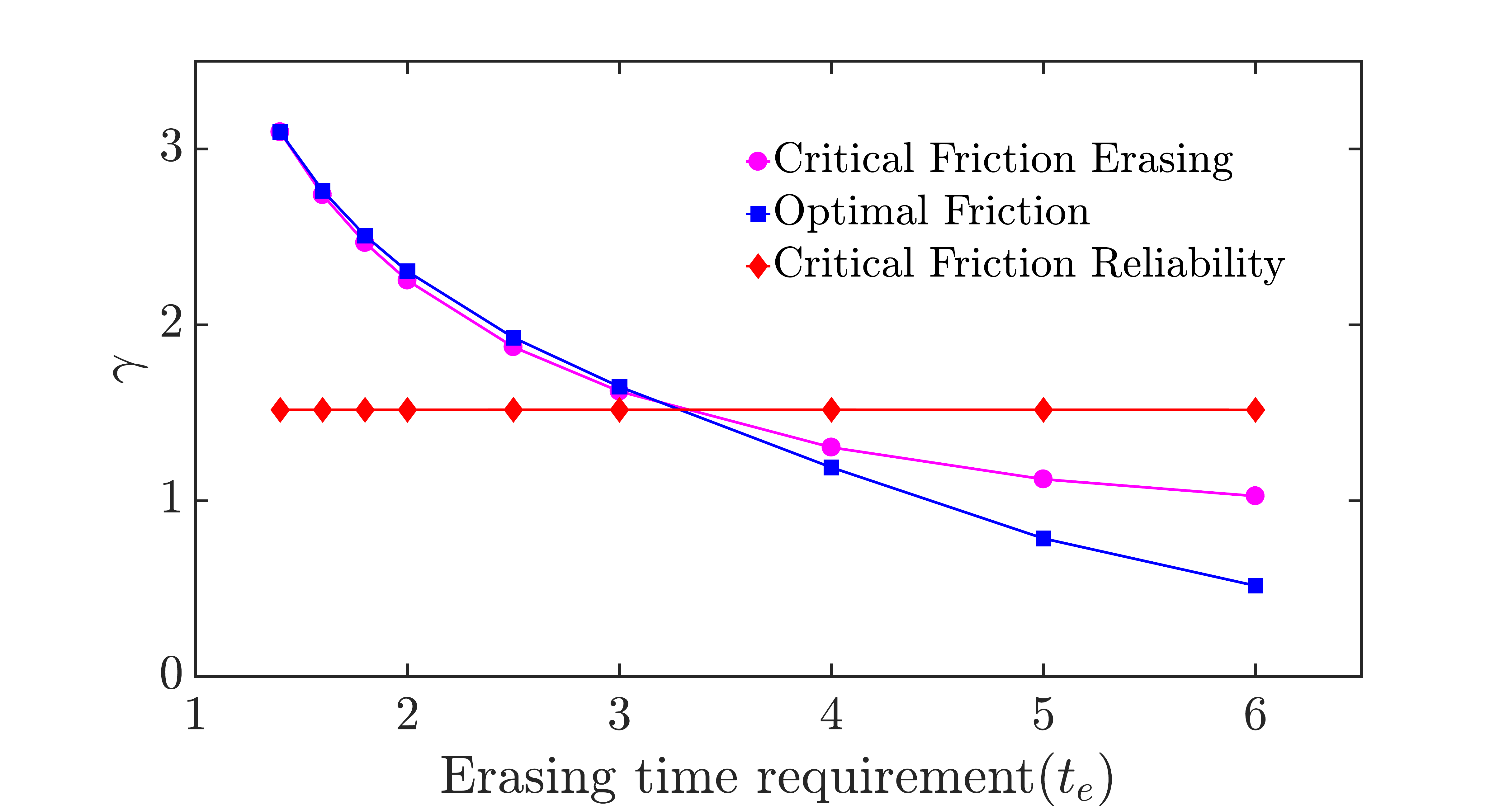}
  \subcaption{Reliability time requirement$(t_r)=10000$}
\end{minipage}
\caption{Optimal friction is either critical damping, or lies within a small region adjacent to critical damping, for our family of controls. We plot friction for optimal designs $(A,F,\gamma)$ against erasing time requirements $(t_e)$ for a fixed value of reliability time requirement $(t_r)$, alongside $\gamma^e_{\rm crit}$ and $\gamma^r_{\rm crit}$. Note that $A$ and $F$ are not fixed, but determined by the optimisation procedure alongside the optimal friction for each requirement $(t_r, t_e)$. The data was obtained from numerical optimisation and minimisation based on regression fits to simulation data.} 
\label{fig:critical_friction}
\end{figure}

\section{Conclusions}\label{conclusion}
We have explored the question of the design of optimal bits. Previously, authors have focused on designing optimal protocols that minimize work input when implementing a finite-time operation on a given system~\cite{Schmeidl2007,Then2008,Aurell2011,zulkowski2014optimal,Gingrich2016}.  
 Our approach differs in considering that bits need to have two distinct functionalities: retain data for long periods of time and allow rapid switching or erasing. Moreover we consider optimising over system parameters such as the intrinsic friction as well as the external control. Our fundamental observation is that friction plays a non-trivial role in the design of bits. Both switching/erasing and the eventual degradation of data involve relaxation towards equilibrium from a non-equilibrium distribution. This process is fastest at intermediate values of the friction, but slow in the overdamped and underdamped regimes. The best bit designs have high reliability times and low switching/erasing times, which implies an inherent trade-off in bit design between extreme values of friction that favour high reliability, and moderate values of friction that favour rapid switching or erasing. 

We have explored the consequences of the biphasic role of friction for a simple class of controls. The existence of non-trivial minima of erasing time in the level set of work leads to the generation of trapped designs. These designs are optimal for reliability requirements smaller than their own reliability time leading to unsaturated requirements. The result of the trade-off between extreme values of friction that maximize reliability time and moderate values of friction that minimise erasing times is that optimal designs are either critically damped or occupy a region of moderate friction close to critical damping. 

Our work opens up a new perspective on the design of efficient computational devices showing that: \emph{the best designs are likely to be neither underdamped nor overdamped}. This observation is particularly important as some authors have considered friction to be inherently problematic for computation~\cite{anacker1980josephson,buttiker1983thermal,klein1982thermal,likharev1982classical}. Equally, the role of friction is suppressed when bits are modelled as discrete two-state systems~\cite{landauer1961irreversibility,Ouldridge2015,gopalkrishnan2014cost}, since this approximation assumes rapid equilibration within the discrete states. 

We have only considered a simple family of controls to motivate our analysis and illustrate our findings. This family is not optimal - it was chosen for it's simplicity and ease of analysis. Moreover, there is some arbitrariness in the definition of both the erasing and reliability times. As such, the numerical details of the results obtained are not very important. We are not claiming to have derived numerical corrections to the minimal cost of erasing a bit, for example, or the specific work costs (substantially larger than $k_BT \ln 2$) which are not that informative. Rather, it is the qualitative results, which hold for a much broader class of controls which are important. The non-monotonic role of friction in both the erasing and reliability time-scales is a generic physical phenomenon that extends beyond the details of our implementation, and implies a competition between the goals of fast manipulation and long reliability times. Relatively weak assumptions -- that it is always possible to decrease work at fixed reliability time and that the minimal erasing time decreases with increased work imply that erasing time requirements are always saturated by optimal bits and that trapped designs lead to unsaturated reliability time requirements respectively. Other results rely more on the simplicity of the control family: the existence of only one local minimum of erasing time at fixed work simplifies the question of whether a requirement specification is saturated. The fact that work is independent of friction simplifies the task of eliminating certain values of friction as sub-optimal.   

Explicit exploration of a broader class of controls, including those with more complex variation over time, and varying parameters such as particle mass and distance between wells, are possible directions for future work. It is not immediately clear whether minima in erasing time at fixed work cost will become more or less prominent features of the optimisation landscape when the complexity of the system is increased, for example. In particular, raising or lowering the barrier between metastable states is a common idea~\cite{Ouldridge2015,zulkowski2014optimal,berut2012experimental,precision_feedback}. Lowering the barrier during erasing potentially allows for faster erasing at fixed reliability time and lower work cost. If said barriers could be raised and lowered arbitrarily far and quickly, it may be possible to circumvent any conflict between high reliability and low erasure time. However, real physical systems are not generally this flexible. Indeed, in order to apply a complex time-dependent control to a small colloid, experimenters typically use optical feedback traps ~\cite{berut2012experimental,precision_feedback}, which are not true potentials and rely on the continuous input of energy to apply forces and perform feedback control. For true physical protocols that permit finite raising and lowering of barriers between metastable states, we expect that our findings would still apply to a family of protocols with optimal barrier manipulation. An alternative direction would be to consider similar effects in systems with inherently quantum mechanical behaviour. 

\section{Data accessibility}
The codes and computational data for this manuscript can be downloaded from \url{https://www.imperial.ac.uk/principles-of-biomolecular-systems/contact--obtain-code-and-data/}

\section{Competing interests}
The authors declare no competing interests.

\section{Authors' contributions}
MG conceived the project. AD performed the calculations. AD, MG, TEO and NSJ planned the research, analysed the results and wrote the paper.

\section{Acknowledgement}
We would like to acknowledge Grigoris Pavliotis, Jure Vogrinc, Leonard Adleman, Rahul Dandekar, Tridib Sadhu, Girish Varma, Sanjoy Mitter, Charles Bennett, Deepak Dhar, Vivek Borkar, Venkat Anantharam, David Soloveichik and Aditya Raghavan for useful discussions. 

\section{Funding}
AD is supported by the ROTH scholarship of the Department of Mathematics, Imperial College London. TEO is funded by a University Research Fellowship from the Royal Society.

\bibliographystyle{unsrt}
\bibliography{Bibliography}

\appendix

\section{Appendix}

\subsection{Validating the timestep of the integrator}\label{timestep_validation}

We validate the accuracy of our Langevin integrator by considering the dependence of thermodynamic expectations on the time step. We calculate the average potential and kinetic energies for a particle in a quadratic potential $W_{A,B}=A\left(\frac{x}{B}-1\right)^2$, a quadratic proxy for a single well of the quartic resting-state potential. We plot the results in Figure~\ref{fig:potential_kinetic_timestep} for a few representative values of the friction coefficient $\gamma = [0.1,1.10,100]$ and $A=10$. Each result is based on an average from $10$ simulations each of $5\times 10^8$ time steps. As is evident from the figure, a time step of $0.001$ gives good convergence to the equipartition limit of $k_BT/2$.

\begin{figure}[h!]
\begin{minipage}{.5\linewidth}
\begin{subfigure}{.5\textwidth}
\includegraphics[width=\linewidth]{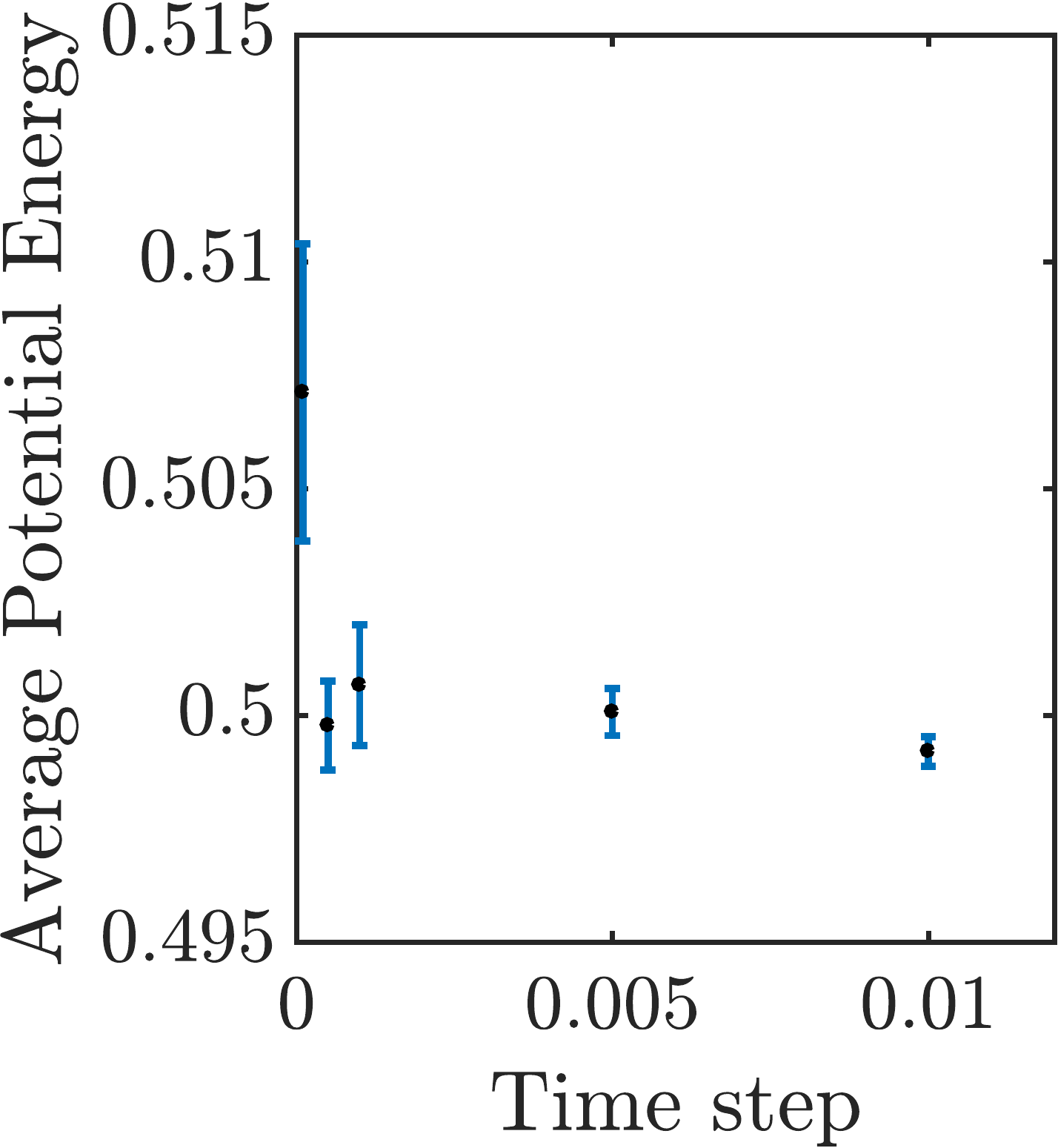}
\end{subfigure}%
\begin{subfigure}{.5\textwidth}
\includegraphics[width=\linewidth]{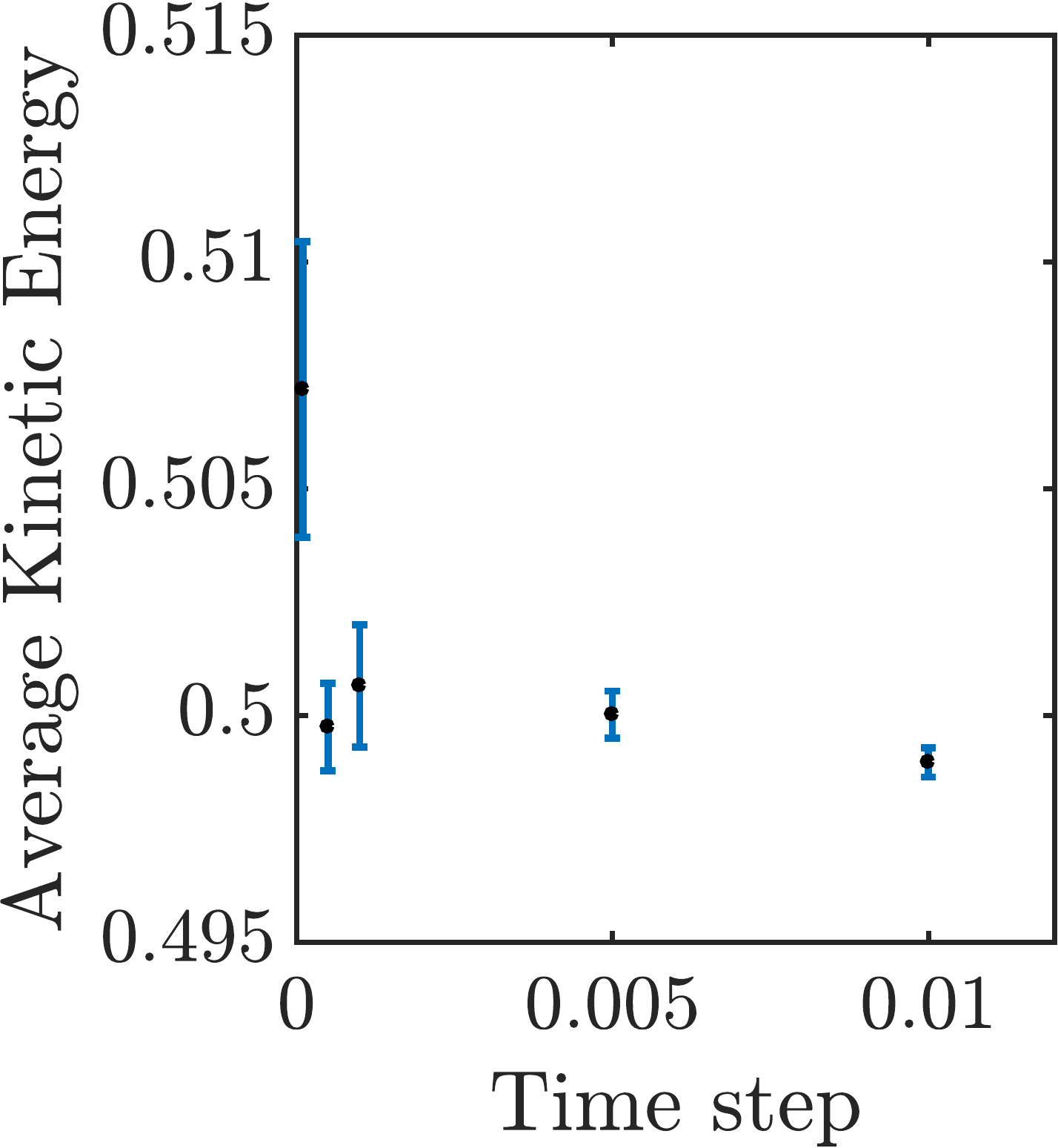}
\end{subfigure}
\subcaption{$\gamma=0.1$}
\end{minipage}
\begin{minipage}{.5\linewidth}
\begin{subfigure}{.5\textwidth}
\includegraphics[width=\linewidth]{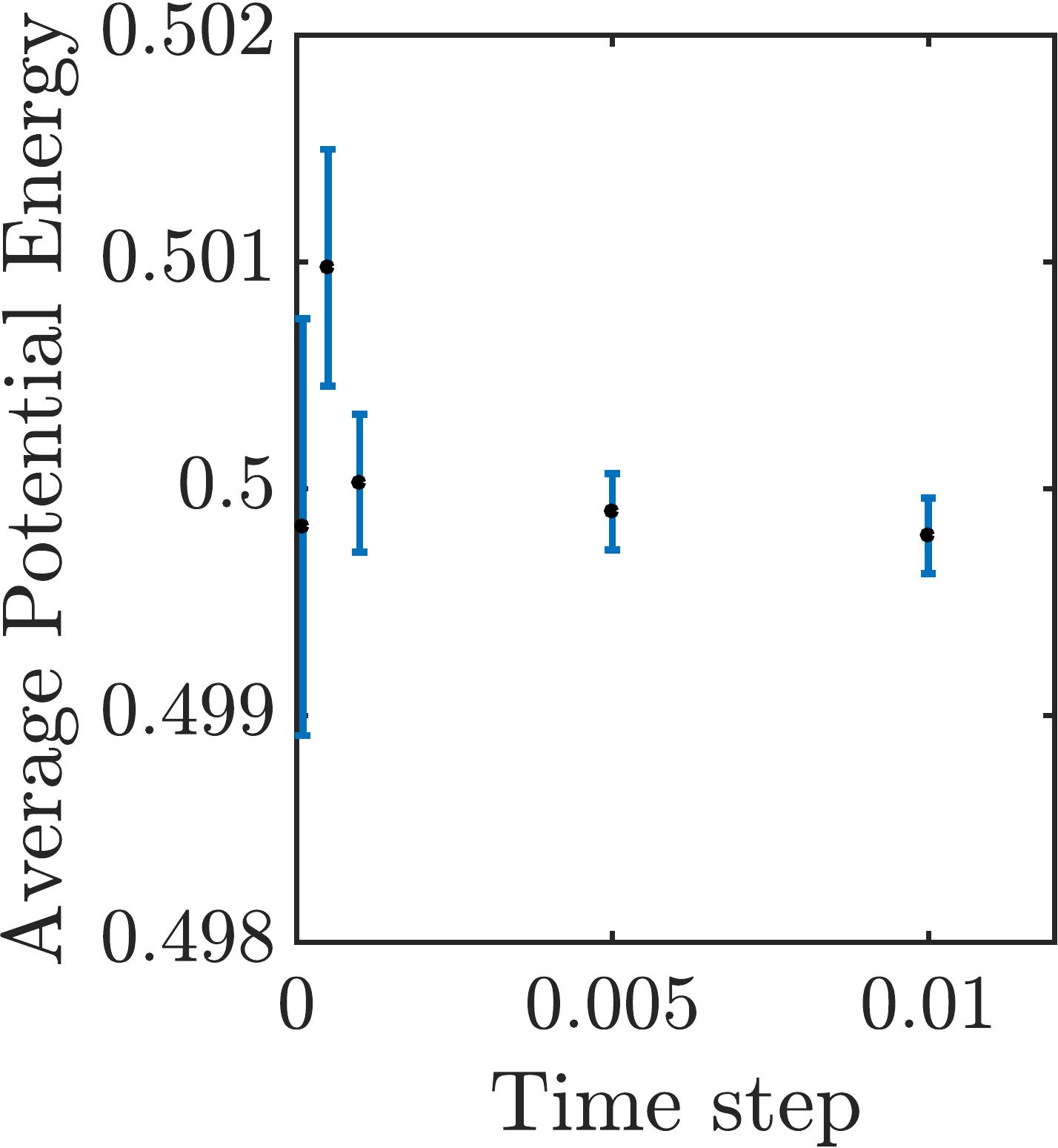}
\end{subfigure}%
\begin{subfigure}{.5\textwidth}
\includegraphics[width=\linewidth]{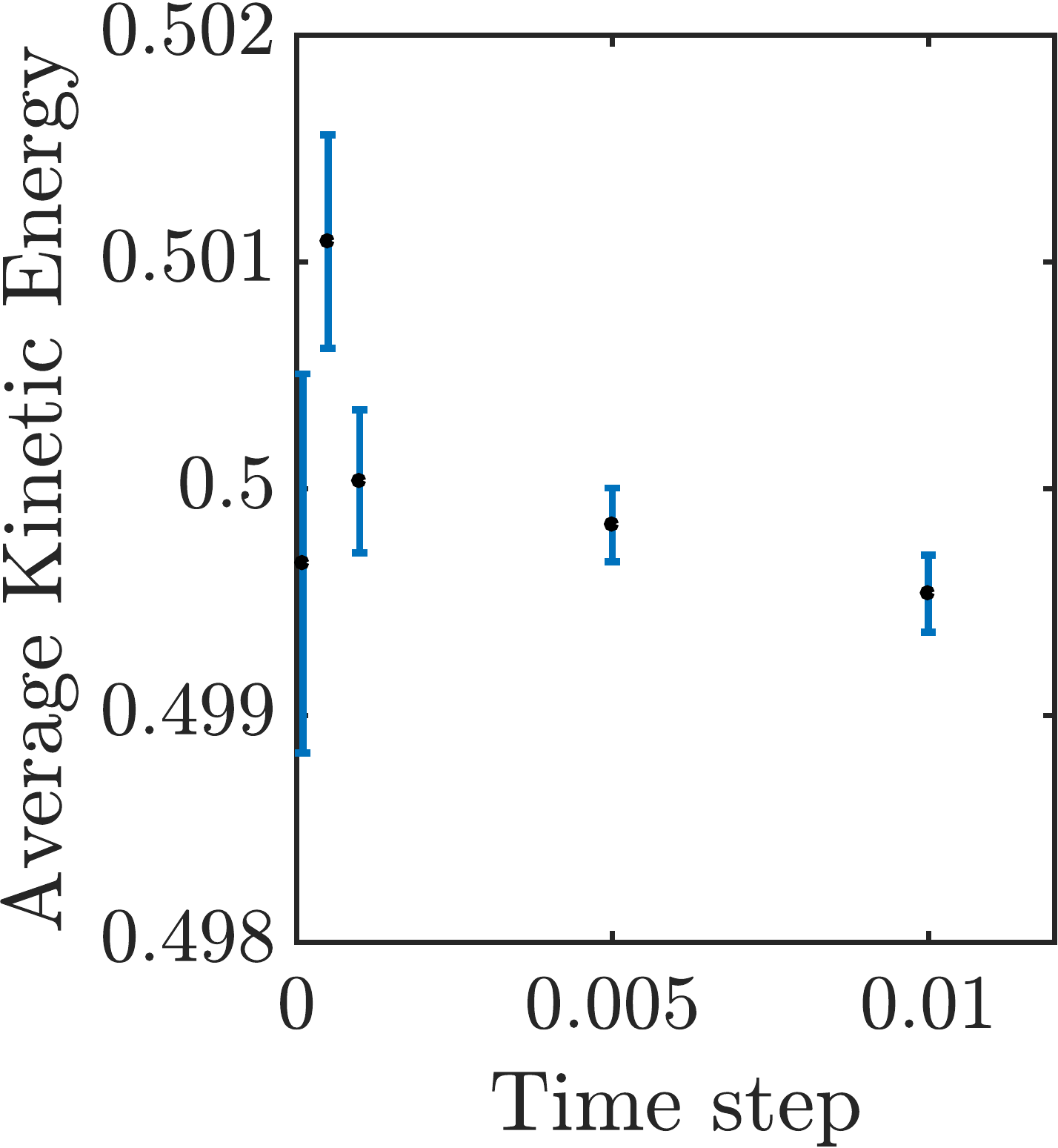}
\end{subfigure}
\subcaption{$\gamma=1$}
\end{minipage}

\begin{minipage}{.5\linewidth}
\begin{subfigure}{.5\textwidth}
\includegraphics[width=\linewidth]{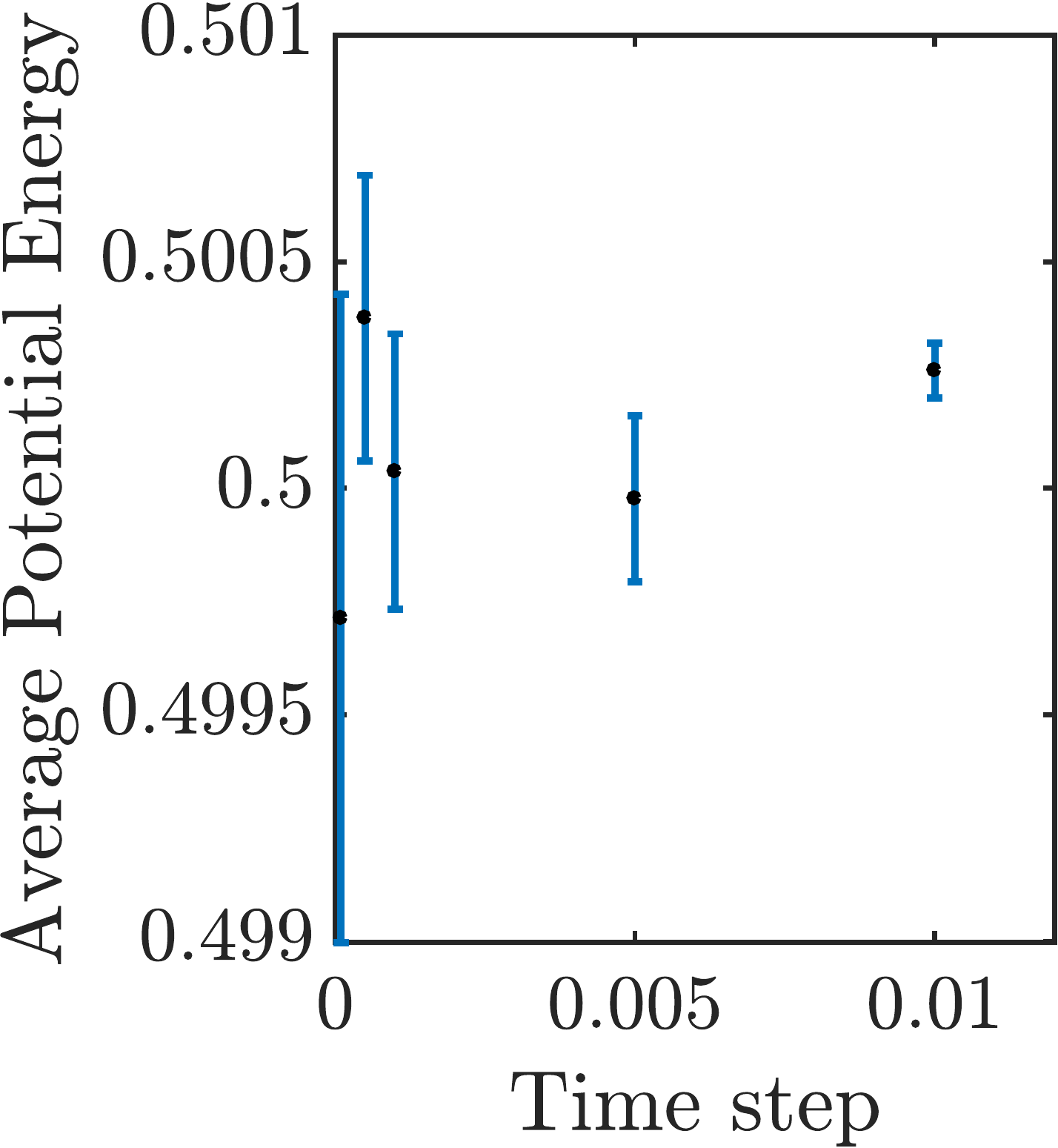}
\end{subfigure}%
\begin{subfigure}{.5\textwidth}
\includegraphics[width=\linewidth]{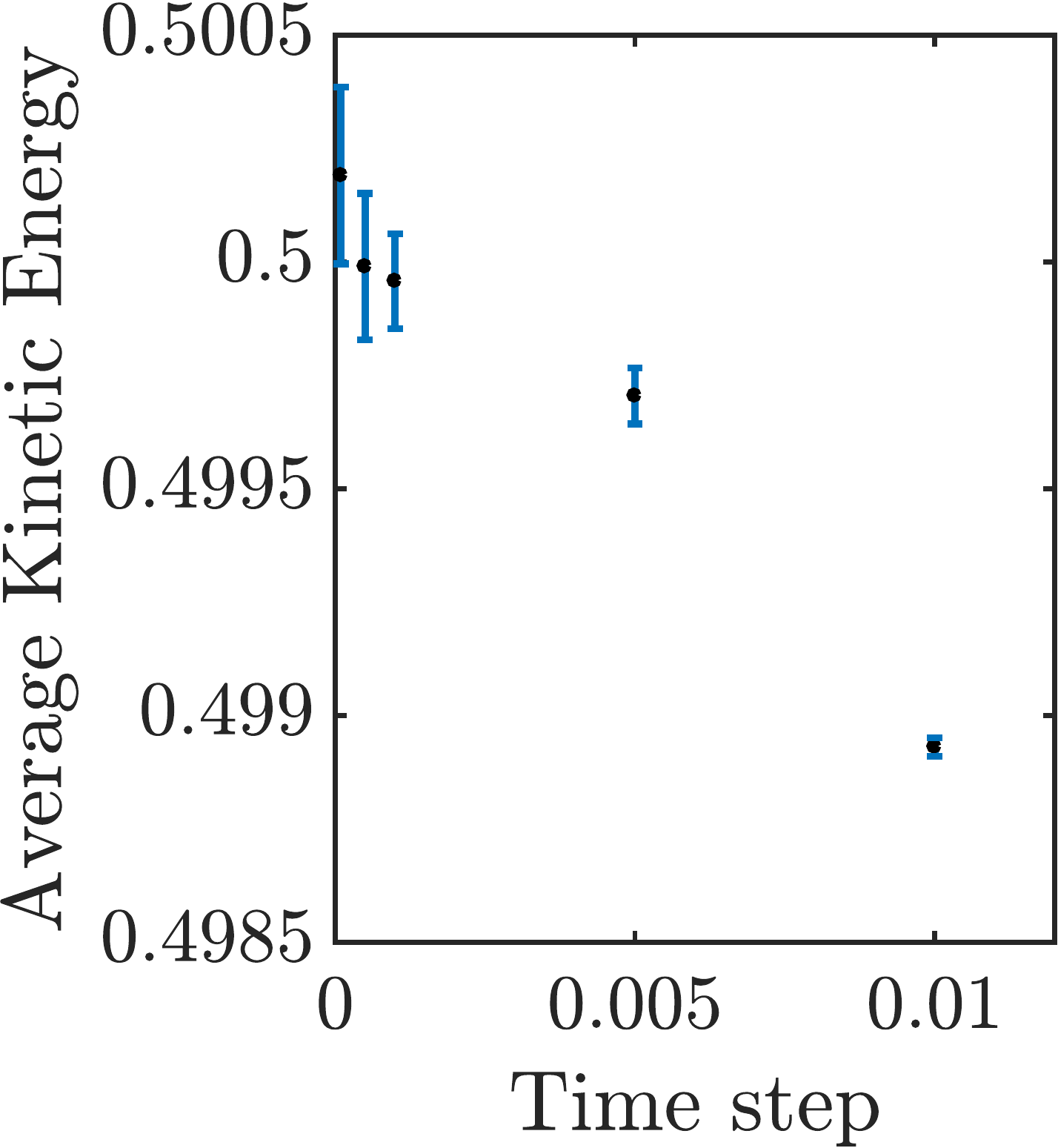}
\end{subfigure}
\subcaption{$\gamma=10$}
\end{minipage}
\begin{minipage}{.5\linewidth}
\begin{subfigure}{.5\textwidth}
\includegraphics[width=\linewidth]{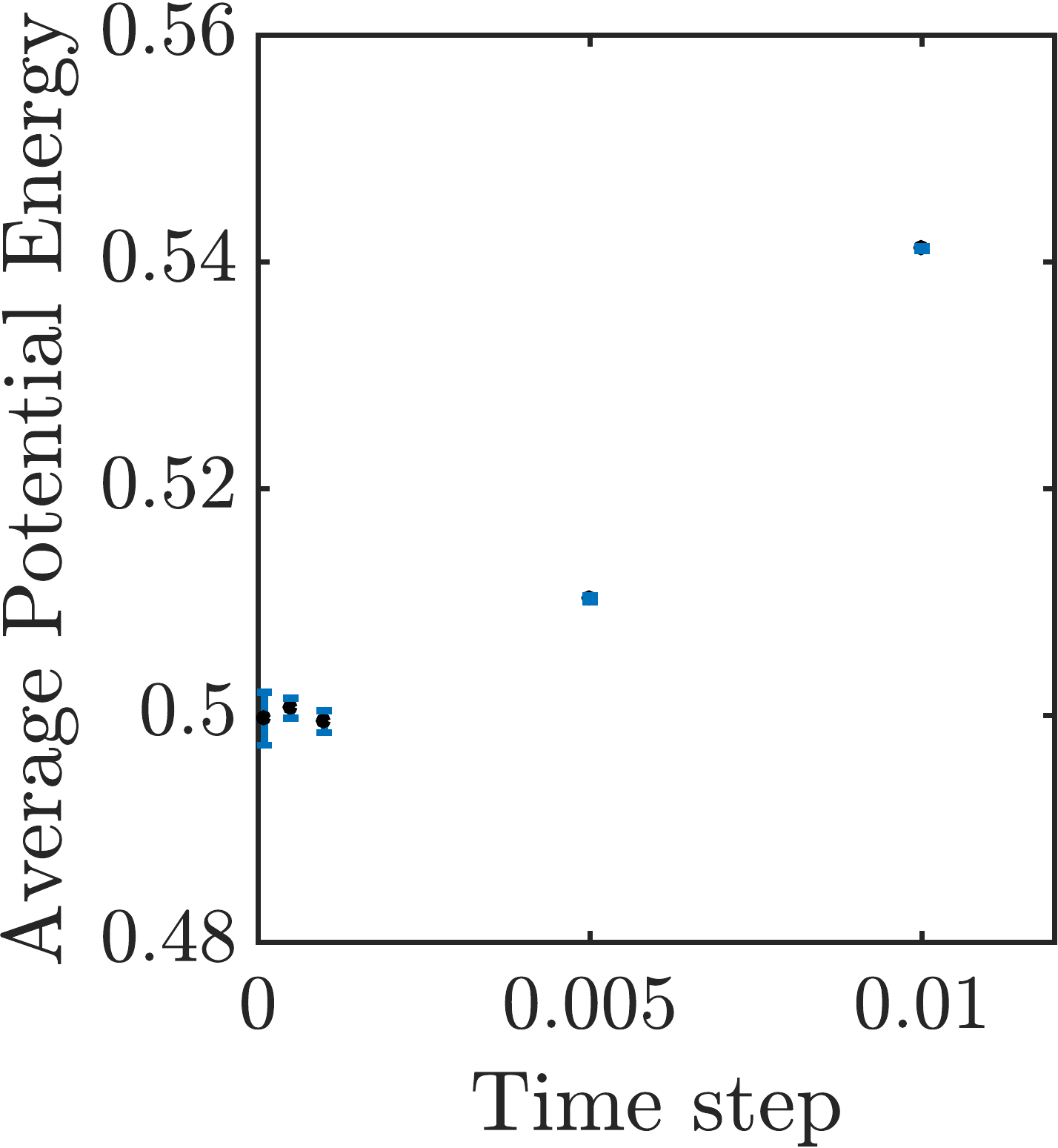}
\end{subfigure}%
\begin{subfigure}{.5\textwidth}
\includegraphics[width=\linewidth]{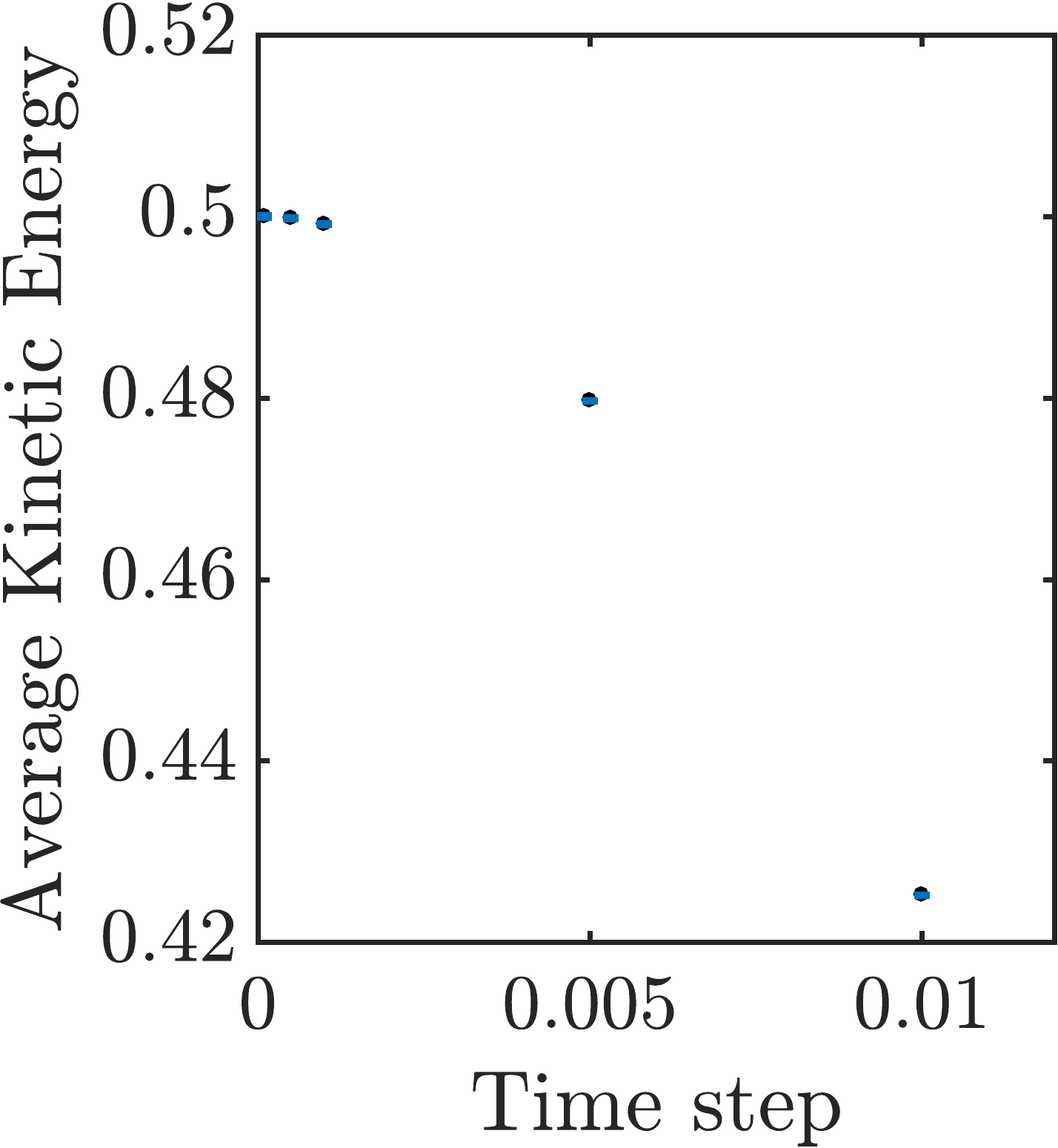}
\end{subfigure}
\subcaption{$\gamma=100$}
\end{minipage}
\caption{A time step of $0.001$ is good enough to ensure that the average potential and kinetic energies approaches $\frac{k_BT}{2}=0.5$ for a wide range of friction values.}
\label{fig:potential_kinetic_timestep}
\end{figure}

However, it is not sufficient to just compare the average kinetic and potential energies to the equipartition limit. We need to ensure that the observed kinetics are robust to our choice of time step. In particular, we need to test that a time step of 0.001 is sufficient for the highest values of our control parameter $F$, which presents the most severe challenge to integrating our Langevin equation (due to the behaviour near $x=0$). Figure~\ref{fig:largest_F} confirms that a time step of 0.001 is appropriate for $F=100$ and the full range of $\gamma$ tested. Each value in the figure is an average over $1000$ initial conditions.

\begin{figure}[h!]
\begin{center}
\[
\includegraphics[scale=0.21]{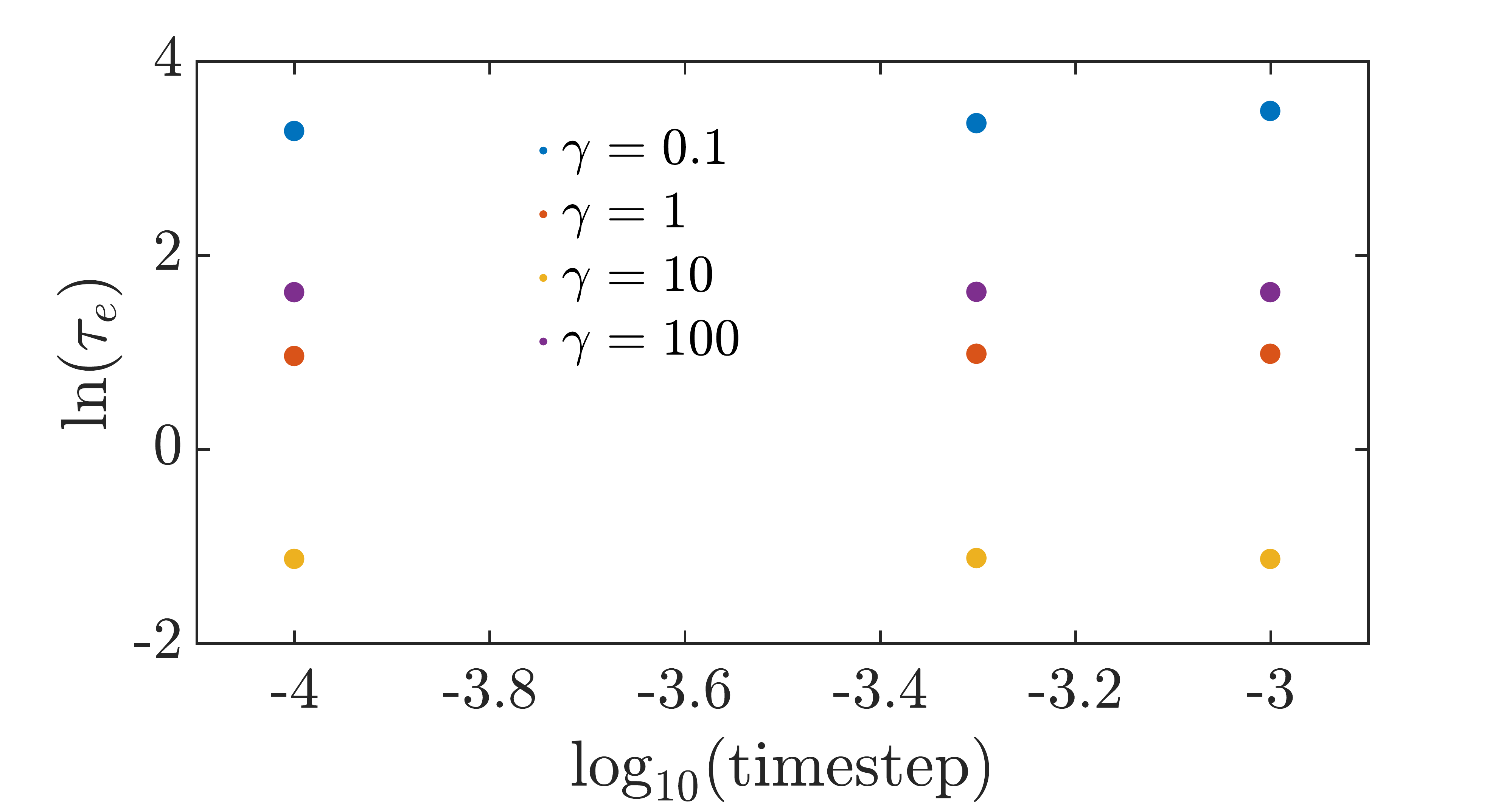}
\]
\end{center}
\caption{A time step of $0.001$ gives reasonable values of erasing time for the largest value of the control parameter $\left(F=100\right)$ that we use in the simulations.}
\label{fig:largest_F}
\end{figure}

\subsection{Erasure region}\label{erasure_region}

In this section, we demonstrate that our results are not limited to our specific definition of the erasure region by considering two alternative criteria for erasing and confirming that our earlier conclusions are supported. 

\subsubsection{Accuracy of erasure: Convergence of probability distribution}\label{subsec:accuracy}

Recall that the erasing time is a sum of transport time $(\tau_t)$ and the mixing time $(\tau_m)$. Since the transport time is independent of the metric used to measure the mixing time, we analyse a proxy for the mixing process in isolation. Specifically, we consider the relaxation of a particle in a harmonic well, initially prepared in an arbitrary non-equilibrium distribution. As an alternative definition of mixing, we consider $\tau_m^\epsilon$ as the first time when the probability distribution of the particle comes within a certain distance (in the appropriate norm, and relative to its initial distribution) of the Gibbs distribution corresponding to the well. More specifically,

\begin{eqnarray*}
\tau^\epsilon_m = \inf_{t\geq 0}\{||\text{law}((x(t),p(t)) -\pi_0(x,p)||_{L^2(\pi_0(x,p))} \leq \epsilon||\text{law}((x(0),p(0)) - \pi_0(x,p)||_{L^2(\pi_0(x,p))}\}
\end{eqnarray*} 
where $(x(t),p(t))$ is the solution to Equation~\ref{eq:control} given appropriate initial conditions. Here, $\pi_0(x,p)$ is the stationary distribution of the harmonic well. As is usual, the weighted norm $\textit{L}^2(\pi_0(x,p)) := \{f|\int_{-\infty}^\infty \int_{-\infty}^\infty|f|^2\pi_0(x,p)dxdp < \infty\}$. Informally, $\tau^\epsilon_e$ is the time required for the distribution to be a factor $\epsilon \ll 1$ ``closer" to the equilibrium distribution than in the initial condition. \\

We define $\eta = 1-\epsilon$ as the \textbf{accuracy} of erasure. Lesser the $\epsilon$, the closer the distribution of the particle is to the Gibbs distribution of the harmonic well and hence more accurate the erasure.\\

Consider the modified Langevin equation

\begin{eqnarray}\label{eq:langevin}
\begin{aligned}
m\,dx = {}& p\,dt \\
\,dp = {}& -\gamma p\,dt - \partial_xN_{A,B}\left(x\right)\,dt + \sqrt{2m\gamma k_B T}\,dW
\end{aligned}
\end{eqnarray}

Here $N_{A,B}(x)=\frac{1}{2}m\omega_0^2\left(x-B\right)^2$ where $\omega_0=\sqrt\frac{8A}{mB^2}$ is the harmonic potential that approximates well ``0''. Equation\ref{eq:langevin} has the generator~\cite[pp.~182]{pavliotis2014stochastic} given by 
\begin{align}\label{eq:generator}
\mathcal{L} = \frac{p}{m}\partial_x - \left(\partial_xN_{A,B}(x)\right)\partial_p + \gamma\left(-p\partial_p + k_BT\partial^2_p\right)
\end{align}

It is common knowledge that the following equation is true~\cite{mattingly2002geometric}. 

\begin{eqnarray*}
||\text{law}((x(t),p(t)) -\pi_0(x,p)||_{L^2(\pi_0(x,p))} \leq e^{-\lambda t}||\text{law}((x(0),p(0)) - \pi_0(x,p)||_{L^2(\pi_0(x,p))}
\end{eqnarray*} 

where $\lambda$ is the first non-zero eigenvalue of the generator $\mathcal{L}$ given by Equation~\ref{eq:generator}. Setting $e^{-\lambda t}=\epsilon$, we get useful \textbf{upper bounds} on the mixing time. In particular, we get

\begin{eqnarray}
\tau^\epsilon_m\leq\frac{1}{\lambda}\ln\frac{1}{\epsilon}
\end{eqnarray}

For the sake of rough scaling, we will use $\tau^\epsilon_m\approx\frac{1}{\lambda}\ln\frac{1}{\epsilon}
$ as an approximate estimate of the mixing time. It is important to note that the generator $\mathcal{L}$ is not self-adjoint and may possess imaginary eigenvalues. The rate of convergence in such cases will be determined by the real part of the eigenvalue. In fact using~\cite[pp.~200]{pavliotis2014stochastic}, the first non-zero eigenvalue of the generator is 

\begin{eqnarray*}
\lambda = \frac{\gamma}{2} - \frac{1}{2}\sqrt{\gamma^2 - 4\omega_0^2}
\end{eqnarray*}

In the underdamped limit when $\gamma \ll 2\omega_0$, we have $Re(\lambda)=\frac{\gamma}{2}$. Therefore $\tau^\epsilon_m\approx\frac{2}{\gamma}\log(\frac{1}{\epsilon})$ in the low friction regime. When friction is very high i.e. $\gamma \gg \omega_0$, we have $\lambda \approx\frac{\omega_0^2}{\gamma}$. As a result we get $\tau^\epsilon_m\approx\frac{\gamma}{\omega_0^2}\log\frac{1}{\epsilon}$.
Thus our proxy for the mixing process produces  $\tau^\epsilon_m\propto\frac{1}{\gamma}$ in the low friction regime and $\tau^\epsilon_m\propto\gamma$ in the high friction regime,  consonant with the scaling and non-monotonicity observed using the erasure region criterion. As a consequence, using a convergence criterion for erasure would not change the physics of the problem, merely perturbing the erasing time-scale quantitatively.

\subsubsection{$4k_BT$ criterion for erasure region}\label{subsec:4kT_criterion}

Within the framework of the original ``erasure region" criterion discussed in the main text, we now consider the robustness of results to changing the numerical value of the criterion. Specifically, we here define the erasure region as all phase space points with total energy atleast $4k_BT$ below the barrier height. More formally, 

\begin{eqnarray*}
\tau^{4 k_BT}_e = \E\left[\inf\{t\geq 0 \mid x(t)<0 \text{ and } H(x(t),p(t))\leq A - 4k_BT\}\right]
\end{eqnarray*}
where $\left(x\left(t\right),p\left(t\right)\right)$ is the solution to Equation~\ref{eq:control} with the initial condition $\left(x\left(0\right),p\left(0\right)\right)\sim_{\text{law}}\pi_1\left(x,p\right)$. We now show that we get the same non-monotonicity and scaling of erasing time as a function of friction-coefficient that we got using the $3k_BT$ criterion. In particular, the erasing time scales as $\frac{1}{\gamma}$ in the low friction regime and scales as $\gamma$ at high friction. Figure~\ref{fig:supplement_erasing} illustrates this fact. Fits are performed using analytical expressions equivalent to those discussed in the main text, but adjusted for the new numerical value of the boundary of the erasure region.

\begin{enumerate}
\item Low friction regime:
\begin{eqnarray}
\tau^{4k_BT}_e\approx \sqrt{\frac{2mB}{F}} + \frac{1}{\gamma}\ln{\frac{A + F\cdot B}{A-4k_BT}}.
\end{eqnarray}
\item High friction regime:
\begin{eqnarray}
\tau^{4 k_BT}_e \approx \frac{mB\gamma}{F} + \frac{2mB^2\gamma}{5A}
\end{eqnarray}
\end{enumerate}

\begin{figure}[ht!]
\begin{minipage}{0.5\textwidth}
  \centering
  \includegraphics[scale=0.21]{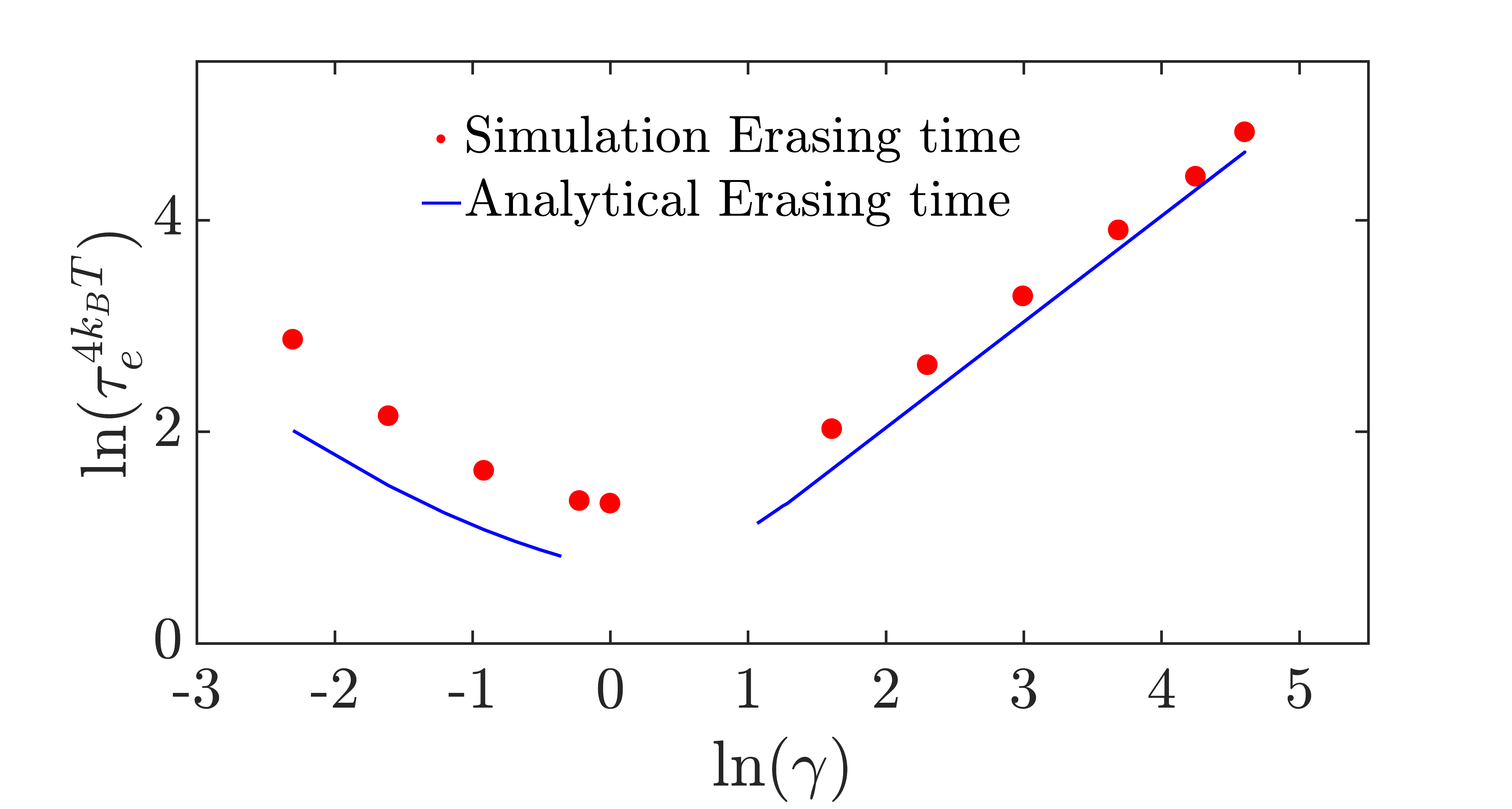}
  \subcaption{$A=10,F=1$}
  \label{fig:erasing_low_F}
\end{minipage}
\begin{minipage}{0.5\textwidth}
  \centering
  \includegraphics[scale=0.21]{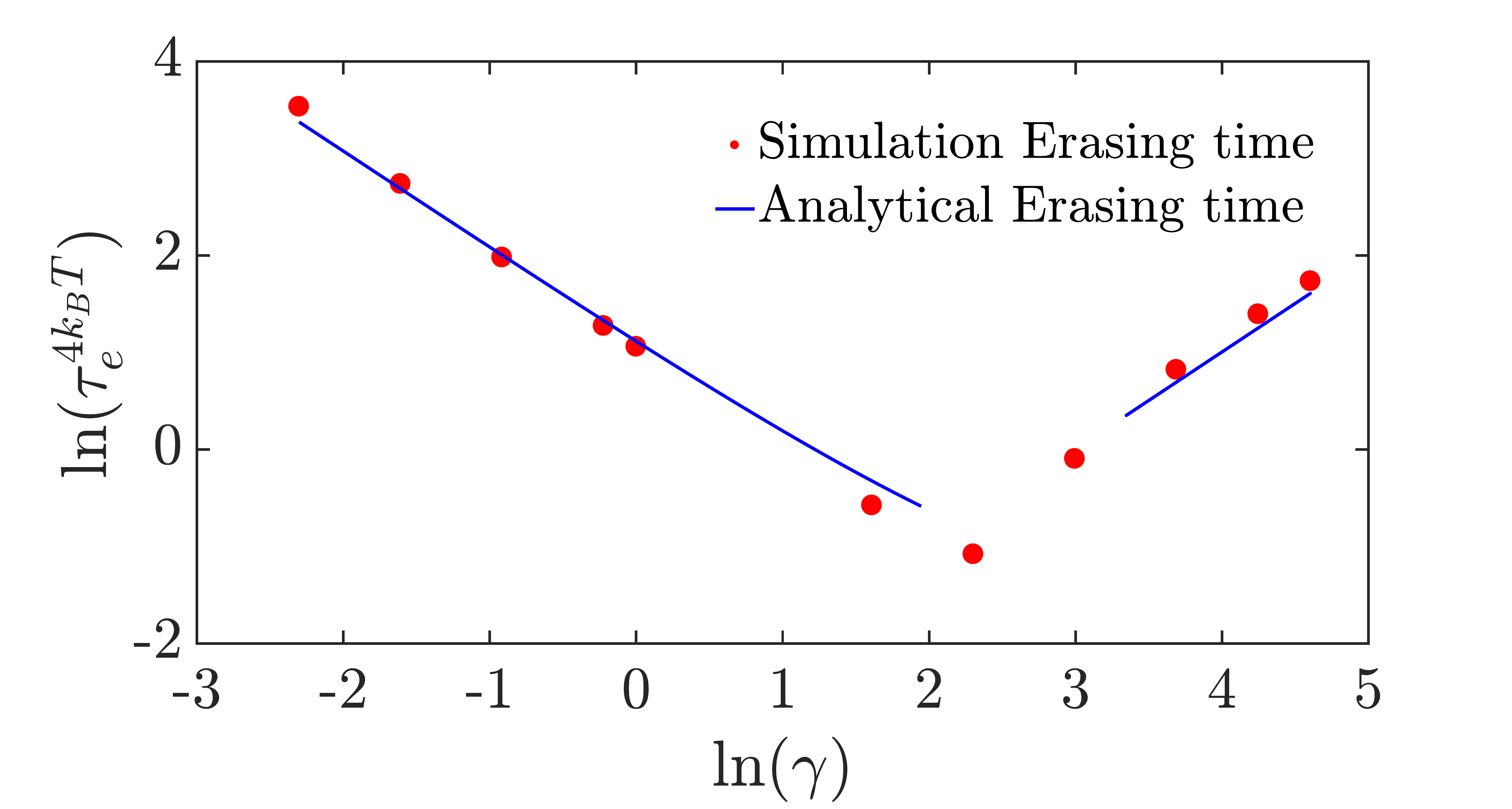}
  \subcaption{$A=10,F=100$}
  \label{fig:erasing_high_F}
\end{minipage}
\caption{Evidence from simulation that the use of $4k_BT$ to define the erasure reason does not change the fundamental physics of the problem.}
\label{fig:supplement_erasing} 
\end{figure}

\subsection{Work calculation}

\subsubsection{Full calculation for work done at t=0}\label{appendix-full work}

We first provide a more detailed justification of the approximation $W=(A+F\cdot B-k_BT/2)/2$ for the work done when the control potential is switched on, then compare it to the exact result. From Equation~\ref{eq:sekimoto_work}, and letting $I = \{x\geq 0\mid A - U_{A,B}\left(x\right) + F\cdot x\geq 0\}$, we get

\begin{eqnarray}\label{eq:true formula}
\langle W\rangle &=&  \Int_{0}^{\tau}\Int_{I}V_F\left(x,t\right)p\left(x,t\right)\delta\left(t\right)dxdt \nonumber \\ 
                 &=& \Int_{I}\left(A - U_{A,B}\left(x\right) + F\cdot x\right)p\left(x,0\right)dx  
\end{eqnarray}

Since $p\left(x,0\right)\propto \e^{-\frac{U_{A,B}\left(x\right)}{k_BT}}$, we can rewrite the expression of work as

\begin{eqnarray}\label{eq:work_interval_expression}
\langle W \rangle &=& \frac{\Int_{I}\left(A - U_{A,B}\left(x\right) + F\cdot x\right)\e^{-\frac{U_{A,B}\left(x\right)}{k_BT}}dx}{\Int_{-\infty}^{\infty}\e^{-\frac{U_{A,B}\left(x\right)}{k_BT}}dx}
\end{eqnarray}

Since $I\subseteq\left[0,\infty\right)$ and $\left(A - U_{A,B}\left(x\right) + F\cdot x\right)\e^{-\frac{U_{A,B}\left(x\right)}{k_BT}}$ is negligible as $x \rightarrow \infty$, replacing the upper limit of integration by $\infty$ is reasonable. Hence the integral becomes
\begin{eqnarray}
\langle W \rangle \approx \frac{\Int_{0}^{\infty}(A - U_{A,B}\left(x\right) + F\cdot x)\e^{-\frac{U_{A,B}\left(x\right)}{k_BT}}dx}{\Int_{-\infty}^{\infty}\e^{-\frac{U_{A,B}\left(x\right)}{k_BT}}dx}
\end{eqnarray}

When $A >> k_BT$, we can use Bessel's functions to approximate this integral giving
\begin{eqnarray}\label{eq:work_expression}
\langle W \rangle \approx \frac{\left(A + F\cdot B -\frac{k_BT}{2}\right)}{2}
\end{eqnarray}
justifying the crude approximation in the main text. The accuracy of this expression compared to Eq.~\ref{eq:work_interval_expression} is illustrated in Figure~\ref{fig:work_accuracy}. 

\begin{figure}[h!]
\begin{center}
\[\includegraphics[scale=0.22]{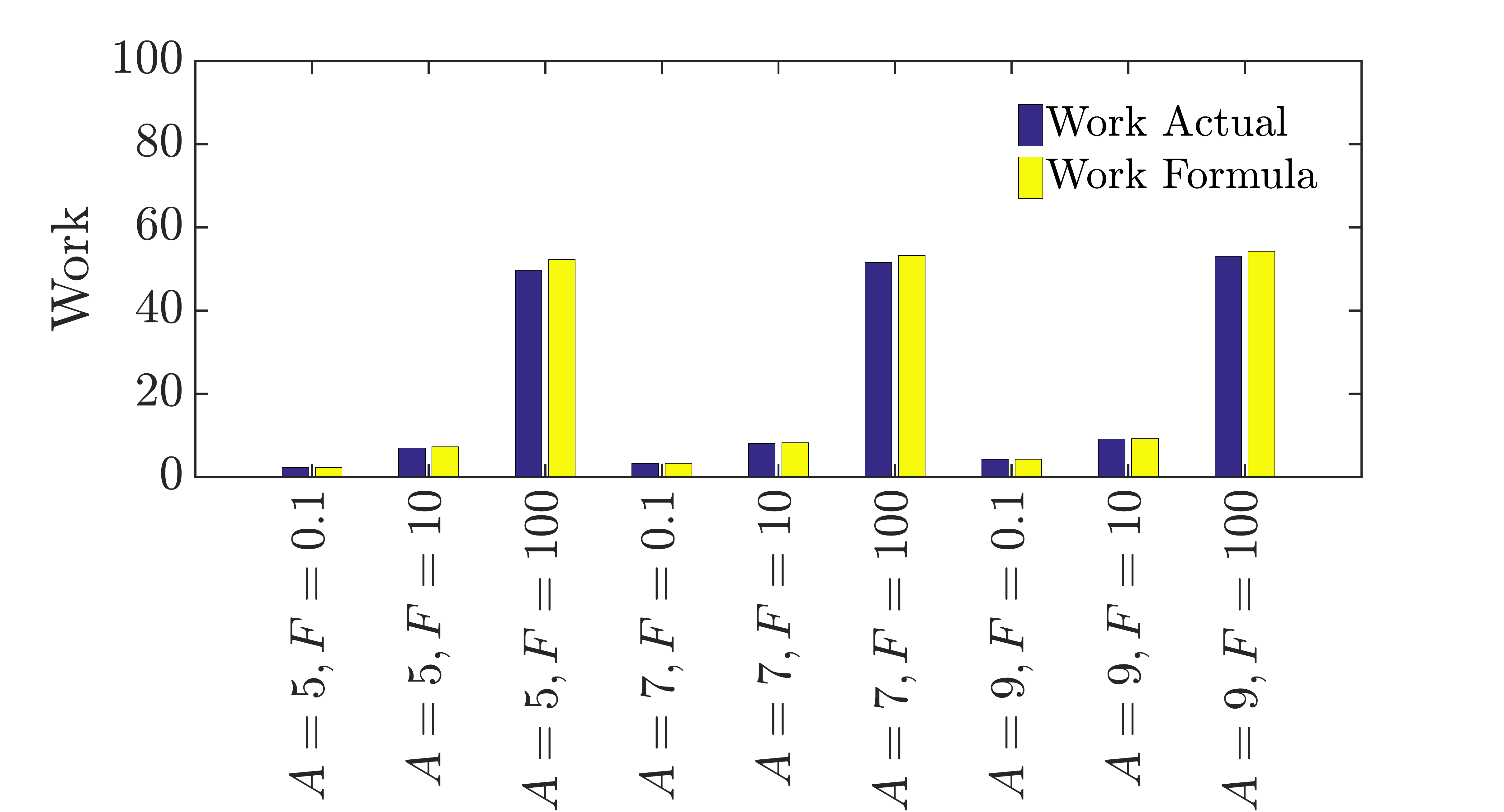}\]
\end{center}
\caption{Work when the control is switched on against various values of $A$ and $F$, comparing the full expression Eq.~\ref{eq:work_interval_expression} and the approximate result Eq.~\ref{eq:work_expression}.}
\label{fig:work_accuracy}
\end{figure}

\subsection{The potential for energy recovery is negligible}\label{subsec:no_recovery_energy}

Here, we argue that energy recoverable at the end of the protocol is very small, and hence may be neglected. We assume that the control is switched off after a time $\tau$ sufficiently large compared to $\tau_e$ so that the proportion of particles remaining on the right hand side of the well is determined by the Boltzmann factor. The work that we could then in principle recover is given by the following expression:

\begin{eqnarray}\label{eq:recovered_work}
\langle W_{rec} \rangle &\approx \frac{\Int_{I}\left(A - U_{A,B}\left(x\right) + F\cdot x\right)\e^{\frac{-(A + F\cdot x)}{k_BT}}dx}{\Int_{-\infty}^{0}\e^{\frac{-U_{A,B}\left(x\right)}{k_BT}}dx  + \Int_{0}^{\infty}\e^{\frac{-(A + F\cdot x)}{k_BT}}dx}
\end{eqnarray}

Recall that $I = \{x \geq 0 \mid A - U_{A,B}\left(x\right) + F\cdot x \geq 0\}$. This implies that $I=[0,x^*]$, where $A - U_{A,B}\left(x^*\right) + F\cdot x^*=0$. Using Equations~\ref{eq:work_expression} and ~\ref{eq:recovered_work}, we will calculate the fraction of recovered work i.e., $W^{f}_{rec} = \frac{\langle W_{rec} \rangle}{\langle W \rangle}$. Figure~\ref{fig:fraction_recovered_work} precisely calculates this quantity. As is evident from the figure, the fraction is almost negligible and reaches its maximum value at low $A$ and $F$.

\begin{figure}[h!]
\begin{center}
\[
\includegraphics[scale=0.6]{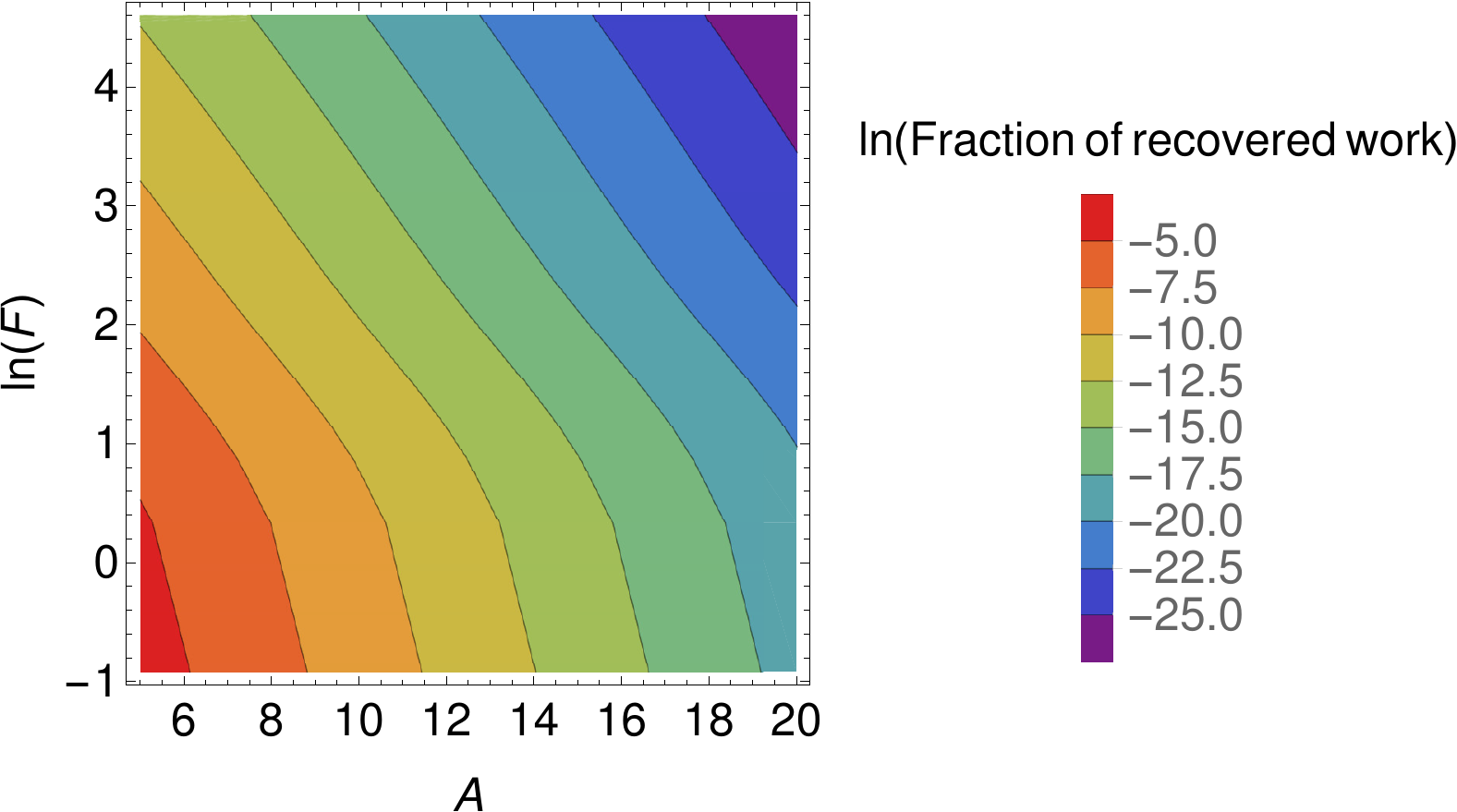}
\]
\end{center}
\caption{Negligible energy can be recovered for our family of controls.}\label{fig:fraction_recovered_work}
\end{figure}

\subsection{Regression and Cross Validation}\label{subsec:regression}

We use cubic regression to interpolate between simulation data for both the reliability and erasing time-scales. Let $F'=\log(F)$ and $\gamma'=\log(\gamma)$. Then we use the following polynomials to fit the time-scales.

\begin{enumerate}

\item \textbf{Erasing Polynomial}:
\begin{eqnarray}\label{erasing_polynomial}
\begin{split}
\log(\tau_e) &= b_1 + b_2A^3 + b_3F'^3 + b_4\gamma'^3 + b_5A^2F' + b_6A'F'^2 + b_7F'^2\gamma' + b_8F'\gamma'^2 \\
 & + b_9A'^2\gamma' + b_{10}A\gamma'^2 + b_{11}A^2 + b_{12}F'^2 + b_{13}\gamma'^2 + b_{14}AF' + b_{15}F'\gamma'\\
 &+ b_{16}A\gamma' + b_{17}A + b_{18}F' + b_{19}\gamma'
\end{split}
\end{eqnarray}

\item \textbf{Reliability Polynomial}:
\begin{eqnarray}
\log(\tau_r) = c_1 + c_2A^3 + c_3\gamma'^3 + c_4A^2\gamma' + c_5A\gamma'^2 + c_6A^2 + c_7\gamma'^2 + c_8A\gamma' + c_9A + c_{10}\gamma
\end{eqnarray}
, where the coefficients $b_1,b_2,\cdots b_{19}$ and $c_1,c_2,\cdots c_{10}$ are to be determined by regression.

\end{enumerate}

\begin{figure}[h!]
\begin{minipage}{0.45\textwidth}
  \includegraphics[scale=0.16]{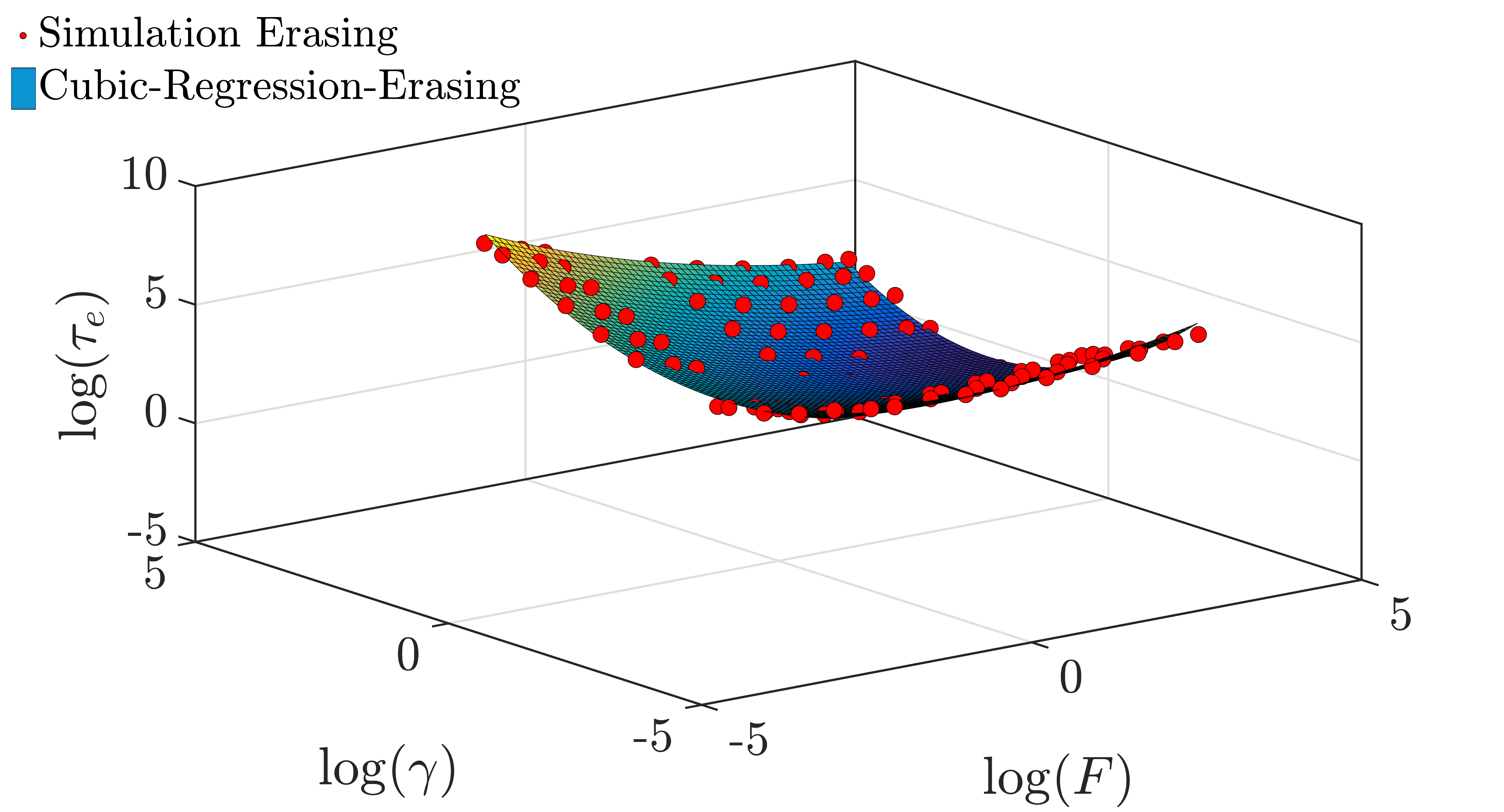}
  \subcaption{Erasing time from simulation and cubic-regression for $A=6$}
\end{minipage}
\begin{minipage}{0.45\textwidth}
   \includegraphics[scale=0.16]{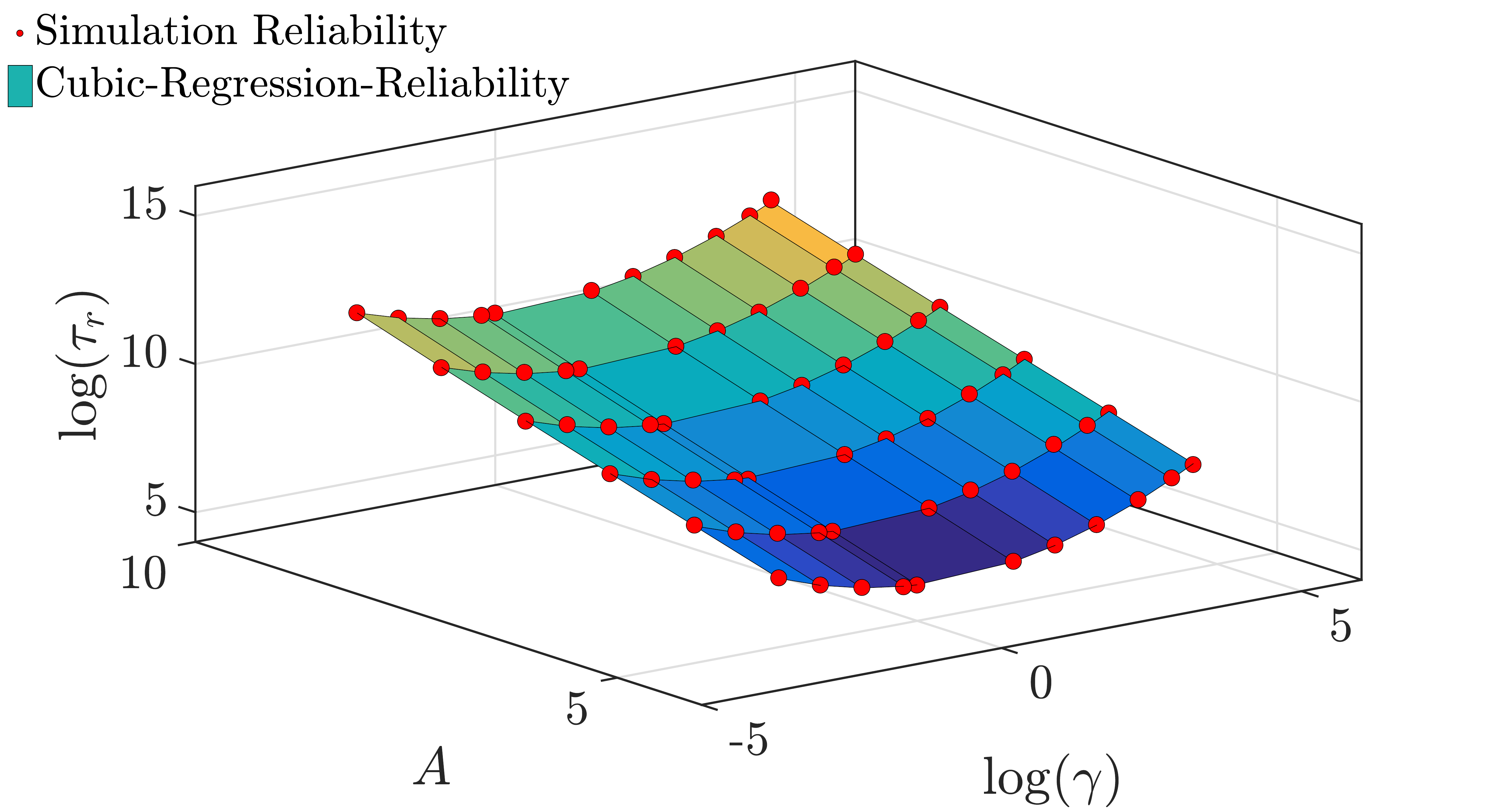}
  \subcaption{Reliability time from simulation and cubic-regression}
\end{minipage}
\caption{}\label{fig:regression_fits}
\end{figure}

Figure~\ref{fig:regression_fits} gives a visual illustration of the fact that cubic fits offer a good approximation to the simulation results for both the erasing and reliability time-scales. In what follows, we present a more detailed and formal justification using cross-validation.

\begin{figure}[h!]
\begin{minipage}{0.5\textwidth}
  \includegraphics[scale=0.17]{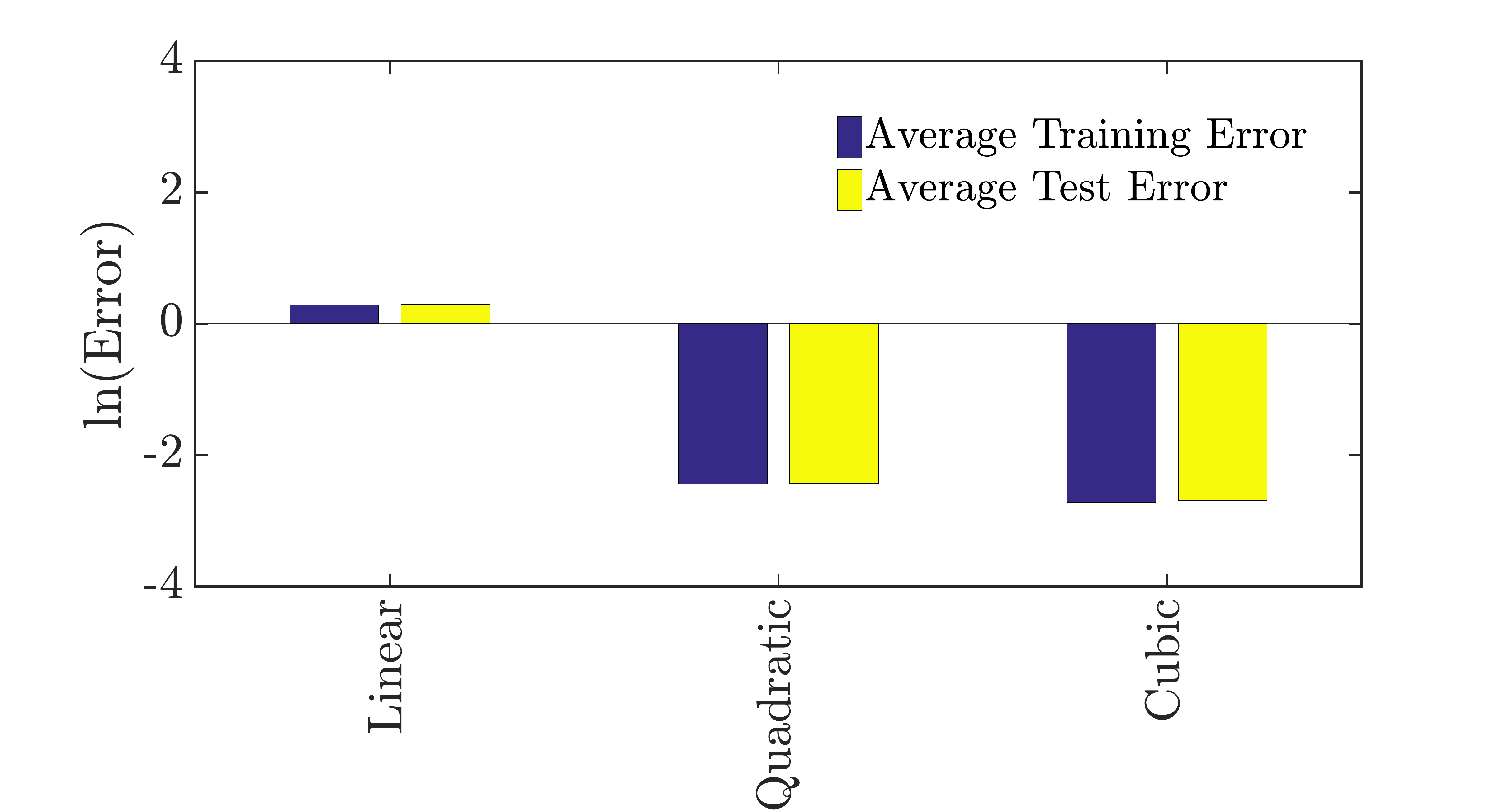}
  \subcaption{Errors of regression fits for erasing time}
\end{minipage}
\begin{minipage}{0.5\textwidth}
   \includegraphics[scale=0.17]{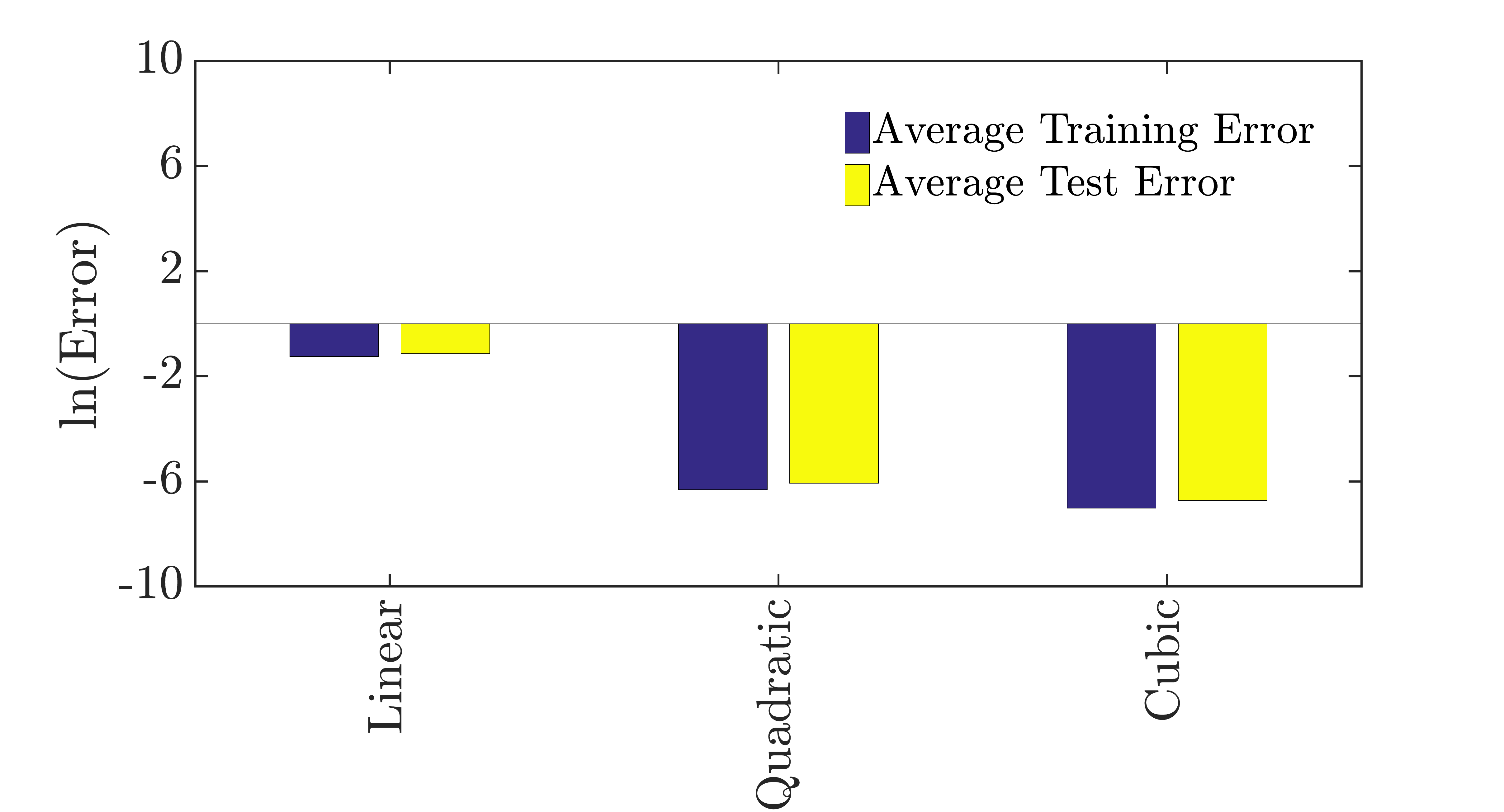}
  \subcaption{Errors of regression fits for reliability time}
\end{minipage}
\caption{}\label{fig:cross_validation_errors}
\end{figure}

We perform ``Leave-one-out'' cross validation to justify the use of cubic regression. Figure~\ref{fig:cross_validation_errors} reports the mean square training and testing cross-validation errors corresponding to linear, quadratic and cubic fits. A lower value of the testing error indicates a good fit. Cubic regression has the lowest value of testing errors amongst the fits considered for both the reliability and erasing time-scales. Figure~\ref{fig:cross_validation_errors} confirms that the training and testing errors corresponding to cubic-regression for both the time-scales are roughly comparable(with the training error being slightly lower than the testing error). As a result, we can safely assume that the cubic polynomial does not over-fit the data and use it for modelling both the time-scales.

\subsection{Calculation of well parameters}\label{subsec:well_parameters}
Here we calculate the quantities needed to apply Equation~\ref{eq:meshkov_melnikov} to our system, with the potential $U_{A,B}\left(x\right) = A\left(\frac{x^2}{B^2}-1\right)^2$. We have $\partial_x{U_{A,B}\left(x\right)} = \frac{4Ax\left(x^2-B^2\right)}{B^4}$ and $\partial_{xx}U_{A,B}\left(x\right) = \frac{4A\left(3x^2-B^2\right)}{B^4}$. 

\begin{enumerate}
\item \textbf{Angular frequency at barrier height $\left(\omega_b\right)$}: We can approximate the region near the barrier by an inverted harmonic oscillator. By Taylor expanding the potential about the point $x=0$, we get
\begin{eqnarray}
U_{A,B}\left(x\right)\approx U\left(0\right) + \partial_x{U_{A,B}\left(x\right)}\bigg|_{x=0}x + \frac{\partial_{xx}{U_{A,B}\left(x\right)}\bigg|_{x=0}x^2}{2} \approx A - \frac{2Ax^2}{B^2}= A - \frac{m\omega_b^2x^2}{2} 
\end{eqnarray}
Therefore we have $\omega_b = \sqrt{\frac{4A}{mB^2}}$.
\item \textbf{Angular frequency at the bottom of the well $\left(\omega_0\right)$}: We can approximate the region near the bottom of the well by a harmonic oscillator. By Taylor expanding the potential about the point $x=B$, we get
\begin{eqnarray}
\begin{aligned}
U_{A,B}\left(x\right) &\approx U\left(B\right) + \partial_x{U_{A,B}\left(x\right)}\bigg|_{x=B}\left(x-B\right) + \frac{\partial_{xx}{U_{A,B}\left(x\right)}\bigg|_{x=B}\left(x-B\right)^2}{2} \\
&= \frac{4A\left(x-B\right)^2}{B^2}= \frac{m\omega_0^2\left(x-B\right)^2}{2} 
\end{aligned}
\end{eqnarray}
Therefore we have $\omega_0 = \sqrt{\frac{8A}{mB^2}}$.
\item \textbf{Action at barrier height} $I\left(A\right)$: Consider a particle of mass $m$ with a starting velocity $v=0$, moving along a constant energy surface with energy $A$. The particle starts at $x=0$ and moves to $x = \sqrt{2}B$ and returns back to $x=0$. The action for this round trip is given by $I\left(A\right) = \oint pdx = 2\sqrt{2m}\int_{0}^{\sqrt{2}B}\sqrt{A- A\left(\frac{x^2}{B^2}-1\right)^2}dx = \frac{8B\sqrt{mA}}{3}$.
\end{enumerate}

\subsection{Locally trapped bits are uniquely trapped}\label{single_minimum_work}

In this section, we give typical plots for our family of controls that show no evidence of multiple local minima in erasing time within a level set of work. Towards this we let $F'=\log(F)$ and $\gamma'=\log(\gamma)$. Using the same form of regression polynomial as in Equation~\ref{erasing_polynomial}, but at constant work $W$, this translates to

\begin{eqnarray}
\begin{split}
\log(\tau_e) &= b_1 + b_2(W-e^{F'})^3 + b_3F'^3 + b_4\gamma'^3 + b_5(W-e^{F'})^2F' + b_6(W-e^{F'})F'^2 + b_7F'^2\gamma' + b_8F'\gamma'^2 \\
 & + b_9(W-e^{F'})^2\gamma' + b_{10}(W-e^{F'})\gamma'^2 + b_{11}(W-e^{F'})^2 + b_{12}F'^2 + b_{13}\gamma'^2 + b_{14}(W-e^{F'})F'\\
 & + b_{15}F'\gamma'+ b_{16}(W-e^{F'})\gamma' + b_{17}(W-e^{F'}) + b_{18}F' + b_{19}\gamma'
\end{split}
\end{eqnarray}

Note that $\left(\frac{d\tau_e}{d\gamma'}\right)_{W,F'} = \gamma\left(\frac{d\tau_e}{d\gamma}\right)_{W,F}$ and $\left(\frac{d\tau_e}{dF'}\right)_{W,\gamma'}=F\left(\frac{d\tau_e}{dF}\right)_{W,\gamma}$. Therefore solving for $\left(\frac{d\tau_e}{d\gamma}\right)_{W,F}=\left(\frac{d\tau_e}{dF}\right)_{W,\gamma}=0$ is equivalent to solving for $\left(\frac{d\tau_e}{d\gamma'}\right)_{W,F'}=\left(\frac{d\tau_e}{dF'}\right)_{W,\gamma'}=0$. We solve $\left(\frac{d\tau_e}{d\gamma'}\right)_{W,F'}=\left(\frac{d\tau_e}{dF'}\right)_{W,\gamma'}=0$ numerically and plot it in Figure~\ref{fig:local_unique_trapped}. As illustrated by Figure~\ref{fig:local_unique_trapped}, there is exactly one solution to the equations $\left(\frac{d\tau_e}{d\gamma}\right)_{W,F}=\left(\frac{d\tau_e}{dF}\right)_{W,\gamma}=0$ within the broad range of parameters allowed, confirming our assumption that locally trapped bits are uniquely trapped.

\begin{figure}[h!]
\begin{minipage}{0.5\textwidth}
  \includegraphics[scale=0.17]{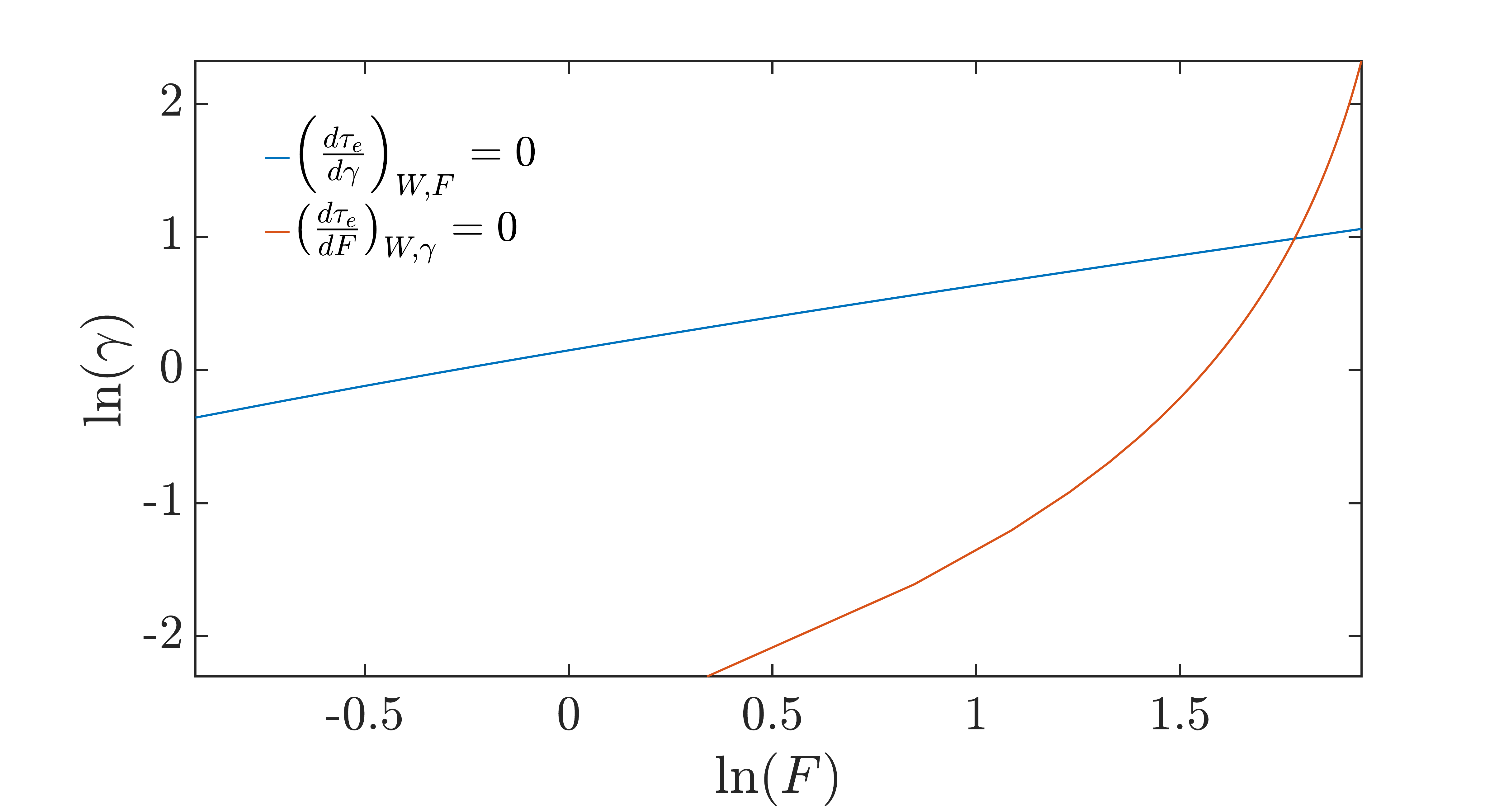}
  \subcaption{W=12}
\end{minipage}
\begin{minipage}{0.5\textwidth}
   \includegraphics[scale=0.17]{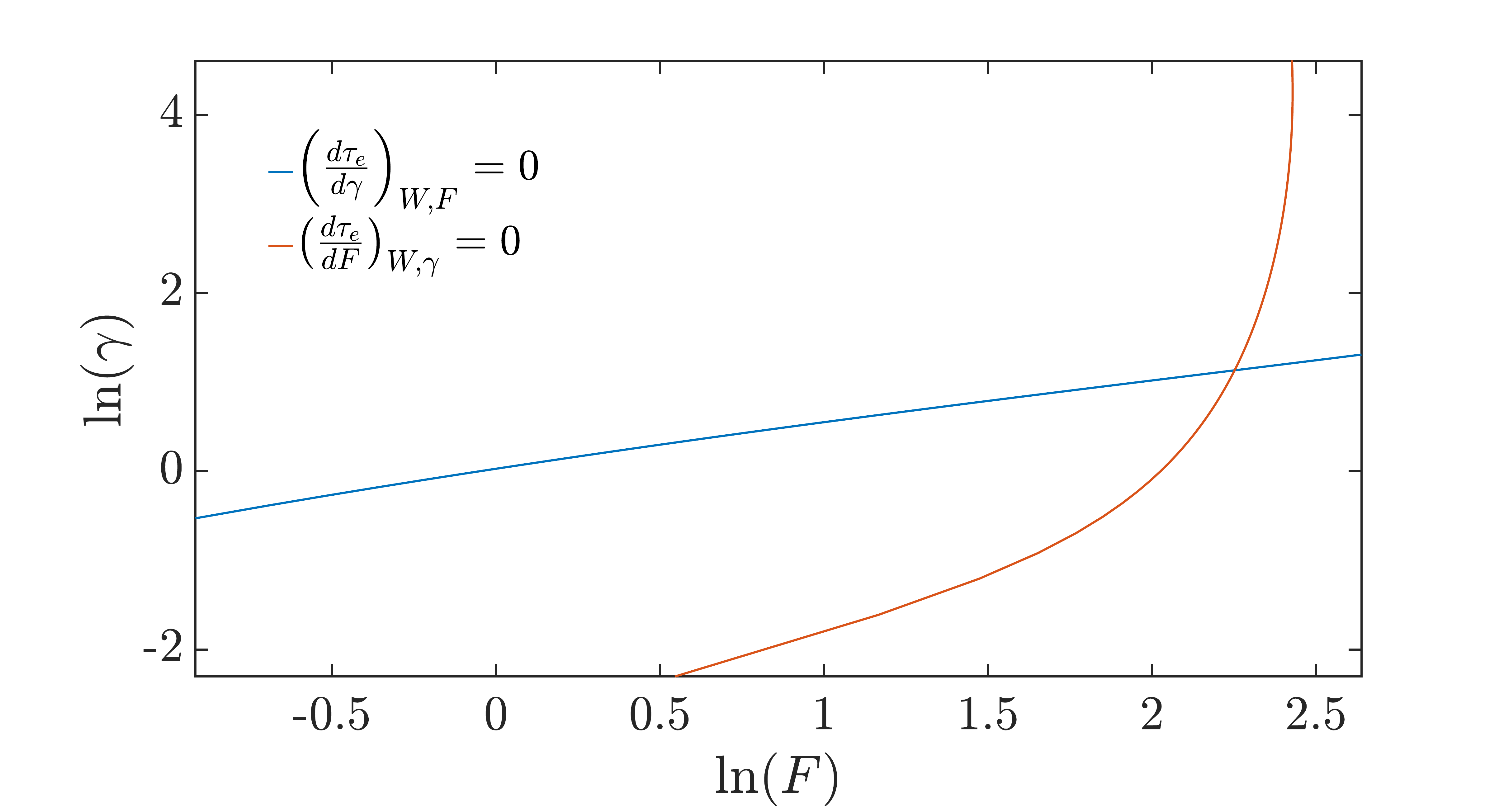}
  \subcaption{W=19}
\end{minipage}
\caption{Evidence that for our family of controls, locally trapped bits are uniquely trapped. The system $\left(\frac{d\tau_e}{d\gamma}\right)_{W,F}=\left(\frac{d\tau_e}{dF}\right)_{W,\gamma}=0$ has exactly one solution within the broad range of parameters considered, which corresponds to a unique local minimum of erasing time in a level set of work. This is illustrated for work $W=12$ and $W=19$. The situation is representative for other values of work.}\label{fig:local_unique_trapped}
\end{figure}

\subsection{Geometry of the Optimal Bit}\label{sec:geometry_optimal_bit}

\begin{enumerate}

\item \textbf{Proof of Claim~\ref{lem:improvee}}
\begin{proof}
\begin{enumerate}
\item For contradiction suppose that $\tau_e(A,F,\gamma) < t_e$. Since $\tau_e$ is continuous and it is possible to locally decrease work at fixed reliability time, there exists a design $(A',F',\gamma')$ with $W(A',F',\gamma')<W(A,F,\gamma)$ that is $(t_r,t_e)$-feasible contradicting the optimality of the design $(A,F,\gamma)$.
\item For contradiction suppose that $\tau_r(A,F,\gamma) > t_r$. Since the design $(A,F,\gamma)$ is not locally trapped and $\tau_r$ is continuous, there exists a design $(A_0,F_0,\gamma_0)$ requiring work $W(A_0,F_0,\gamma_0)=W(A,F,\gamma)$, erasing time $\tau_e(A_0,F_0,\gamma_0) < \tau_e(A,F,\gamma)\leq t_e$ and maintaining reliability time $\tau_r(A_0,F_0,\gamma_0)\geq t_r$. Thus the design $(A_0,F_0,\gamma_0)$ is $(t_r,t_e)$-optimal contradicting Claim~\ref{lem:improvee}.\ref{erasing_tight} that the optimal design saturates the bound on the erasing time constraint.
\end{enumerate}
\end{proof}

\item\begin{observation}\label{observ:decreasing_work}
The erasing time of trapped designs is a strictly decreasing function of work for our family of protocols (Example~\ref{ex:control_potential}). In other words, if $(A_0,F_0,\gamma_0)$ and $(A^*,F^*,\gamma^*)$ are trapped designs with $ W(A^*,F^*,\gamma^*) > W(A _0,F_0,\gamma_0)$ then $\tau_e(A^*,F^*,\gamma^*) < \tau_e(A_0,F_0,\gamma_0)$. 
\end{observation}

\begin{proof}
Since $W$ is a continuous increasing function of $A$ (Observation\,~\ref{observ:work_increasing_height}), one can choose $A' > A^*$ such that $W(A',F_0,\gamma_0)=W(A^*,F^*,\gamma^*)$. Noting that increasing well height at fixed $F$ and $\gamma$ decreases erasing time (Observation\,~\ref{observ:erasing_decreasing_height}), we get $\tau_e(A',F_0,\gamma_0) < \tau_e(A_0,F_0,\gamma_0)$ . Using the fact that $(A^*,F^*,\gamma^*)$ is a trapped design, we get $\tau_e(A^*,F^*,\gamma^*)\leq \tau_e(A',F_0,\gamma_0) < \tau_e(A_0,F_0,\gamma_0)$ establishing the claim.
\end{proof}

\item\textbf{Proof of Claim~\ref{claim2}
\begin{proof}
\begin{enumerate}
\item Since $\tau_r(A^*,F^*,\gamma^*)\geq t_r$ and $\tau_e(A^*,F^*,\gamma^*)=t_e$, the design $(A^*,F^*,\gamma^*)$ is $(t_r,t_e)$-feasible. Suppose that the design $(A^*,F^*,\gamma^*)$ is not $(t_r,t_e)$-optimal. Then there exists a $(t_r,t_e)$-feasible design $(A',F',\gamma')$ such that $W(A',F',\gamma') < W(A^*,F^*,\gamma^*)$. Let $(A_0,F_0,\gamma_0)$ be a trapped design with $W(A_0,F_0,\gamma_0)=W(A',F',\gamma')<W(A^*,F^*,\gamma^*)$. Then $\tau_e(A_0,F_0,\gamma_0) > \tau_e(A^*,F^*,\gamma^*)$ since the erasing time of trapped bits is a strictly decreasing function of work. Using the fact that $(A_0,F_0,\gamma_0)$ is a trapped design, we get $\tau_e(A',F',\gamma')\geq \tau_e(A_0,F_0,\gamma_0) > \tau_e(A^*,F^*,\gamma^*)=t_e$, a contradiction since $(A',F',\gamma')$ is a $(t_r,t_e)$-feasible design.
\item Immediate from Claim~\ref{claim2}.~\ref{same_optimal_bit}.
\item For contradiction suppose that the requirement $(t_r,t_e)$ is unsaturated. Then by Claim~\ref{lem:improvee}.~\ref{erasing_tight} there exists a $(t_r,t_e)$-optimal design $(A_0,F_0,\gamma_0)$ such that $\tau_r(A_0,F_0,\gamma_0) > t_r$ and $\tau_e(A_0,F_0,\gamma_0)=\tau_e(A^*,F^*,\gamma^*)=t_e$. Since locally trapped designs are uniquely trapped, using Claim~\ref{lem:improvee}.~\ref{reliability_tight}, we get that the design $(A_0,F_0,\gamma_0)$ must be uniquely trapped. Noting that uniquely trapped bits are trapped and using the fact that the erasing time of trapped designs is a strictly decreasing function of work, we get $W(A_0,F_0,\gamma_0)=W(A^*,F^*,\gamma^*)$. This implies that $(A_0,F_0,\gamma_0)=(A^*,F^*,\gamma^*)$, a contradiction since $\tau_r(A_0,F_0,\gamma_0) > t_r \geq \tau_r(A^*,F^*,\gamma^*)$. 
\end{enumerate}
\end{proof}
}
\end{enumerate}

\subsection{Alternative proof via KKT conditions}\label{subsec:KKT_conditions}

KKT conditions form the foundation of optimization problems~\cite{nocedal2006numerical,rapcsak2013smooth}. In order to study the KKT conditions, we consider the optimization problem of finding the design with the lowest work that is $(t_r,t_e)$-feasible.

\begin{problem}\label{prob:bit_optimal}
\begin{align*}
\left(A^*,F^*,\gamma^*\right) = \displaystyle\arg\inf_{A,F,\gamma} W\left(A,F\right) \\
t_r - \tau_r\left(A^*,\gamma^*\right)\leq 0 \\
\tau_e\left(A^*,F^*,\gamma^*\right) - t_e \leq 0
\end{align*}
\end{problem}

In order to state the KKT conditions, we will need the notion of a regular point. The following definition will make this precise.

\begin{definition}[\textbf{Regular point}]\label{def:regular}
Let $Sat\left(x^*\right)$ denote the set of gradients of the constraints that are saturated at the point $x^*$. Then $x^*$ is regular iff $Sat\left(x^*\right)$ does not form a linearly dependent set. 
\end{definition}

\begin{theorem}[\textbf{KKT conditions}]\label{theorem:KKT_conditions}
Let $(A^*,F^*,\gamma^*)$ be a local optimum of ~\ref{prob:bit_optimal} and a regular point. Then by~\cite[$\left(12.1\right)$, pp.~95]{nocedal2006numerical}, there exists $\lambda_1^*,\lambda_2^*\in\R_{\geq 0}$ such that 
\begin{enumerate}
\item $\nabla W\left(A^*,F^*,\gamma^*\right) - \lambda_1^*\nabla\tau_r\left(A^*,\gamma^*\right) + \lambda_2^*\nabla\tau_e\left(A^*,F^*,\gamma^*\right) = 0$.
\item $\lambda_1^*\left(t_r- \tau_r\left(A^*,\gamma^*\right)\right)=0$ \text{and} $\lambda_2^*\left(\tau_e\left(A^*,F^*,\gamma^*\right) -t_e\right)=0$.
\end{enumerate} 
\end{theorem}

Given this powerful theorem~\ref{theorem:KKT_conditions}, we are now ready to prove the the same result that we obtained earlier but with the machinery of KKT conditions. 

\begin{lemma}
Let us assume that it is always possible to locally decrease work at fixed reliability time. Let $\left(A^*, F^*,\gamma^*\right)$ be a local optimum of ~\ref{prob:bit_optimal}. Then either 
\begin{enumerate}
\item The design $\left(A^*, F^*,\gamma^*\right)$ saturates the bound on both constraints i.e. $\tau_r\left(A^*,\gamma^*\right)=t_r$ and $\tau_e\left(A^*,F^*,\gamma^*\right)=t_e$ or
\item The design $\left(A^*, F^*,\gamma^*\right)$ saturates the bound on the erasing time constraint i.e. $\tau_e\left(A^*,F^*,\gamma^*\right)=t_e$ but does not saturate the bound on the reliability time constraint i.e. $\tau_r\left(A^*,\gamma^*\right)>t_r$ and is locally trapped.
\end{enumerate}
\end{lemma}

\begin{proof}
Consider an optimal design $\left(A^*, F^*,\gamma^*\right)$ such that either it does not saturate the bound on the reliability time constraint i.e. $\tau_r\left(A^*,\gamma^*\right)>t_r$ or it does not saturate the bound on the erasing time constraint i.e. $\tau_e\left(A^*,F^*,\gamma^*\right) < t_e$. Then we have the following cases:
\begin{bullets}
\item Case 1: The design $\left(A^*,F^*,\gamma^*\right)$ saturates the bound on the erasing time constraint, but does not saturate the bound on the reliability time constraint i.e. $\tau_r\left(A^*,\gamma^*\right)>t_r$ and $\tau_e\left(A^*,F^*,\gamma^*\right) = t_e$. This implies that $\lambda_1^*=0$. Since only one constraint is active, $(A^*, F^*,\gamma^*)$ is a regular point. Hence, by Theorem~\ref{theorem:KKT_conditions} on KKT conditions, there exists $\lambda_2^* >0$ such that $\nabla W\left(A^*,F^*,\gamma^*\right) + \lambda_2^*\nabla t_e\left(A^*,F^*,\gamma^*\right)=0$. This implies that $(A^*, F^*,\gamma^*)$ is a stationary point of erasing time in the level set of it's work $W(A^*, F^*,\gamma^*)$. The fact that this stationary point is actually a local minimum follows from Claim~\ref{lem:improvee}.~\ref{reliability_tight}.
\item Case 2: The design $\left(A^*,F^*,\gamma^*\right)$ saturates the bound on the reliability time constraint, but does not saturate the bound on the erasing time constraint i.e. $\tau_r\left(A^*,\gamma^*\right)=t_r$ and $\tau_e\left(A^*,F^*,\gamma^*\right) < t_e$. This implies that $\lambda_2^*=0$. Since only one constraint is active, $(A^*, F^*,\gamma^*)$ is a regular point. Hence, by Theorem~\ref{theorem:KKT_conditions} on KKT conditions, there exists $\lambda_1^* >0$ such that $\nabla W\left(A^*,F^*,\gamma^*\right) = \lambda_1^*\nabla t_r\left(A^*,\gamma^*\right)$, a contradiction since $\frac{\partial W}{\partial F}\neq 0$ but $\frac{\partial\tau_r}{\partial F}=0$.
\item Case 3: The design $\left(A^*,F^*,\gamma^*\right)$ does not saturate the bound on either constraints i.e. $\tau_r\left(A^*,\gamma^*\right)>t_r$ and $\tau_e\left(A^*,F^*,\gamma^*\right) < t_e$. Since no constraint is active we have $\nabla W\left(A^*,F^*,\gamma^*\right)=0$, which is not possible.
\end{bullets}
\end{proof} 

\end{document}